\documentclass[12pt,A4]{article}

\usepackage{acro}
\usepackage{amsmath}
\usepackage{amssymb}
\usepackage{amsthm} 
\usepackage{bm}
\usepackage{hyperref}
\usepackage{cleveref}
\usepackage{fullpage} 
\usepackage{natbib}
\usepackage{xcolor}
\usepackage{mathtools}
\usepackage{enumerate}  
\usepackage{chngcntr} %
\usepackage{placeins}
\usepackage{setspace} %

\newtheorem{definition}{Definition}
\newtheorem{theorem}{Theorem}
\newtheorem{example}{Example}
\newtheorem{remark}{Remark}
\newtheorem{lemma}{Lemma}
\newtheorem{proposition}{Proposition}

\DeclareMathOperator*{\argmin}{arg\,min}

\DeclareAcronym{GP}{
    short = {GP},
    long  = {Gaussian process},
    long-plural = {es}
}
\DeclareAcronym{GRE}{
    short = {GRE},
    long  = {Gauss--Richardson Extrapolation}
}
\DeclareAcronym{HPC}{
    short = {HPC},
    long  = {high performance computing}
}
\DeclareAcronym{SPRE}{
    short = {SPRE},
    long  = {Sparse Probabilistic Richardson Extrapolation}
}
\DeclareAcronym{PDE}{
    short = {PDE},
    long  = {partial differential equation}
}
\DeclareAcronym{PRE}{
    short = {PRE},
    long  = {Probabilistic Richardson Extrapolation}
}
\DeclareAcronym{RKHS}{
    short = {RKHS},
    long  = {reproducing kernel Hilbert space}
}
\DeclareAcronym{MRE}{
    short = {MRE},
    long  = {multivariate Richardson extrapolation}
}
\DeclareAcronym{LOOCV}{
    short = {LOOCV},
    long  = {leave-one-out cross-validation}
}
\DeclareAcronym{HIFI}{
    short = {HiFi},
    long  = {high fidelity}
}
\DeclareAcronym{LOFI}{
    short = {LoFi},
    long  = {low fidelity}
}
\DeclareAcronym{MFM}{
    short = {MFM},
    long  = {multi-fidelity model}
}
\DeclareAcronym{MuJoCo}{
    short = {MuJoCo},
    long  = {Multi-Joint dynamics with Contact}
}

\newif\ifshowalttext
\showalttextfalse
\newcommand{\alttext}[1]{%
\ifshowalttext
{\color{green!50!black}Alt text: #1}%
\fi
}

\usepackage[margin=0.8in]{geometry}

\usepackage{changepage}   %

\begin{document}

\title{Sparse Probabilistic Richardson Extrapolation}
\author{Chris. J. Oates$^1$\footnote{Correspondence should be sent to \texttt{chris.oates@ncl.ac.uk}.} , Richard Howey$^2$, Toni Karvonen$^{3}$ \\ \vspace{-5pt}
\footnotesize $^1$ School of Mathematics, Statistics and Physics, Newcastle University, UK \\ \vspace{-5pt}
\footnotesize $^2$ Research Software Engineering, Newcastle University, UK \\ \vspace{-5pt}
\footnotesize $^3$ School of Engineering Sciences, Lappeenranta--Lahti University of Technology LUT, FI 
}

\maketitle

\begin{abstract}
    Almost every numerical task can be cast as extrapolation with respect to the fidelity or tolerance parameters of a consistent numerical method.
    This perspective enables probabilistic uncertainty quantification and optimal experimental design functionality to be deployed, and also unlocks the potential for the convergence of numerical methods to be accelerated.
    Recent work established \emph{Probabilistic Richardson Extrapolation} as a proof-of-concept, demonstrating how parallel multi-fidelity simulation can be used to accelerate simulation from a whole-heart model.
    However, the number of simulations was required to increase super-exponentially in $d$, the number of tolerance parameters appearing in the numerical method.
    This paper develops a refined notion of \emph{extrapolation dimension}, drastically reducing this simulation requirement when multiple tolerance parameters feature in the numerical method.
    Sparsity-exploiting methodology is developed that is simultaneously simpler and more powerful compared to earlier work, and this is accompanied by sharp theoretical guarantees and substantial empirical support.

    \medskip

\end{abstract}

\section{Introduction}

Mathematical models are routinely used to represent mechanisms hypothesised to govern phenomena of interest. 
Inference for the parameters of these models demands a comparison of simulations from the model against a real-world dataset. 
The computational challenge is two-fold; to produce simulations from the model, and to search across the parameter space for values that provide an acceptable fit to the dataset.
As a motivating example, the whole-heart model of \citet{strocchi2023cell} involves dozens of parameters, and each simulation requires hundreds of CPU hours on \ac{HPC} equipment.
To drive progress in this, and other domains where sophisticated mathematical models are used, there is an urgent need for methodology to facilitate simulation and inference at reduced computational cost.
The focus of this paper is on the first challenge; accurate simulation of mathematical models with a limited compute budget.

Almost every numerical task can be cast as extrapolation with respect to the fidelity or tolerance parameters of a consistent numerical method.
Abstractly, we enumerate all of the \emph{discretisation parameters} involved in approximate simulation from a mathematical model as $\mathbf{x}$, such that each component of $\mathbf{x}$ controls the error due to a particular aspect of discretisation; for example, in a physical simulation $x_1$ could be a time step size, $x_2$ could be the width of a spatial mesh, and $x_3$ could be an error tolerance for an adaptive numerical method.
The simulator output is denoted $f(\mathbf{x})$, and we assume that the simulator is consistent, meaning that $f$ is continuous at $\mathbf{0}$ and exact simulation of the mathematical model corresponds to $f(\mathbf{0})$.
The computational cost of simulating $f(\mathbf{x})$ will be denoted $c(\mathbf{x})$, with $\lim_{\mathbf{x} \rightarrow \mathbf{0}} c(\mathbf{x}) = \infty$ being typical.

The \emph{extrapolation} task that we consider in this work is to produce an accurate approximation to $f(\mathbf{0})$, based on a dataset $\{ f(\mathbf{x}_i) \}_{i=1}^n$, such that the computational cost associated with running the simulations in the dataset remains within a prescribed budget.
In the case of a single scalar discretisation parameter, so-called \emph{extrapolation methods} such as Richardson extrapolation \citep{richardson1911ix} (a.k.a. polynomial extrapolation) are a classical yet powerful tool to reduce the cost of simulation, enabling the convergence of foundational (deterministic) numerical methods, such as quadratures and finite difference methods, to be accelerated.
A textbook treatment is provided in \citet{brezinski2013extrapolation}.

However, these tools were not developed with modern computer codes in mind, where multiple continua are discretised and convergence orders are not easily analysed.  
A statistical perspective on multivariate extrapolation offers not just an opportunity to accelerate simulation from complex mathematical models, but also the potential to use established statistical techniques to select a parsimonious regression model, the data-driven learning of convergence orders, integrating extrapolation uncertainty into decision-making, and opening up the potential for experimental design to be used to select the configurations $\{ \mathbf{x}_i \}_{i=1}^n$ at which simulation is to be performed.
More broadly, new statistical insights into extrapolation could be leveraged in wide-ranging applications, such as extrapolation of treatment effects \citep{fauvel2025comparison}, extrapolation of survival curves \citep{kearns2021comparing}, and causal inference \citep{kong2025towards}.

Simulation at multiple fidelities is also well-studied in statistics.
Surrogate modelling \citep{sacks1989design} and related data-driven techniques such as reduced order modelling \citep{lucia2004reduced} are popular in settings where the simulator involves additional scientific parameters $\bm{\theta}$, and the parameter dependence of the simulator output is of interest.
In \emph{multi-fidelity modelling} (\acs{MFM})\acuse{MFM} one constructs a joint statistical model for a hierarchy of simulators $\bm{\theta} \mapsto f_{\bm{\theta}}(\mathbf{x}_i)$ of different fidelities $\{\mathbf{x}_i\}$, leveraging plentiful low-fidelity simulations to make predictions regarding the highest-fidelity simulator \citep{peherstorfer2018survey}.
Seminal work in this direction includes the autoregressive co-kriging model of \citet{kennedy2000predicting} and its recursive formulation by \citet{legratiet2014recursive}, with more recent developments including nonlinear relationships between fidelities \citep{perdikaris2017nonlinear} and graphical and conglomerate model structures \citep{ji2024graphical,ji2022conglomerate}.
There is also active interest in cost-constrained experimental design in \ac{MFM} \citep[e.g.][]{dixon2025optimally}.
These approaches typically treat the highest-fidelity simulator as the target, whereas in reality the continuum limit $f_{\bm{\theta}}(\mathbf{0})$ is of principal interest.
In specific applications, extrapolation of a \ac{GP} model to the continuum limit has been considered, for example for finite element simulation with different mesh densities \citep{tuo2014surrogate}, to simulate ice sheets in Antarctica \citep{thodoroff2023multi}, and to simulate the evolution of the quark-gluon plasma following the Big Bang \citep{ji2022conglomerate}.
However, to the best of our knowledge, convergence \emph{acceleration} via extrapolation has not been considered in this context.

A first step towards statistical extrapolation methods was made in \citet{oates2024probabilistic}, who modelled the data-generating function $f$ using a \emph{numerical analysis-informed} \ac{GP}, leading to a method called \emph{Gauss--Richardson Extrapolation} (\acs{GRE}); we provide further details in \Cref{sec: GRE}.
The effectiveness of \ac{GRE} was demonstrated using the aforementioned whole-heart model of \citet{strocchi2023cell}, where extrapolation of clinical quantities of interest (e.g. the maximum volume of each chamber during a heartbeat) with respect to discretisation parameters $\mathbf{x} = (x_1,x_2)$ of the simulator (in that case, $x_1$ was a time step size and $x_2$ was a spatial mesh width) was considered.
In particular, the authors demonstrated that extrapolation of low-fidelity data can be more accurate than running one high-fidelity simulation at comparable computational cost, a finding that has the potential to change how cardiac simulations are performed.
To appreciate the practical impact of extrapolation in this context, consider that cardiac digital twins -- patient-specific models of the heart calibrated to clinical data -- are envisioned as a means to personalise the diagnosis and treatment of heart disease \citep{corral2020digital,niederer2021scaling}.
Calibration, sensitivity analysis and uncertainty quantification each require a large number of simulations from the model, and when each simulation consumes hundreds of CPU hours this is only feasible if simulations are performed at reduced fidelity; for example, the global sensitivity analysis of \citet{strocchi2023cell} relied on \ac{GP} emulators trained on a limited number of simulations at a fixed fidelity.
Any discretisation error is then propagated, unquantified, into the conclusions that are drawn.
Extrapolation offers a means to remove much of this error at a computational cost comparable to that of low-fidelity simulation, while \emph{statistical} extrapolation additionally quantifies the error that remains.
However, whole-heart simulators involve several discretisation parameters beyond the time step and mesh width considered in \citet{oates2024probabilistic}, such as tolerances for the iterative linear and nonlinear solvers that are used at each time step.
Existing extrapolation methods are therefore either restricted to an impractically small number of discretisation parameters (\ac{GRE}), or require a numerical analysis that is intractable for such complex simulators (\ac{MRE}).
A detailed account of the cardiac application, including visualisation of the simulator output, can be found in \citet{oates2024probabilistic}.

Despite this initial success, there is a major limitation to \ac{GRE}; its theoretical justification is only meaningful in settings where the dimension $d = \mathrm{dim}(\mathbf{x})$ is extremely small.
Indeed, the established error bounds for \ac{GRE} \citep[][Theorems 2 and 4]{oates2024probabilistic} hold only when a minimum number of data $n_{\min}$ are provided, and this number can be lower-bounded by $(2^d d!)^d$, which is super-exponential in $d$.
For $d = 1$, this lower bound is $n_{\min} \geq 2$, which is best-possible since for extrapolation at least two data points are required.
For $d = 2$, this lower bound is $n_{\min} \geq 64$, which is large but could still be practical.
However, if we consider $d = 4$ discretisation parameters, \ac{GRE} requires a minimum of $2 \times 10^{10}$ simulations to be performed.
Given that even a single simulation of a computational cardiac model can require hundreds of CPU hours on \ac{HPC}, this data requirement is clearly impractical.
We are led to urgently ask the question, \emph{is better dimension dependence possible in statistical extrapolation?}

This paper provides a positive answer in two different respects:
First, we present an alternative approach called \emph{Sparse Probabilistic Richardson Extrapolation} (\acs{SPRE}).
As with \ac{GRE}, our \ac{SPRE} approach is based on \acp{GP}, but here we leverage a new notion of \emph{extrapolation dimension} which enables us to drastically reduce the number of simulations that need to be performed.
Indeed, we establish convergence acceleration guarantees for \ac{SPRE} for which the data requirement is typically polynomial in $d$, as opposed to super-exponential.
Second, through several detailed case studies we identify low extrapolation dimension as a general phenomenon in real-world applications, in the sense that all of the simulators we consider exhibit some form of low effective dimension (or \emph{sparsity}) which can be directly exploited in \ac{SPRE}.
Together, these contributions establish sparsity-exploiting extrapolation methods as a novel frontier for acceleration of complex simulators.

The remainder of the paper is organised as follows:
\Cref{sec: background} sets notation, introduces the concept of extrapolation dimension, and recalls existing extrapolation methods which serve as baselines for \ac{SPRE}.
Our novel methodology is contained in \Cref{sec: methods}, with several detailed case-studies contained in \Cref{sec: experiments}.
A discussion concludes in \Cref{sec: discussion}.
Code to reproduce our experiments can be downloaded from \href{https://github.com/NewcastleRSE/sparse-probabilistic-richardson-extrapolation/}{GitHub}.

\section{Background}
\label{sec: background}

Our set-up is specified in \Cref{subsec: setup}.
As baselines, we review multivariate Richardson extrapolation in \Cref{subsec: MRE} and recall the \ac{GRE} method of \citet{oates2024probabilistic} in \Cref{sec: GRE}.
 
\subsection{Set-Up}
\label{subsec: setup}

\Cref{subsec: notation} introduces notation, \Cref{subsec: standing assum} clarifies our standing assumptions on the simulator, and the notion of extrapolation dimension is introduced in \Cref{subsec: extrap sparse}.

\subsubsection{Notation}
\label{subsec: notation}

Given $\bm{\alpha} \in \mathbb{N}_0^d$ , the multi-index notation $x^{\bm{\alpha}} = (x^{\alpha_1} , \dots , x^{\alpha_d}) \in \mathbb{R}^d$ for $x \in \mathbb{R}$ and $\mathbf{x}^{\bm{\alpha}} = x_1^{\alpha_1} \cdots x_d^{\alpha_d}$ for $\mathbf{x} \in \mathbb{R}^d$ will be used.
In particular, for a multi-index $\bm{\alpha} \in \mathbb{N}_0^d$ let $|\bm{\alpha}| = |\alpha_1| + \dots + |\alpha_d|$.
For $\mathbf{0} \in A \subset \mathbb{N}_0^d$, let $\mathcal{P}_A$ be the polynomial space of real-valued functions on $\mathbb{R}^d$ spanned by the $\mathbf{x}^{\bm{\alpha}}$ with $\bm{\alpha} \in A$, and let $\mathrm{Lead}(A)$ denote indices in $A \setminus \{\mathbf{0}\}$ corresponding to the \emph{leading order} monomials; i.e. for which $|\bm{\alpha}|$ is minimal.
For $\mathbf{0} \in A = \{\bm{\alpha}_i \}_{i=1}^m \subset \mathbb{N}_0^d$ and $X_n = \{\mathbf{x}_i\}_{i=1}^n \subset \mathbb{R}^d$, define the Vandermonde matrix
\begin{align*}
\mathbf{V}_A(X_n) = \left[ \begin{array}{ccc} \mathbf{x}_1^{\bm{\alpha}_1} & \dots & \mathbf{x}_1^{\bm{\alpha}_m} \\ \vdots & & \vdots \\ \mathbf{x}_n^{\bm{\alpha}_1} & \dots & \mathbf{x}_n^{\bm{\alpha}_m}
  \end{array} \right] ,
\end{align*}
noting the convention with which the elements of $A$ and $X_n$ are ordered.
In particular, we will assume that $\bm{\alpha}_1 = \mathbf{0}$ throughout.
The set $X_n$ is called $\mathcal{P}_A$-\emph{unisolvent} when $\mathbf{V}_A(X_n)$ has full rank, which in turn implies the existence and uniqueness in $\mathcal{P}_A$ of polynomial interpolants to data given at $X_n$ \citep[][Definition 2.6]{wendland2004scattered}.
This paper also employs the shorthand $[\mathbf{a},\mathbf{b}] = [a_1,b_1] \times \dots \times [a_d,b_d]$ and uses $\mathbf{e}_1 = [1,0,\dots,0]^\top$ for the first canonical basis vector of $\mathbb{R}^d$.

\subsubsection{Standing Assumptions on the Simulator}
\label{subsec: standing assum}

For this paper we focus on scalar-valued simulator output $f : \mathcal{X} \rightarrow \mathbb{R}$, but generalisation to multivariate and infinite-dimensional simulator output can be considered in an analogous manner to \citet[][Section 2.9]{oates2024probabilistic}.
It is assumed that $\mathbf{0} \in \mathcal{X} \subset [0,\infty)^d$ and that $\mathbf{0}$ is an accumulation point of $\mathcal{X}$, so that in principle it is possible to run simulations to arbitrary precision for a sufficiently high computational budget.

To formulate our main results we will assume that there exist an index set $\mathbf{0} \in A \subset \mathbb{N}_0^d$ and coefficients $\beta_{\bm{\alpha}}$ for which the residual
\begin{align}
R(\mathbf{x}) = f(\mathbf{x}) - \sum_{ \bm{\alpha} \in A } \beta_{\bm{\alpha}} \mathbf{x}^{\bm{\alpha}}  \label{eq: error}
\end{align}
satisfies $R(\mathbf{x}) = o(f(\mathbf{x})-f(\mathbf{0}))$ as $\mathbf{x} \rightarrow \mathbf{0}$.
The set $A$ will be referred to as the \emph{index} set; note that $A$ must contain more than just the element $\mathbf{0}$ for this property to hold. 
It is not assumed that the coefficients $\beta_{\bm{\alpha}}$ are known at the outset.
This assumption implies that the convergence rate of $f(\mathbf{x})$ to $f(\mathbf{0})$ is sub-optimal and thus can in principle be accelerated.
That is, a rate-optimal numerical method for the numerical task at hand would \emph{not} satisfy \eqref{eq: error}.
So why is extrapolation of interest?
High-order numerical methods are usually less stable compared to tried-and-tested low-order numerical methods; it is the reason that the trapezoidal rule is routinely used to numerically integrate smooth functions, and it is the reason that piecewise linear finite elements are routinely used to discretise \acp{PDE}, even if the (weak) solution to the \ac{PDE} admits higher than first order (weak) derivatives.
Further, since error analysis is usually intractable for complex simulators, the designer of a simulator will typically not risk using a high-order numerical method whose assumptions cannot be established.
Taken together, we can expect \eqref{eq: error} to hold in a broad range of applications where simulators are used.

\subsubsection{Extrapolation Dimension}
\label{subsec: extrap sparse}

The general set-up in \eqref{eq: error} encompasses a variety of expansions, with the cardinality of the index set $A$ being related to the data requirement for the extrapolation task:

\begin{definition}[Extrapolation dimension]
    The \emph{extrapolation dimension} $d_{\mathrm{ext}}(f)$ of a simulator $f$ is the minimal value of $\mathrm{dim}(A)$, i.e. the minimal cardinality of the index set $\mathbf{0} \in A \subset \mathbb{N}_0^d$ for which $f$ has a non-trivial representation \eqref{eq: error}.
\end{definition}

\noindent Put in practical terms, we will see that $d_{\mathrm{ext}}(f)$ is the minimum number of simulations required to accelerate convergence to $f(\mathbf{0})$; i.e. it is the number of coefficients $\beta_{\bm{\alpha}}$ in an expansion of the form \eqref{eq: error} that would need to be estimated.
To illustrate the definition with a toy example, if $f(\mathbf{x}) = f(\mathbf{0}) + \bm{\beta}^\top \mathbf{x}$ is an affine function with $\bm{\beta} \neq \mathbf{0}$, then $d_{\mathrm{ext}}(f) = 1 + \|\bm{\beta}\|_0$, and indeed a minimum number $n_{\min} = 1 + \|\bm{\beta}\|_0$ of distinct data $f(\mathbf{x}_i)$ are required to recover $f(\mathbf{0})$ under this affine model.
(Here $\|\bm{\beta}\|_0$ denote the number of non-zero entries in the vector $\bm{\beta} \in \mathbb{R}^d$.)

A central argument to this paper, which we will substantiate with empirical support, is that \emph{if} a simulator admits an expansion of the form \eqref{eq: error}, then the cardinality of the index set $A$ is typically (relatively) small.
For example, if we consider a continuum $\mathcal{X} \subset [0,\infty)^d$ and a simulator $f$ that is $s$ times continuously differentiable at $\mathbf{0}$, for which not all derivatives vanish at $\mathbf{0}$, then from Taylor's theorem there is a non-trivial expansion \eqref{eq: error} with index set $A = \{\bm{\alpha} : |\bm{\alpha}| \leq s , \; \bm{\alpha} \in \mathbb{N}_0^d \}$ for which $R(\mathbf{x}) = o(\sum_{|\bm{\alpha}| = s} \mathbf{x}^{\bm{\alpha}} )$, and the dimension dependence of $\mathrm{dim}(A) \asymp d^s$ is polynomial.
Informally, we say that the extrapolation task $f$ exhibits \emph{sparsity}, in the sense that only a (relatively) small number of non-zero regression coefficients $\beta_{\bm{\alpha}}$ need to be considered\footnote{The terminology is consistent with variable selection in linear regression, where \emph{sparsity} refers to the number of non-zero regression coefficients, i.e. $\|\bm{\beta}\|_0$ in the notation of the previous paragraph.}.
Further discussion on when extrapolation sparsity can be expected to hold can be found in \Cref{app: discuss sparsity}.

However, many important numerical methods do not have tolerance parameters defined on a continuum $\mathcal{X}$; for example, the accuracy of certain cubature methods depends on the number of cubature nodes that are used.
This precludes the use of Taylor expansions in general, but nevertheless an expansion of the form \eqref{eq: error} can still hold.
A simple example illustrating the case of a discrete set $\mathcal{X}$ is presented in \Cref{subsec: cubature}.

\begin{example}[Extrapolation sparsity for a simple cubature method]
\label{subsec: cubature}
The Riemann midpoint approximation to the integral $f(\mathbf{0}) = \int_{[\mathbf{0},\mathbf{1}]} g(\mathbf{t}) \; \mathrm{d}\mathbf{t} $ is defined as $f(\mathbf{x}) = \sum_i \mathrm{vol}(C_i) g(\mathbf{t}_i)$ where $\mathbf{t}_i$ are the midpoints of the cells $C_i \subset [\mathbf{0},\mathbf{1}]$ formed by a Cartesian grid.
The side lengths of each cell are $\mathbf{x} = (x_1,\dots,x_d) \in \mathcal{X} := \{\mathbf{0}\} \cup (\{\frac{1}{n} : n \in \mathbb{N}\} )^d$, and $\mathrm{vol}(C_i) = \prod_{i=1}^d x_i$ is the volume of each cell.
If $g \in C^{2s+2}([\mathbf{0},\mathbf{1}])$, then 
\begin{align}
    f(\mathbf{x}) = \sum_{0 \leq |\bm{\alpha}| \leq s} \beta_{\bm{\alpha}} \mathbf{x}^{2\bm{\alpha}} + O(\|\mathbf{x}\|^{2s+2}) \label{eq: Riemann Taylor}
\end{align}
for some coefficients $\beta_{\bm{\alpha}}$ \citep[][Theorem 1.1 in Chapter 3]{liem1995splitting}.
\end{example}

Until now the questions of whether -- or to what extent -- extrapolation sparsity holds for complex simulators has not been studied.
Thus extrapolation sparsity must instead be treated as a hypothesis to be assessed using data.
Further, if one was able to establish extrapolation sparsity for a complex simulator in a data-driven way, then it is natural to ask whether this sparsity can be exploited within the extrapolation task.
This sets the scene for \Cref{sec: methods}, but first we review multivariate Richardson extrapolation (\Cref{subsec: MRE}) and the \ac{GRE} method of \citet{oates2024probabilistic} (\Cref{sec: GRE}), which will each serve as baselines for our novel \ac{SPRE} method.

\subsection{Multivariate Richardson Extrapolation}
\label{subsec: MRE}

Here we recall a multivariate generalisation\footnote{Some authors alternatively refer to the case where $x_1$ is varied and $x_2 , \dots , x_d$ are scaled in proportion to $x_1$ as the `multivariate version' of Richardson extrapolation \citep[][p45]{liem1995splitting}, calling what we present here the \emph{splitting extrapolation method} \citep{qun1983splitting,liem1995splitting}.} of the classical extrapolation technique of \citet{richardson1911ix}.
The idea is to exploit \eqref{eq: error}, approximating $f$ using an interpolant from $\mathcal{P}_A$, and then extrapolating this approximating polynomial to $\mathbf{0}$.
Although not a statistical extrapolation method, one of the main motivations for considering \ac{MRE} here is that the number $n_{\min}$ of data needed to implement it is $d_{\mathrm{ext}}(f)$, which is the minimal possible data requirement.
As such, \ac{MRE} serves as a (non-statistical) sparsity-exploiting benchmark.

As with classical Richardson extrapolation, \ac{MRE} is used to accelerate convergence and obtain a more accurate approximation to the mathematical object of interest.
To illustrate the convergence acceleration functionality of \ac{MRE} we present a general convergence result in \Cref{thm: SEM}.
For simplicity, here we suppose $\mathcal{X}$ is a continuum, rather than a discrete set, but we will allow for the possibility that $\mathcal{X}$ is a discrete set when we present our new methods in \Cref{sec: methods}.

\begin{theorem}[Convergence acceleration for \ac{MRE}]
\label{thm: SEM}
Let $\mathcal{X} = [\mathbf{0},\bm{1}] \subset \mathbb{R}^d$ and let $X_n = \{\mathbf{x}_i\}_{i=1}^n \subset \mathcal{X}$ be $\mathcal{P}_A$-unisolvent where $n = \mathrm{dim}(A)$.
Let $\mathcal{X}_h = [\mathbf{0}, h \mathbf{1}]$ and $X_n^h = \{ h \mathbf{x}_i\}_{i=1}^n$ for $h \in (0,1]$. 
Let $f_n^h \in \mathcal{P}_A$ interpolate $f$ on $X_n^h$.
Then there exists a constant $C$, depending only on $X_n$, such that
\begin{align}
\underbrace{ |  f_n^h(\mathbf{0}) - f(\mathbf{0}) | }_{\text{\normalfont extrapolation error}} \leq \underbrace{ \phantom{|} C }_{\text{\normalfont constant in $h$}} \; \; \underbrace{ \sup_{\mathbf{x} \in \mathcal{X}_h} | R(\mathbf{x}) | }_{\text{\normalfont acceleration}} ,  \label{eq: bound for MRE}
\end{align}
meaning that $f_n^h(\mathbf{0})$ converges faster than the original simulator output $f(\mathbf{x})$, for any $\mathbf{x} \in X_n^h$, as $h \rightarrow 0$ under the standing assumptions of \Cref{subsec: standing assum}.
\end{theorem}

\noindent A self-contained proof of \Cref{thm: SEM}, and further background on the \ac{MRE}, can be found in \Cref{app: SEM}.
To clarify the content of \Cref{thm: SEM}, here the number $n$ of data is \emph{fixed}\footnote{The total number of simulator evaluations, $n$, will typically need to be small to control total computational cost, so the large $n$ limit is not usually considered.} and we are examining the convergence of approximations based on a dataset $\{f(h \mathbf{x}_i)\}_{i=1}^n$ as the simulator inputs $\{h \mathbf{x}_i\}_{i=1}^n$ are collectively drawn closer to $\mathbf{0}$, via $h \rightarrow 0$.
This result is sharp, in the sense that no further acceleration is possible under the limited assumptions in \eqref{eq: error}.
\ac{MRE} is straight-forward to implement, and in practice the main effort goes into the mathematical analysis required to select the index set $A$ and to establish control on the residual $R(\mathbf{x})$.
However, \ac{MRE} does not unlock the additional functionalities of a \emph{statistical} extrapolation method\footnote{In practice heuristics are used in combination with Richardson extrapolation, such as the \emph{grid convergence index} developed in \citet{roache1998verification} which is used in finite element applications \citep{kwasniewski2013application}.} (inc. uncertainty quantification, model criticism, and experimental design).

\subsection{Gauss--Richardson Extrapolation}
\label{sec: GRE}

To first statistical extrapolation method to offer provable convergence acceleration was the \ac{GRE} method of \citet{oates2024probabilistic} in which $f(x)$ was modelled as a \ac{GP}, denoted $\mathsf{f}(\mathbf{x})$ in the sequel.
Roughly speaking, the mean of $\mathsf{f}(\mathbf{x})$ was taken to be a constant, assigned an improper flat prior and integrated out, while the covariance between $\mathsf{f}(\mathbf{x})$ and $\mathsf{f}(\mathbf{x}')$ was taken to be of the form $\epsilon(\mathbf{x}) \epsilon(\mathbf{x}') k(\mathbf{x} , \mathbf{x}')$ where $\epsilon(\mathbf{x})$ was either obtained from numerical analysis, assuming access to a convergence rate $f(\mathbf{x}) - f(\mathbf{0}) = O(\epsilon(\mathbf{x}))$, or was learned, and $k$ was a symmetric positive definite kernel.
This choice ensures that samples $\mathsf{f}$ from the \ac{GP} also satisfy $\mathsf{f}(\mathbf{x}) - \mathsf{f}(\mathbf{0}) = O(\epsilon(\mathbf{x}))$, leading the authors to call this model a \emph{numerical analysis-informed} \ac{GP} \citep[such covariance functions appear also in][]{tuo2014surrogate,teymur2021black,bect2021quantification}.
Based on \eqref{eq: error}, a choice for $\epsilon(\mathbf{x})$ that encodes the correct convergence order could be
\begin{align}
\epsilon(\mathbf{x}) \asymp f(\mathbf{x}) - f(\mathbf{0}) = \sum_{\bm{\alpha} \in A \setminus \{\mathbf{0}\} } \beta_{\bm{\alpha}} \mathbf{x}^{\bm{\alpha}} \asymp \sum_{\bm{\alpha} \in \mathrm{Lead}(A)} \mathbf{x}^{\bm{\alpha}} ,   \label{eq: b}
\end{align}
where in the final expression unit coefficients are used, since the ideal coefficients $\beta_{\bm{\alpha}}$ are unknown in general.
Under this set-up, the authors were able to establish convergence acceleration by assuming that the normalised error $e(\mathbf{x}) := ( f(\mathbf{x}) - f(\mathbf{0}) ) / \epsilon(\mathbf{x})$ is regular enough to belong to the \ac{RKHS} $\mathcal{H}_k(\mathcal{X})$ of real-valued functions on $\mathcal{X}$ for which $k$ is the reproducing kernel \citep[see][for background]{berlinet2011reproducing}.

To state one such result, define the \emph{box fill distance} $\rho_{X_n,\mathcal{X}}$ as the supremum value of $\nu$ such that there is a box of the form $[\mathbf{x} , \mathbf{x} + \nu \mathbf{1}]$ contained in $\mathcal{X}$ for which $X_n \cap [\mathbf{x} , \mathbf{x} + \nu \mathbf{1}] = \emptyset$, and define the constants $\gamma_d$ using the induction $\gamma_d := 2d(1+\gamma_{d-1})$ with base case $\gamma_1 \:= 2$.
The following result can be found as Theorem 2 of \citet{oates2024probabilistic}:

\begin{theorem}[Convergence acceleration for GRE]
\label{thm: GRE}
Let $\mathcal{X} = [\mathbf{0},\bm{1}] \subset \mathbb{R}^d$ and $X_n = \{\mathbf{x}_i\}_{i=1}^n \subset \mathcal{X}$.
Let $\mathcal{X}_h = [\mathbf{0}, h \mathbf{1}]$ and $X_n^h = \{ h \mathbf{x}_i\}_{i=1}^n$ for $h \in (0,1]$.
Assume that $\epsilon$ is a polynomial of total degree $r$, that $e \in \mathcal{H}_k(\mathcal{X})$, and $k \in C^{2s}(\mathcal{X} \times \mathcal{X})$.
Then there exists an explicit $n$- and $h$-independent constant $C_{r,s}$ such that
\begin{align*}
\underbrace{ |f(\mathbf{0}) - \mathbb{E}[\mathsf{f}(\mathbf{0}) | \mathsf{f}(\mathbf{x}) = f(\mathbf{x}), \, \mathbf{x} \in X_n^h] | }_{\text{\normalfont extrapolation error for posterior mean of GP}} \; \leq \; \underbrace{ C_{r,s} \; \rho_{X_n,\mathcal{X}}^s \; \| e \|_{\mathcal{H}_k(\mathcal{X})} }_{\text{\normalfont constant in $h$}} \; \underbrace{ \phantom{|_{|^|}} h^s}_{ \text{\normalfont acceleration}} \; \; \underbrace{ \sup_{\mathbf{x} \in \mathcal{X}_h} |\epsilon(\mathbf{x})| }_{\text{\normalfont error of original method }}
\end{align*}
whenever the box fill distance satisfies $\rho_{X_n , \mathcal{X}} \leq 1 / (\gamma_d (r + 2s))$.
\end{theorem}

\noindent Compared to \ac{MRE}, the number $n$ of data does not need to be equal to $\mathrm{dim}(A)$.
However, convergence acceleration in \Cref{thm: GRE} requires $s \geq 1$, and the box fill distance requirement becomes more stringent as $r + 2s$ increases.
Thus, even the weakest requirement $\rho_{X_n,\mathcal{X}} \leq 1/\gamma_d$, corresponding to $r + 2s = 1$, already entails at least $n \geq (2^d d!)^d$, so that $n_{\min} \geq (2^d d!)^d$, which is super-exponential in $d$.
This is undesirable for complex simulators with multiple sources of discretisation, since running a large number of simulations can be impractical.
Though it is possible that the data requirement in \Cref{thm: GRE} could be reduced using a sharper argument, at least an exponential dimension dependence appears fundamental to \ac{GRE}, since a kernel method is being used to approximate the normalised error $e(\mathbf{x})$, an unstructured $d$-dimensional functional \citep[such problems require space-filling designs, which manifests as a fill distance requirement; see][]{wendland2004scattered}.
On the other hand, \ac{GRE} unlocks statistical functionalities not available to \ac{MRE}, and empirical results show \ac{GRE} can be effective in dimensions $d \in \{1,2\}$ \citep{oates2024probabilistic}.
The present paper seeks to replicate these functionalities while addressing the super-exponential data requirement of \ac{GRE}, as explained next.

\section{Methods}
\label{sec: methods}

This section contains our novel methodological development, with \ac{SPRE} introduced in \Cref{sec: SPRE}, and theoretical analysis contained in \Cref{subsec: converge accel}.
Remarkably, it is shown that not only does  \ac{SPRE} match the convergence acceleration of \ac{MRE} for the same data requirement (which in turn depends only on the extrapolation dimension), but \ac{SPRE} can unlock \emph{additional} convergence acceleration if appropriate regularity conditions are satisfied.
\Cref{sec: UQ} concerns implementational detail, such as the data-driven identification of the index set $A$ and experimental design to select the simulator configurations $\{\mathbf{x}_i\}_{i=1}^n$ which comprise the dataset.

\subsection{Sparse Probabilistic Richardson Extrapolation}
\label{sec: SPRE}

A major strength of our \ac{SPRE} method is that it is straight-forward to state and to implement.
The idea is to model $f$ as a \ac{GP}, whose mean function
\begin{align*}
    \mathbb{E}[\mathsf{f}(\mathbf{x}) | \{\beta_{\bm{\alpha}} : \bm{\alpha} \in A \} ] = \sum_{\bm{\alpha} \in A} \beta_{\bm{\alpha}} \mathbf{x}^{\bm{\alpha}} 
\end{align*}
depends on unknown coefficients $\beta_{\bm{\alpha}}$, which are in turn assigned an improper flat prior and integrated out, while the covariance function $\mathbb{C}[\mathsf{f}(\mathbf{x}) , \mathsf{f}(\mathbf{x}')] = k(\mathbf{x} , \mathbf{x}')$ is specified in terms of a symmetric and positive definite kernel $k$ to be specified.
Following standard calculations\footnote{See for instance Section 2.7 of \citet{williams2006gaussian}.} the posterior mean and variance can be explicitly computed.
Indeed, letting $\mathbf{V}_A := \mathbf{V}_A(X_n)$, $\mathbf{K} := [k(\mathbf{x}_i,\mathbf{x}_j)]_{i,j = 1}^n$, $\mathbf{f} := [f(\mathbf{x}_i)]_{i=1}^n$, $\mathbf{k}(\mathbf{x})  := [k(\mathbf{x}_i,\mathbf{x})]_{i=1}^n$, and $\mathbf{v}_A(\mathbf{x}) := [\mathbf{x}^{\bm{\alpha}_i}]_{i=1}^m$, we have that
\begin{align}
  \hspace{-10pt}  \mathbb{E}[\mathsf{f}(\mathbf{x}^*) | \mathsf{f}(\mathbf{x}) = f(\mathbf{x}), \, \mathbf{x} \in X_n] & = \mathbf{k}(\mathbf{x}^*)^\top \mathbf{K}^{-1} \mathbf{f} + \mathbf{r}_A(\mathbf{x}^*)^\top \hat{\bm{\beta}}_A \label{eq: SPRE mean}  \\
  \hspace{-10pt}  \mathbb{V}[\mathsf{f}(\mathbf{x}^*) | \mathsf{f}(\mathbf{x}) = f(\mathbf{x}), \, \mathbf{x} \in X_n] & = k(\mathbf{x}^*,\mathbf{x}^*) - \mathbf{k}(\mathbf{x}^*)^\top \mathbf{K}^{-1} \mathbf{k}(\mathbf{x}^*) + \mathbf{r}_A(\mathbf{x}^*)^\top ( \mathbf{V}_A^\top \mathbf{K}^{-1} \mathbf{V}_A )^{-1} \mathbf{r}_A(\mathbf{x}^*) \label{eq: SPRE var}
\end{align}
where $\mathbf{r}_A(\mathbf{x}) := \mathbf{v}_A(\mathbf{x}) - \mathbf{V}_A^\top \mathbf{K}^{-1} \mathbf{k}(\mathbf{x})$ and $\hat{\bm{\beta}}_A := ( \mathbf{V}_A^\top \mathbf{K}^{-1} \mathbf{V}_A )^{-1} \mathbf{V}_A^\top \mathbf{K}^{-1} \mathbf{f}$.
In particular we are interested in prediction at $\mathbf{x}^* = \mathbf{0}$, where under our convention $\mathbf{v}_A(\mathbf{0}) = \mathbf{e}_1$.
Note that $X_n$ is required to be $\mathcal{P}_A$-unisolvent in order for $\mathbf{V}_A^\top \mathbf{K}^{-1} \mathbf{V}_A$ to be inverted.

From a statistical perspective, the regression model underlying \ac{SPRE} is \emph{universal kriging} \citep{matheron1969krigeage,cressie1993statistics,stein1999interpolation}; that \ac{GP} regression with an improper flat prior on the coefficients of the mean function reproduces the universal kriging predictor and its variance is due to \citet{ohagan1978curve,handcock1993bayesian}.
The novelty of \ac{SPRE} therefore lies not in the regression model itself, but in its use for extrapolation with a \emph{sparse} index set $A$.
The central technical question, answered in \Cref{subsec: converge accel}, is whether universal kriging accelerates convergence when the number $n$ of data is \emph{fixed} and the data are collectively drawn towards the extrapolation point $\mathbf{0}$; this regime is quite different from the classical analysis of kriging, which concerns interpolation as the fill distance of the design vanishes \citep[Chapter 11]{wendland2004scattered}.

\begin{remark}[Extrapolation sparsity and SPRE]
    The minimum number of data required for \ac{SPRE} is $d_{\mathrm{ext}}(f)$, identical to \ac{MRE}.
    Thus \ac{MRE} and \ac{SPRE} are able to exploit extrapolation sparsity while \ac{GRE} is not.
    An illustration of \ac{MRE}, \ac{GRE} and \ac{SPRE} applied to the cubature setting of \Cref{subsec: cubature} is presented in \Cref{fig: illus}, where the reduced data requirement for \ac{SPRE} is striking compared to \ac{GRE}. 
    Full details of this experiment are contained in \Cref{app: cubature illustration}.
\end{remark}

\begin{figure}[t!]
\includegraphics[width=\textwidth]{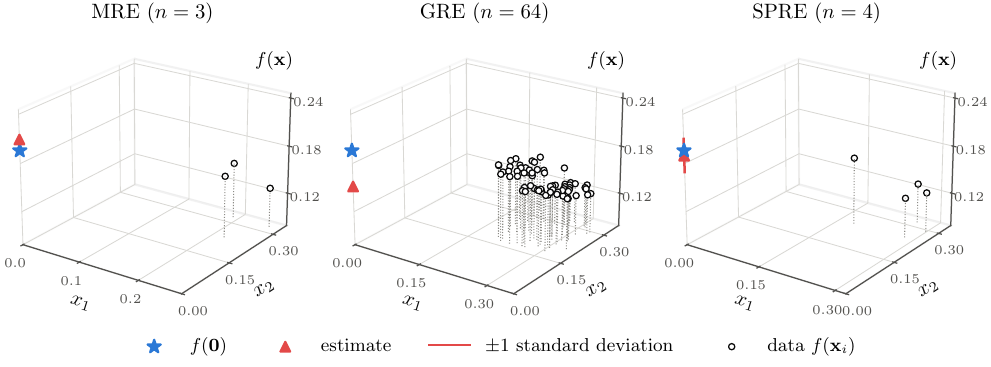}
\caption{Illustration of extrapolation methods to predict $f(\mathbf{0})$ from data $\{f(\mathbf{x}_i)\}$ on the multivariate ($d=2$) extrapolation task described in \Cref{subsec: cubature}.
Here the extrapolation dimension was $d_{\mathrm{ext}}(f) = 3$.
Left: Classical \Acf{MRE} requires a dataset of size $n_{\min} = d_{\mathrm{ext}}(f)$, to which a polynomial is fitted.
Middle: \Acf{GRE} requires a dataset of size $n_{\min} \geq (2^d d!)^d = 64$ to guarantee convergence acceleration, and achieves this by fitting a numerical analysis-informed \ac{GP}.
Right: \Acf{SPRE} applied to a dataset of size $n = 1 + d_{\mathrm{ext}}(f)$, one more than its minimum $n_{\min} = d_{\mathrm{ext}}(f)$, where the additional data point enables kernel parameters to be estimated using leave-one-out cross-validation, as described in \Cref{sec: UQ}.
In each panel $n_{\min}$ denotes the smallest dataset size for which convergence acceleration is guaranteed by theory (cf. \Cref{subsec: extrap sparse}).
\alttext{A figure showing MRE, GRE and SPRE applied to a two-dimensional extrapolation task.}
}
\label{fig: illus}
\end{figure}

\begin{remark}[$\mathcal{P}_A$ exactness of SPRE]
\label{rem: exactness}
    Suppose that $f \in \mathcal{P}_A$, so that $\mathbf{f} = \mathbf{V}_A \bm{\beta}_A$ where, under the convention of \Cref{subsec: notation}, the first entry of $\bm{\beta}_A$ corresponds to $f(\mathbf{0})$, the quantity of interest.
    Assuming $X_n$ is $\mathcal{P}_A$-unisolvent, we can substitute this expression into \eqref{eq: SPRE mean} to obtain $\mathbb{E}[\mathsf{f}(\mathbf{0}) | \mathsf{f}(\mathbf{x}) = f(\mathbf{x}), \, \mathbf{x} \in X_n] = \mathbf{e}_1^\top \bm{\beta}_A = f(\mathbf{0})$, so that the posterior mean of the \ac{GP} coincides with the true value of $f(\mathbf{0})$; we say that \ac{SPRE} is $\mathcal{P}_A$ \emph{exact}.
\end{remark}

\begin{remark}[\ac{MRE} as a special case of \ac{SPRE}]
\label{rem: spre as sem}
Suppose that $n = \mathrm{dim}(A)$ and that the set $X_n$ is $\mathcal{P}_A$-unisolvent, so that there is a unique element $p \in \mathcal{P}_A$ that coincides with $f$ on $X_n$.
From \Cref{rem: exactness}, the posterior mean of \ac{SPRE} coincides with $p(\mathbf{0})$, which is also the output of \ac{MRE}. 
In this edge case, a convergence acceleration result for \ac{SPRE} can be deduced from \Cref{thm: SEM}.    
However, unlike \ac{MRE}, \ac{SPRE} naturally extends to settings where $n > \mathrm{dim}(A)$, provides rigorous probabilistic uncertainty quantification for $f(\mathbf{0})$ which can be integrated into downstream inference and decision support, and unlocks cost-constrained experimental design and data-driven model selection functionality (cf. \Cref{sec: UQ}).
\end{remark}

\begin{remark}[Comparing \ac{GRE} and \ac{SPRE}]
    Compared to \ac{GRE}, \ac{SPRE} can be implemented with exponentially fewer data in multiple dimensions when extrapolation sparsity is present (cf. \Cref{thm: convergence SPRE}), ameliorates the need for careful choice of kernel by admitting a well-defined limiting case in which minimal smoothness is assumed (cf. \Cref{lem: equiv norms}), and naively supports discrete domains $\mathcal{X}$, which occur in practice for example when a time step size $x_1$ is required to exactly divide a fixed length time interval $[0,T]$, so that $x_1 \in \{\frac{T}{m} : m \in \mathbb{N}\}$.
\end{remark}

At first sight it may seem paradoxical that \ac{SPRE}, which must estimate the $\mathrm{dim}(A)$ coefficients of its mean function, requires fewer data than \ac{GRE}, whose mean function is a constant.
The explanation is that the two methods place the burden of approximation in different places.
In \ac{GRE}, the dependence of $f(\mathbf{x})$ on $\mathbf{x}$ must be learned through the kernel, i.e. nonparametrically, and approximating an unstructured function of $d$ variables requires a space-filling design whose size grows at least exponentially in $d$.
In \ac{SPRE}, the leading behaviour of $f$ near $\mathbf{0}$ is captured parametrically by the $\mathrm{dim}(A)$ coefficients $\beta_{\bm{\alpha}}$, and the kernel need only account for the residual $R$, which is asymptotically negligible as $\mathbf{x} \rightarrow \mathbf{0}$.
Accordingly, \Cref{thm: convergence SPRE} below requires only that $X_n$ is $\mathcal{P}_A$-unisolvent, with no fill distance condition of the kind appearing in \Cref{thm: GRE}.
Put simply, if $f$ is close to a polynomial with few terms, then estimating those few coefficients is far cheaper than learning $f$ nonparametrically; this is precisely the situation described by extrapolation sparsity.
The price to be paid is that the index set $A$ must be specified or learned, which we address in \Cref{sec: UQ}.

The potential for convergence acceleration noted in \Cref{rem: spre as sem} is certainly not the end of the story for \ac{SPRE}.
Indeed, it would be desirable to make use of all available data even if the number of data $n$ exceeds $d_{\mathrm{ext}}(f)$.
This rules out a trivial proof of convergence acceleration along the lines of \Cref{rem: spre as sem}, and an alternative strategy is required.
In particular, any theoretical analysis will need to account for the effect of the choice of kernel.
Further, there are the important practical matters of how the index set $A$ can be estimated from data, and how the evaluation locations $\{\mathbf{x}_i\}_{i=1}^n$ should be designed to be maximally informative subject to a total computational budget.
These questions will all be addressed, beginning with a novel analysis of convergence acceleration.

\subsection{Convergence Acceleration}
\label{subsec: converge accel}

Before commencing our theoretical analysis, we recall that often $\mathcal{X}$ is not a continuum but rather a discrete set.
For example, a simple quadrature on $[0,1]$ with uniformly-spaced nodes is only possible when the distance between nodes is of the form $\frac{1}{N}$ for some $N \in \mathbb{N}$.
To accommodate both continuous and discrete scenarios, we define \emph{permissible scalings} to be values of $h$ such that $h \mathbf{x} \in \mathcal{X}$ whenever $\mathbf{x} \in \mathcal{X}$.
As an example, if $\mathcal{X} = \{\frac{1}{N} : N \in \mathbb{N}\}$ then we can take $h = \frac{1}{M}$ for any $M \in \mathbb{N}$.
For our conclusions to be meaningful we assume that the set of permissible scalings has $0$ as an accumulation point.
Now we are in a position to state \Cref{thm: convergence SPRE}, our main convergence result.
The proofs of the following result is contained in \Cref{subsec: SPRE converge proof}.

\begin{theorem}[Convergence acceleration for SPRE]
\label{thm: convergence SPRE}
Let $\mathbf{0} \in \mathcal{X} \subset \mathbb{R}^d$ and let $X_n = \{\mathbf{x}_i\}_{i=1}^n \subset \mathcal{X}$ be $\mathcal{P}_A$-unisolvent.
Let $\mathcal{X}_h = \{h \mathbf{x}\}_{\mathbf{x} \in \mathcal{X}} \subset \mathcal{X}$ and $X_n^h = \{ h \mathbf{x}_i\}_{i=1}^n \subset \mathcal{X}$ for permissible scalings $h \in (0,1]$. 
Let $k : \mathcal{X} \times \mathcal{X} \rightarrow \mathbb{R}$ be a symmetric, positive definite and bounded kernel.
Let $f : \mathcal{X} \rightarrow \mathbb{R}$ be a simulator for which \eqref{eq: error} holds with $R \in \mathcal{H}_k(\mathcal{X})$.
Then for each permissible scaling $h$ there exists $C_{h,k} \in \mathbb{R}$ such that
\begin{align}
\underbrace{ | f(\mathbf{0}) - \mathbb{E}[\mathsf{f}(\mathbf{0}) | \mathsf{f}(\mathbf{x}) = f(\mathbf{x}), \, \mathbf{x} \in X_n^h] | }_{\text{\normalfont extrapolation error for the posterior mean of the GP }} \leq \underbrace{C_{h,k}}_{\text{\normalfont bounded in $h$}} \quad \underbrace{ \| R \|_{\mathcal{H}_k(\mathcal{X}_h)} }_{\text{\normalfont acceleration}} , \label{eq: main bound}
\end{align}
and moreover $\sup_{h \in (0,1]} C_{h,k} < \infty$.
\end{theorem}

\noindent Thus \ac{SPRE} can accelerate convergence up to a rate determined by the smoothness of the residual $R$ in \eqref{eq: error} as measured using the norm associated with the kernel.
Here the number $n$ of data is fixed, and the acceleration term $\|R\|_{\mathcal{H}_k(\mathcal{X}_h)}$ is a norm over the \emph{shrinking} set $\mathcal{X}_h$, and is therefore $h$-dependent; it plays the same role as the term $\sup_{\mathbf{x} \in \mathcal{X}_h} |R(\mathbf{x})|$ in \Cref{thm: SEM}.
In \Cref{lem: equiv norms} below we show that, for a suitably rough kernel, the norm $\|R\|_{\mathcal{H}_k(\mathcal{X}_h)}$ reduces to the magnitude $\sup_{\mathbf{x} \in \mathcal{X}_h} |R(\mathbf{x})|$ of the residual, so that \Cref{thm: convergence SPRE} recovers the rate for \ac{MRE} in \Cref{thm: SEM}, demonstrating that \eqref{eq: main bound} is sharp in general.
Moreover, when $n = d_{\mathrm{ext}}(f)$ (the minimum $n$ for which $X_n$ can be $\mathcal{P}_A$-unisolvent), convergence acceleration is achieved for \emph{any} kernel, since in this case the posterior mean of \ac{SPRE} coincides with the output of \ac{MRE} (\Cref{rem: spre as sem}).

\begin{figure}[t!]
\includegraphics[width=\textwidth]{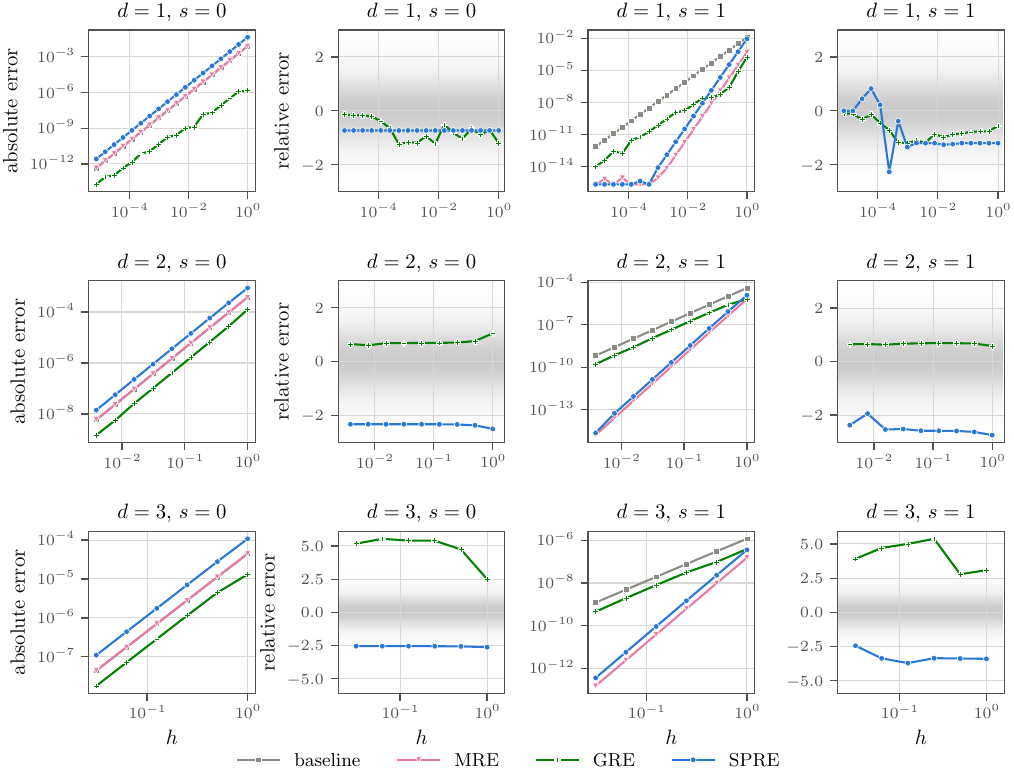}
\caption{Convergence in the $d$-dimensional cubature setting of \Cref{subsec: cubature}. 
Here $s \in \{0,1\}$ controls the smoothness of the extrapolation task; for $s = 0$ there is no smoothness to exploit, while for $s = 1$ there is smoothness that can be exploited.
The absolute error is displayed for \acf{MRE}, \acf{GRE} and \acf{SPRE}.
As a baseline, we also consider the estimator $f(\mathbf{x})$ where $\mathbf{x}$ is the element closest to $\mathbf{0}$ in the dataset $X_n^h$.
The relative error \eqref{eq: rel error} is displayed for \ac{GRE} and \ac{SPRE}, where the shaded region represents the density of a standard normal.
The full experimental protocol is described in \Cref{app: cubature illustration}. 
\alttext{A figure with 12 panels arranged as a $3 \times 4$ grid.  Absolute errors and relative errors are shown as a function of the scaling parameter $h$, as the dimension $d$ is varied in $\{1,2,3\}$ and the smoothness $s$ is varied in $\{0,1\}$.}
}
\label{fig: convergence}
\end{figure}

The convergence acceleration achieved by \ac{SPRE} is demonstrated in the setting of \Cref{subsec: cubature} in \Cref{fig: convergence}.
The index set $A$ was selected according to \eqref{eq: Riemann Taylor}, while all kernel parameters were optimised using cross-validation as described in \Cref{sec: UQ}.
The full experimental protocol is described in \Cref{app: cubature illustration}.
The first and third columns of \Cref{fig: convergence} present the \emph{absolute error} of \ac{MRE}, and the posterior mean estimators for \ac{GRE} and \ac{SPRE}, i.e.
\begin{align}
    \text{absolute error} = |f(\mathbf{0}) - \mathbb{E}[\mathsf{f}(\mathbf{0}) | \mathsf{f}(\mathbf{x}) = f(\mathbf{x}), \, \mathbf{x} \in X_n^h] |  \label{eq: abs error}
\end{align}
As a baseline, we also consider the estimator $f(\mathbf{x})$ where $\mathbf{x}$ is the element closest to $\mathbf{0}$ in $X_n^h$.
For $s = 0$, where there is no regularity to exploit, all methods converge (in $h$) at the same rate as the baseline method, as expected.
That \ac{GRE} improved upon the baseline can be attributed to the numerical analysis-aware \ac{GP}, which exploits knowledge that the Riemann sum midpoint rule is a second-order accurate numerical method.
For $s = 1$, both \ac{MRE} and \ac{SPRE} can be seen to converge faster than the baseline, consistent with \Cref{thm: SEM,thm: convergence SPRE}; the actual convergence rates are clarified below \Cref{lem: equiv norms}.
Note that the fast convergence is observed only until the floating point limit of $10^{-16}$ is reached.
On the other hand, no rate improvement over the baseline is observed with \ac{GRE}.
This can be attributed to the number of data $n$ being too small for \ac{GRE}, or equivalently that \ac{GRE} is unable to exploit extrapolation sparsity in the manner of \ac{MRE} and \ac{SPRE}.
In the second and forth columns of \Cref{fig: convergence} we plot the \emph{relative error}
\begin{align}
    \text{relative error} = \frac{f(\mathbf{0}) - \mathbb{E}[\mathsf{f}(\mathbf{0}) | \mathsf{f}(\mathbf{x}) = f(\mathbf{x}), \, \mathbf{x} \in X_n^h] }{ \sqrt{ \mathbb{V}[\mathsf{f}(\mathbf{0}) | \mathsf{f}(\mathbf{x}) = f(\mathbf{x}), \, \mathbf{x} \in X_n^h] } }  \label{eq: rel error}
\end{align}
which captures whether the statistical error bars are commensurate with the size of the actual error.
It can be seen that both \ac{GRE} and \ac{SPRE} are reasonably well-calibrated as a result of the kernel parameter learning described in \Cref{sec: UQ}; the actual error is usually within $3$ standard deviations of the mean for each method, the main exception being \ac{GRE} in dimension $d = 3$, which is somewhat over-confident.

\Cref{thm: convergence SPRE} suggests selecting the kernel $k$ to minimise the bound in \eqref{eq: main bound}.
For $k$ too smooth, the term $\|R\|_{\mathcal{H}_k(\mathcal{X}_h)}$ will fail to exist, placing an upper limit on the smoothness of the kernel.
On the other hand, the coefficient $C_{h,k}$ will tend to be larger for a rougher kernel.
Remarkably, this trade-off enables \ac{SPRE} to achieve additional convergence acceleration beyond \ac{MRE}, as clarified by the next result, whose proof is contained in \Cref{app: extra smooth}:

\begin{theorem}[Further convergence acceleration for \ac{SPRE}]
\label{thm: extra smooth}
In the setting of \Cref{thm: convergence SPRE}.
Let $s$ be maximal such that $\{|\bm{\alpha}| < 2 s \} \subseteq A$.
Let $k : \mathcal{X} \times \mathcal{X} \rightarrow \mathbb{R}$ be a symmetric and positive definite kernel of the form $k(\mathbf{x},\mathbf{y}) = \phi(\mathbf{x} - \mathbf{y})$ where $\phi : \mathbb{R}^d \rightarrow \mathbb{R}$ is $2s$ times continuously differentiable at $\mathbf{0}$.
Then $C_{h,k} = O( h^s )$ as $h \rightarrow 0$.
\end{theorem}

\noindent Informally, the addition of the kernel, which is not part of \ac{MRE}, enables any additional smoothness present in $R$ to be exploited.
For balance, \emph{if} one knew there was additional smoothness at the start, then one could leverage a higher-order expansion in \eqref{eq: error} and improve the performance of \ac{MRE}; \Cref{thm: extra smooth} demonstrates that this knowledge is not needed upfront for \ac{SPRE}.
In practice, kernel learning provides a principled statistical approach to \ac{SPRE} that realises this potential; see \Cref{sec: UQ}.

\begin{remark}[Beyond translation-invariant kernels]
For presentational purposes \Cref{thm: extra smooth} assumes a translation-invariant kernel, but our argument applies to any kernel $k$ that admits a Taylor-type expansion using terms from $\mathcal{P}_A$.
In this case, one can relate $C_{h,k}$ to the square root of the remainder term in the Taylor expansion.
\end{remark}

A conservative choice of kernel could be the \emph{white noise} kernel $k(\mathbf{x} , \mathbf{x}') = \sigma^2  \delta_{\mathbf{x} , \mathbf{x}'}$, for some amplitude $\sigma^2 > 0$, for which minimal regularity of $R$ is required.
Remarkably, even for this kernel \ac{SPRE} performs at least as well as \ac{MRE}.
The proof of the following result is contained in \Cref{app: proof of equiv norms}:

\begin{proposition}[\Cref{thm: convergence SPRE} is sharp]
\label{lem: equiv norms}
Let $k(\mathbf{x} , \mathbf{x}') = \sigma^2  \delta_{\mathbf{x} , \mathbf{x}'}$. 
Consider a discrete state space $\mathcal{X} = \{\mathbf{0}\} \cup \{ \frac{1}{ 2^m}  \frac{1}{ 2^{\bm{\alpha}} }  : m \in \mathbb{N}_0, \; 0 \leq |\bm{\alpha}| \leq s \}$ together with designs of the form $X_n^h =  \{ h \cdot \frac{1}{ 2^{\bm{\alpha}} }  : 0 \leq |\bm{\alpha}| \leq s \}$ for the permissible scalings $h = 1 / 2^M$, $M \in \mathbb{N}_0$.
Suppose that $|R(\mathbf{x})| \asymp \|\mathbf{x}\|^a$ for some $a > 0$ in the $\mathbf{x} \rightarrow \mathbf{0}$ limit.
Then $\|R\|_{\mathcal{H}_k(\mathcal{X}_h)} = O\left( \sup_{\mathbf{x} \in \mathcal{X}_h} |R(\mathbf{x})| \right)$ in the $h \rightarrow 0$ limit.
\end{proposition}

\noindent That is, the acceleration term appearing in the bound of \Cref{thm: convergence SPRE} decays at least as rapidly as the acceleration term appearing in \Cref{thm: SEM}.
Since \Cref{thm: SEM} is sharp, it follows that \Cref{thm: convergence SPRE} is sharp as well.
Further, \Cref{lem: equiv norms} offers insight into the convergence rates for the absolute error observed in \Cref{fig: convergence}, for which the white noise kernel was used. 
Indeed, from \eqref{eq: Riemann Taylor} we have $\sup_{\mathbf{x} \in \mathcal{X}_h} |R(\mathbf{x})| = O(h^{2s+2})$ so that for $s = 0$ we expect a convergence rate $O(h^2)$, which is indeed observed in the left column of \Cref{fig: convergence}.
For $s = 1$ we expect an asymptotic convergence rate $O(h^4)$; this can be observed in the centre-right column of \Cref{fig: convergence}.

\begin{remark}[Cost-aware convergence acceleration]
\label{cor: cost}
The results above are stated in terms of the scaling $h$, but ultimately it is the error as a function of the total computational budget that is of interest.
Because the number $n$ of data is fixed, our results translate directly into such statements:
Let $X_n = \{\mathbf{x}_i\}_{i=1}^n \subset \mathcal{X} \setminus \{\mathbf{0}\}$, and suppose that the cost function satisfies $c(h \mathbf{x}) = h^{-\gamma} c(\mathbf{x})$ for some $\gamma > 0$, all $\mathbf{x} \in \mathcal{X} \setminus \{\mathbf{0}\}$ and all permissible scalings $h$, so that the total cost of the dataset $X_n^h$ is $C_h := \sum_{i=1}^n c(h\mathbf{x}_i) = h^{-\gamma} C_1$.
Suppose that $|f(h\mathbf{x}_i) - f(\mathbf{0})| \asymp h^b$ as $h \rightarrow 0$ for each $i$, and that an extrapolation method based on $X_n^h$ has error $O(h^a)$ for some $a > b$.
Then the extrapolation error is $O(C_h^{-a/\gamma})$, whereas each simulation in the dataset, whose cost $c(h\mathbf{x}_i) = h^{-\gamma} c(\mathbf{x}_i)$ is a fixed fraction of $C_h$, has error $\asymp C_h^{-b/\gamma}$.
This is analogous in spirit to complexity analyses of multilevel Monte Carlo \citep{giles2015multilevel}.
A more refined analysis, in which the allocation of the budget across design points is itself optimised, would require theory for the experimental design problem \eqref{eq: constrained opt}, which we leave for future work (see \Cref{sec: discussion}).
\end{remark}

Having established that convergence acceleration is achieved, our attention now turns to implementational detail.

\subsection{Implementational Detail}
\label{sec: UQ}

This section discusses data-driven selection of the kernel $k$, the index set $A$ and the design $\{\mathbf{x}_i\} \subset \mathcal{X}$.
Careful selection of $k$ enables maximal convergence acceleration (cf. \Cref{thm: extra smooth}), learning $A$ circumvents the mathematical analysis that would be needed to derive an expansion \eqref{eq: error} for settings where the simulator $f$ is complicated, and experimental design enables us to select configurations $\{\mathbf{x}_i\}$ that are maximally informative about $f(\mathbf{0})$ subject to a constrained computational budget.
However, to simultaneously select all three is difficult; data are required to learn $k$ and $A$, but it is difficult to design the set $\{\mathbf{x}_i\}$ of simulation configurations without knowledge of $k$ and $A$.
This can be resolved either by fixing $\{\mathbf{x}_i\}$ (indeed, oftentimes the statistician does not have influence over the design set), or by using a sequential strategy that alternates between learning $k$ and $A$ and selecting which simulations to run next:

\paragraph{Learning $k$ and $A$}

In this work we employed \ac{LOOCV}.
Given a design $X_n = \{\mathbf{x}_i\}_{i=1}^n$, let $X_n^{(i)} := X_n \setminus \{\mathbf{x}_i\}$.
Let $p(\cdot ; \mu , \sigma^2)$ denote the density of $N(\mu,\sigma^2)$.
Then the suitability of a kernel $k$ and a set $A$ can be quantified by the leave-one-out predictive log-density
\begin{align}
L(k,A) = - \sum_{i=1}^n \log p( f_i ; \mu_i , \sigma_i^2 ) , \qquad \begin{array}{ll} \mu_i & := \mathbb{E}_{k,A}[\mathsf{f}(\mathbf{x}_i) | \mathsf{f}(\mathbf{x}) = f(\mathbf{x}), \, \mathbf{x} \in X_n^{(i)}] \\ \sigma^2_i & := \mathbb{V}_{k,A}[\mathsf{f}(\mathbf{x}_i) | \mathsf{f}(\mathbf{x}) = f(\mathbf{x}), \, \mathbf{x} \in X_n^{(i)}] \end{array}  \label{eq: LOOCV error}
\end{align} 
which we seek to minimise through our choice of $k$ and $A$.
Expressions for the predictive mean and variance can be read off from \eqref{eq: SPRE mean} and \eqref{eq: SPRE var}, and we have indicated their $k$ and $A$ dependence in the subscript.
For the kernel, we fix a parametric family $k_\theta$ and use automatic differentiation to perform gradient-based optimisation on $\theta \mapsto L(k_\theta,A)$ using an adaptive trust region method.
The specific kernels that we used are detailed in \Cref{app: kernels}.
The minimal value of the \ac{LOOCV} error obtained in this manner is denoted 
\begin{align*}
    L(A) = \inf_{\theta} \; L(k_\theta,A) .
\end{align*}
For the set $A$, we perform a stepwise method that is aligned with the notion of extrapolation sparsity, prioritising the search for `small' sets $A$.
Starting with $A_0 = \{\mathbf{0}\}$, at iteration $i$ we consider adding each monomial of order $i$ (i.e. $\mathbf{x}^{\bm{\alpha}}$ with $|\bm{\alpha}| = i$) in turn, to see if the error criterion $L(A_{i-1} \cup \{\mathbf{x}^{\bm{\alpha}}\})$ is reduced compared to $L(A_{i-1})$.  
If there are no qualifying monomials we terminate and return $A_{i-1}$.
Else, we add \emph{all} qualifying monomials of order $i$ to obtain an enlarged model $A_i$.
If the error $L(A_i)$ of this enlarged model is reduced relative to $L(A_{i-1})$ then we increment $i$ and proceed, else we terminate\footnote{The algorithm must terminate because the data $X_n$ will no longer be $\mathcal{P}_{A_i}$-unisolvent; for example, this is guaranteed when $\mathrm{dim}(A_i)$ exceeds $n$.} and return $A_{i-1}$.
The multiple rounds of optimisation entail a computational overhead, but this will typically be negligible with respect to the computational cost of running the simulations $f(\mathbf{x}_i)$ that comprise the dataset.

\paragraph{Experimental design}

Given an existing design $X_n$, we seek to augment this with $m$ new configurations $X_m^*$ such that uncertainty regarding $f(\mathbf{0})$ is minimised.
Let $k$ and $A$ be the kernel and index set learned based on $X_n$ in the manner just described.
Recall that $c(\mathbf{x})$ denotes the cost associated with obtaining the simulation output $f(\mathbf{x})$.
Thus, for a total computational budget $C$, we consider 
\begin{align}
\argmin_{X_m^* \subset \mathcal{X} } \mathbb{V}_{k,A}[\mathsf{f}(\mathbf{0}) | \mathsf{f}(\mathbf{x}) = f(\mathbf{x}), \, \mathbf{x} \in X_n \cup X_m^*] \quad \text{s.t.} \quad \sum_{i=1}^m c(\mathbf{x}_i^*) \leq C .  \label{eq: constrained opt}
\end{align}
The variance objective can be explicitly computed due to \eqref{eq: SPRE var}, which depends on $X_m^*$ but does not depend on the actual values $\{f(\mathbf{x}), \mathbf{x} \in X_m^*\}$.
Solving for $X_m^*$ requires numerical methods for constrained (and possibly discrete) optimisation to be used.

\bigskip

\begin{figure}[t!]
    \centering
    \includegraphics[width=\textwidth]{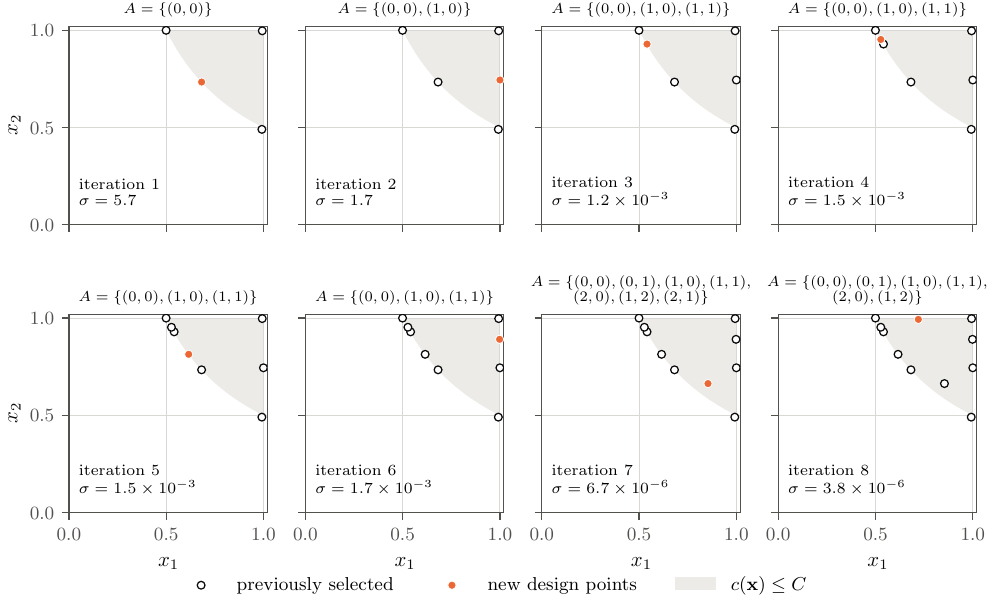}
    \caption{Illustrating iterative learning of $(A,k)$ and experimental design for $\{\mathbf{x}_i\}$.  
    An initial design (hollow circles) is used to estimate the index set $A$ and the kernel $k$; here the Mat\'{e}rn-$\frac{1}{2}$ kernel $k(\mathbf{x},\mathbf{x}') = \sigma^2 \exp(-\|\mathbf{x}-\mathbf{x}'\|/\ell)$ was used and the amplitude parameter $\sigma$ is reported. 
    Experimental design is then used to select new design points (filled circles) which are added to the existing design set (hollow circles).
    The cost function was $c(\mathbf{x}) = \mathbf{x}^{-\mathbf{1}}$ (the shaded region indicates where $c(\mathbf{x}) \leq C$), and the true index set $A$ of maximal size, which for this example is $\{(0,0),(1,0),(0,1),(2,0)\}$, is contained in the estimated index set from iteration 7 onwards.
    \alttext{A figure consisting of 8 panels, each of which shows the experimental design obtained after a certain number of iterations of sequential experimental design have been performed.}
    }
    \label{fig: ED}
\end{figure}

\noindent Alternating between learning the pair $(k,A)$ and selecting configurations $\{\mathbf{x}_i\}$ provides a general framework for \ac{SPRE}.
An illustration is provided in \Cref{fig: ED}, where after 7 iterations the estimated index set contains the true index set $A$ of maximal size, and $f(\mathbf{0})$ is accurately estimated.
For this example the cost function was $c(\mathbf{x}) = \mathbf{x}^{-\mathbf{1}}$ and the computational budget was $C = 2$, which here permits one new simulation at each iteration.
Full details for this experiment are contained in \Cref{app: ED detail}.

Nevertheless, we acknowledge that in many practical applications the statistician does not have complete freedom to select the design set $\{\mathbf{x}_i\}$, and even if they did, the values $c(\mathbf{x}_i^*)$ taken by the cost function can be unknown at the outset, for instance because the time required to run a simulation is difficult to predict.
Thus in \Cref{sec: experiments}, where we focus on several challenging examples of simulators, we continue to learn $(k,A)$ but for the most part we consider the design $\{\mathbf{x}_i\}$ to be fixed.

\section{Empirical Assessment}
\label{sec: experiments}

This section presents three case studies of increasing challenging numerical simulations; two colliding spheres (\Cref{sec:two-spheres}), five colliding shapes (\Cref{subsec: five shapes}), and a deterministic flocking model in which 60 interacting agents are simultaneously simulated (\Cref{subsec: agent}).
To preempt our results, we will see that extrapolation sparsity appears to hold in each case study, suggesting that this phenomenon may be generic in complex numerical simulation.
Accordingly, we will also see evidence of \ac{SPRE} out-performing \ac{GRE} in these settings where extrapolation sparsity can be exploited.

\subsection{Two Spheres 3D Model}
\label{sec:two-spheres}

\begin{figure}[t!]
\centerline{\includegraphics[width=\textwidth]{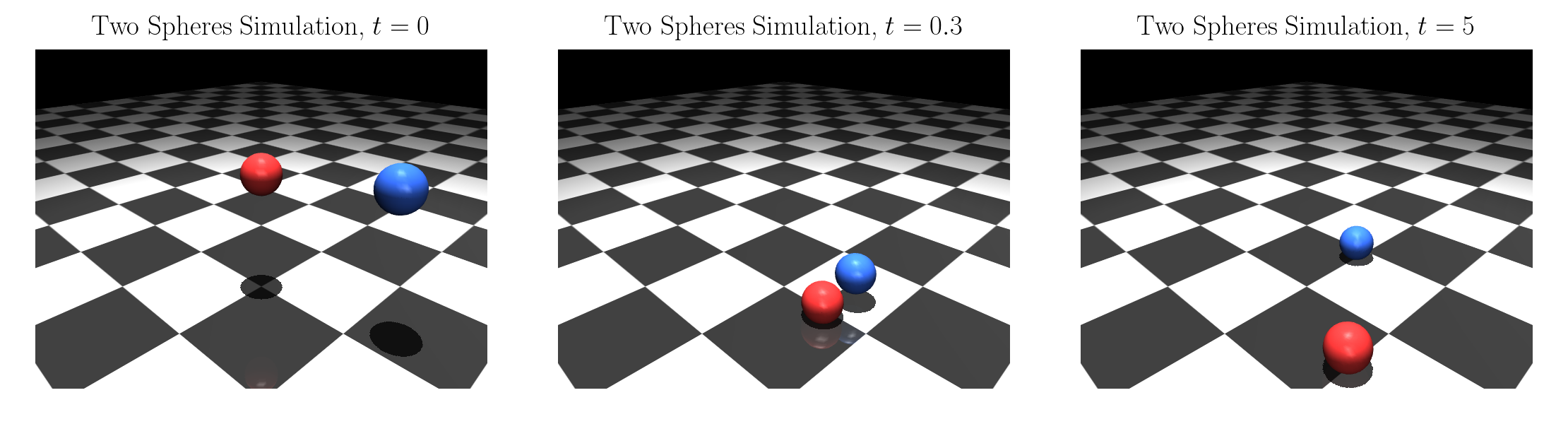}}
\caption{\textit{Two Spheres 3D Model}.
Three frames are displayed from a 3D model simulation with two spheres at times $t = 0$ (left), $t = 0.3$ (middle) and $t = 5$ (right).
The final distance from the origin of the red sphere (closest to viewer) is treated as our quantity of interest.
(A video is available on the GitHub repository.)
\alttext{A figure consisting of 3 panels, each of which contains a snapshot of the simulated system at a certain time $t$, where here $t \in \{0,0.3,5\}$.}
}
\label{fig:two-spheres}
\end{figure}

Our first case study consists of a simulation involving two spheres (one red, one blue), which initially collide, fall to the ground, and then eventually roll to a standstill; see \Cref{fig:two-spheres}. 
Numerical accuracy is controlled by three discretisation parameters $\mathbf{x}=(x_1,x_2,x_3)$: a \emph{time step} $x_1$, a \emph{solver reference} parameter $x_2$, and a \emph{solver impedance} parameter $x_3$, and we are interested in the $\mathbf{x} \rightarrow \mathbf{0}$ limit. 
The time step determines the temporal resolution at which the physical governing equations are discretised, the solver reference parameter specifies the time scale over which constraint violations, such as contact penetration, are corrected by the constraint solver, and the solver impedance parameter controls how constraint forces scale with the magnitude of the violation, effectively tuning the compliance of contacts and constraints.
From a numerical analysis perspective, these simulations are challenging because the dynamics involve contact-rich nonlinear interactions governed by a constraint-based solver; small changes in the time discretisation and solver regularisation parameters can lead to noticeable differences in the resulting trajectories. 
This sensitivity complicates stability analysis, convergence assessment, and the estimation of numerical error; data-driven statistical methods are needed.

For demonstration purposes, the final distance from the origin of the red sphere is treated as our scalar quantity $f(\mathbf{x})$ of interest.
Since the continuum limit is intractable, for assessment purposes we actually consider $f(\mathbf{x} + \mathbf{x}_0)$ for a small offset $\mathbf{x}_0$, so that with sufficient computational effort the limit $f(\mathbf{x}_0)$ can be computed.
Full implementation details are provided in \Cref{app:3d-models-detail}.

\paragraph{Predictive Accuracy of \ac{SPRE}}

\begin{figure}[t!]
\centering
\includegraphics[width=0.75\textwidth]{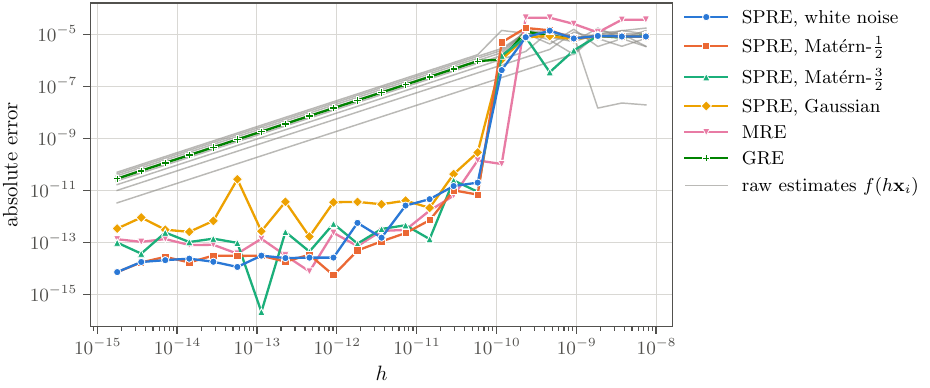}
\caption{\textit{Two Spheres 3D Model}.
Absolute errors of estimates for $f(\mathbf{0})$ as a function of the scaling factor $h$ controlling the design set $X_n^h$. 
The $n=2^d$ raw estimates $\{f(\mathbf{x}) : \mathbf{x} \in X_n^h \}$ are also displayed.
\alttext{A figure showing the absolute error of different extrapolation methods as the global scaling constant $h$ is varied.}
}
\label{fig:two-spheres-abs}
\end{figure}

First we investigated the predictive accuracy of \ac{SPRE} using a sequence of design sets $X_n^h$, obtained by scaling a fixed design set $X_n=\{\mathbf{x}_i\}_{i=1}^n$ of size $n = 2^d$ by a scalar $h\in(0,1]$. 
Full details are provided in \Cref{app:3d-models-detail}.
\Cref{fig:two-spheres-abs} shows the absolute error of \ac{SPRE} for four different choices of kernel $k$, together with the performance of \ac{MRE}, \ac{GRE}, and the raw estimates $\{f(h\mathbf{x}_i)\}_{i=1}^n$.
Several features are apparent:
\begin{enumerate}[(i)]
    \item For $h>10^{-10}$ none of the extrapolation methods improve upon the raw estimates, as the simulations are not yet in the asymptotic regime where convergence to a limit can be predicted. 
    \item For $h<10^{-10}$ all extrapolation methods outperform the raw estimates except \ac{GRE}, rapidly approaching the floating point accuracy limit of around $10^{-13}$. 
    This suggests the theoretical convergence acceleration results we presented in \Cref{subsec: converge accel} may still empirically hold when both the kernel $k$ and the index set $A$ are estimated.
    Among the kernels considered, the estimates produced by the white noise and Mat\'{e}rn-$\frac{1}{2}$ kernels are more accurate and robust.
    The failure of \ac{GRE} to accelerate convergence could be explained by its inability to exploit sparsity, which we argue below is present in this experiment.
    \item \ac{MRE} performs similarly to \ac{SPRE} when $h<10^{-10}$.  However, \ac{MRE} does not provide statistical uncertainty quantification and performs worse than the raw estimates when $h>10^{-10}$.
\end{enumerate}

\paragraph{Predictive Uncertainty of \ac{SPRE}}

\begin{figure}[t!]
\centerline{\includegraphics[width=\textwidth]{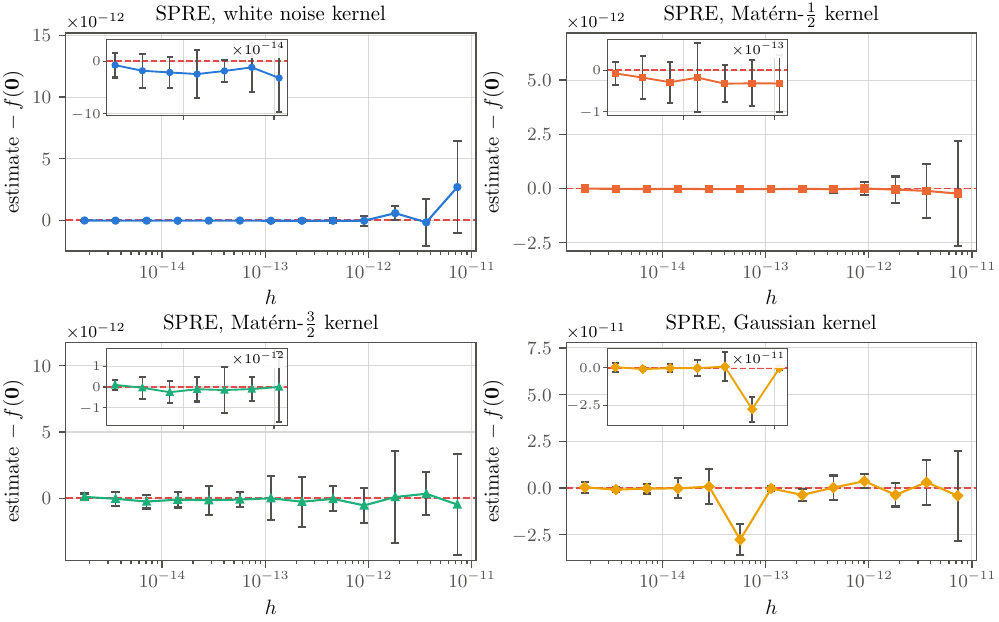}}
\caption{\textit{Two Spheres 3D Model}.
The two standard deviation predictive intervals from \ac{SPRE} are displayed as a function of the scaling factor $h$ controlling the design set $X_n^h$.
The intervals are displayed relative to the true value $f(\mathbf{0})$, which therefore corresponds to the dashed red line at zero; insets display the seven smallest values of $h$ on an expanded scale.
Top left to bottom right: white noise, Mat\'{e}rn-$\frac{1}{2}$, Mat\'{e}rn-$\frac{3}{2}$ and Gaussian kernels.
\alttext{A figure containing 4 panels, each of which shows the predictive error bars produced by SPRE as a function of the global scaling parameter $h$.
Each panel contains results for a different type of kernel.}
}
\label{fig:two-spheres-errorbars}
\end{figure}

To evaluate the uncertainty quantification provided by \ac{SPRE} when both $k$ and $A$ are estimated, predictive error bars for $f(\mathbf{0})$ are displayed in \Cref{fig:two-spheres-errorbars}. 
In this case all kernels produced predictive error bars that were reasonably well-calibrated, reducing in size to be commensurate with the size of the absolute error as $h$ is decreased.

\begin{figure}[t!]
    \centering
    \includegraphics[width=\textwidth]{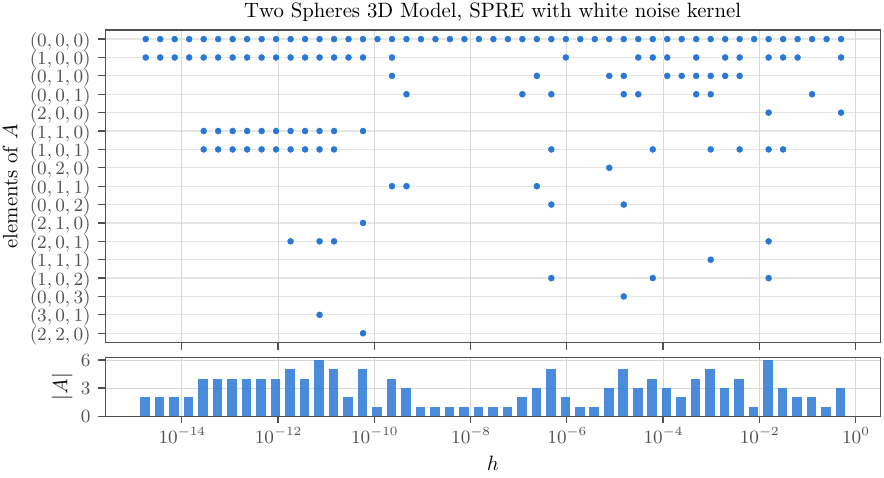}
    \caption{\emph{Two Spheres 3D Model}.  
    The index set $A$ was estimated based on a design of size $n = 2^d$, for a range of values of the scaling factor $h$ controlling the design set $X_n^h$.
    Here the white noise kernel was used.
    The lower panel displays the cardinality $|A|$ of the estimated index set.
    \alttext{A figure displaying the estimated index set $A$ as the global scaling parameter $h$ is varied.}
    }
    \label{fig: two spheres sparsity}
\end{figure}

\paragraph{Evidence of Extrapolation Sparsity}

To investigate whether extrapolation sparsity is present in this numerical simulation, we examined the behaviour of the estimated index set $A$ as the scaling factor $h$ is varied; results for the white noise kernel are shown in \Cref{fig: two spheres sparsity}.
For $h > 10^{-10}$, the estimated index set $A$ is not stable to variation in $h$, which we attribute to simulations not yet being in the asymptotic regime where convergence can be predicted. 
On the other hand, it can be seen that the elements of $A$ are stably estimated in the $h \rightarrow 0$ limit.
Moreover the cardinality of $A$ remains small in this limit, suggesting $d_{\mathrm{ext}}(f) \approx 2 - 4$.
The elements $(0,0,0)$ and $(1,0,0)$ are always selected, suggesting that error due to the time step $x_1$ is dominant, while the elements $(1,1,0)$ and $(1,0,1)$ are sometimes selected, indicating a possible interaction between the time step and the solver reference and solver impedance parameters. 
The same experiment was repeated for the Mat\'{e}rn-$\frac{1}{2}$ kernel (\Cref{fig: two spheres sparsity matern1/2}), the Mat\'{e}rn-$\frac{3}{2}$ kernel (\Cref{fig: two spheres sparsity matern3/2}), and the Gaussian kernel (\Cref{fig: two spheres sparsity gaussian}).
Similar evidence for extrapolation sparsity was observed for the Mat\'{e}rn-$\frac{1}{2}$ kernel, while for the smoother kernels no clear trend could be seen in the estimated index set $A$.
The latter may be explained by the additional regularity requirement associated with these smoother kernels not being satisfied in this experiment.

\paragraph{Sequential Experimental Design}

\begin{figure}[t!]
\centering
\includegraphics[width=0.95\textwidth]{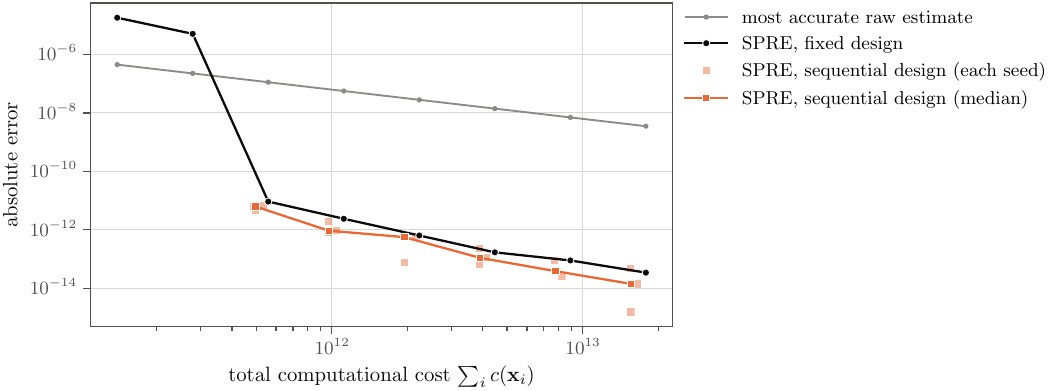}
\caption{\textit{Two Spheres 3D Model}.
Absolute error of \ac{SPRE} as a function of the total computational cost $\sum_i c(\mathbf{x}_i)$.
Results are shown for the fixed design $X_n^h$ (black) and for sequential experimental design subject to the same total budget (orange), for three realisations of the stochastic search algorithm used to select the new design set, together with their median.
The most accurate raw estimate in the fixed design is also displayed (grey).
Here the Mat\'{e}rn-$\frac{1}{2}$ kernel was used.
\alttext{A figure showing the absolute error of SPRE as a function of total computational cost, for a fixed design and for a sequential experimental design.}
}
\label{fig: sequential design}
\end{figure}

The results above were based on the fixed design $X_n^h$.
To assess whether the sequential experimental design methodology of \Cref{sec: UQ} can improve the cost-efficiency of \ac{SPRE}, we compared the fixed design $X_n^h$ with a sequential design at equal computational cost.
The results in \Cref{fig: sequential design} show that sequential design improved the accuracy of \ac{SPRE} at every budget considered, with the median absolute error reduced by a factor of between $1.1$ and $2.6$; equivalently, the accuracy of the fixed design was attained at roughly half of the computational cost.
The rate of convergence was similar for both designs, consistent with \Cref{cor: cost}, in which the rate is determined by the regularity of the residual rather than by the design.
Full details are contained in \Cref{app: sequential balls}.

\subsection{Five Shapes 3D Model}
\label{subsec: five shapes}

\begin{figure}[t!]
\centerline{\includegraphics[width=\textwidth]{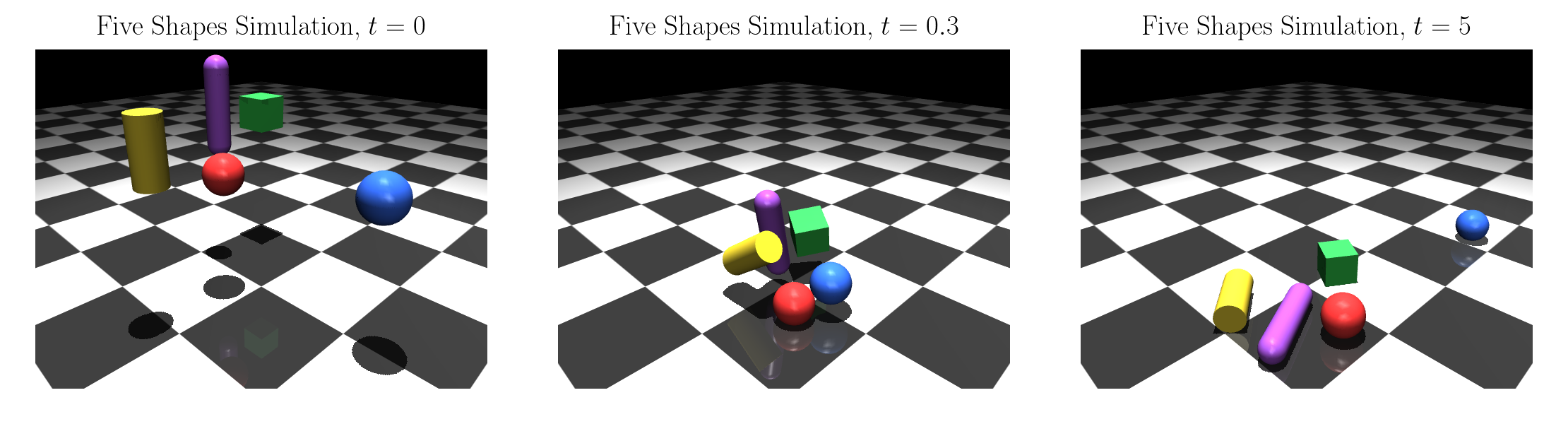}}
\caption{\textit{Five Shapes 3D Model}.
Three frames are displayed from a 3D model simulation with five shapes at times $t = 0$ (left), $t = 0.3$ (middle) and $t = 5$ (right).
The final distance from the origin of the red sphere (closest to viewer) is treated as our quantity of interest.
(A video is available on the GitHub repository.)
\alttext{A figure consisting of 3 panels, each of which contains a snapshot of the simulated system at a certain time $t$, where here $t \in \{0,0.3,5\}$.}
}
\label{fig:five-shapes}
\end{figure}

Next we consider a more challenging numerical simulation using a similar setup to \Cref{sec:two-spheres}, but with five interacting bodies: two spheres, a cylinder, a cube, and a capsule.
A typical simulation is shown in \Cref{fig:five-shapes}.
The shapes collide with one another, fall to the ground, and eventually come to rest. 
As before, the final distance from the origin of the red sphere is used as our scalar quantity of interest.
Full implementation details are provided in \Cref{app:3d-models-detail}.

\paragraph{Predictive Accuracy of \ac{SPRE}}

The same experimental protocol was adopted with a design of size $n = 2^d$, with the predictive accuracy of each method is displayed in \Cref{fig:five-shapes-abs}. 
Two main differences compared with the two-sphere simulation are apparent:
\begin{enumerate}[(i)]
\item Smaller tolerances are needed to reach the asymptotic regime where convergence becomes predictable, here around $h < 10^{-12}$.
For $h<10^{-12}$, all extrapolation methods except \ac{GRE} outperform the raw estimates, reaching accuracies around $10^{-9}$. Although this is less precise than in the two-sphere sphere, the improvement relative to the raw estimates remains clear.
\item There is a smaller window of $h$ values in which convergence can be accelerated before the floating point accuracy limit is reached.
\end{enumerate}
Aside from these differences, the results are broadly consistent with the two-sphere experiment; \ac{SPRE} with either the white noise or Mat\'{e}rn-$\frac{1}{2}$ kernels generally performed best.

\begin{figure}[t!]
\centering 
\includegraphics[width=0.95\textwidth]{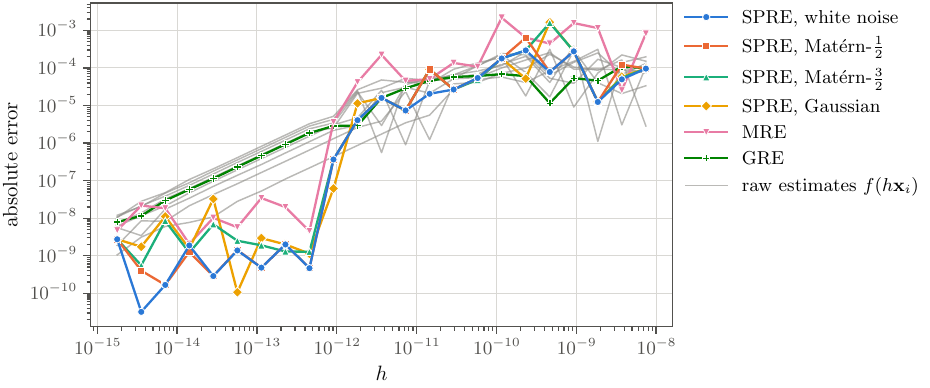}
\caption{\textit{Five Shapes 3D Model}.
Absolute errors of estimates for $f(\mathbf{0})$ as a function of the scaling factor $h$ controlling the design set $X_n^h$. 
The $n=2^d$ raw estimates $\{f(\mathbf{x}) : \mathbf{x} \in X_n^h \}$ are also displayed.
\alttext{A figure showing the absolute error of different extrapolation methods as the global scaling constant $h$ is varied.}
}
\label{fig:five-shapes-abs}
\end{figure}

\paragraph{Predictive Uncertainty of \ac{SPRE}}

Predictive error bars for $f(\mathbf{0})$ are shown in \Cref{fig:five-shapes-errorbars} and again appear to be well-calibrated. 
Compared to the two-spheres experiment, the predictive intervals no longer decrease substantially for $h<10^{-12}$, indicating that the attainable numerical accuracy has been reached.

\paragraph{Evidence of Extrapolation Sparsity}
Again we examined the behaviour of the estimated index set $A$ as the scaling factor $h$ is varied; results for the white noise kernel are shown in \Cref{fig: five shapes sparsity}.
As for the two-spheres experiment the elements $(0,0,0)$ and $(1,0,0)$ are always present in the estimated $A$ when $h$ is sufficiently small, while occasionally additional basis elements such as $(2,0,0)$ are added. 
Thus again we find evidence for an extrapolation dimension $d_{\mathrm{ext}}(f) \approx 2 - 4$ in this experiment.
Results for the Mat\'{e}rn-$\frac{1}{2}$ kernel (\Cref{fig: five shapes sparsity matern1/2}), the Mat\'{e}rn-$\frac{3}{2}$ kernel (\Cref{fig: five shapes sparsity matern3/2}), and the Gaussian kernel (\Cref{fig: five shapes sparsity gaussian}), showed higher variance among the the estimated elements of $A$, which we attribute to the more challenging nature of this task, in which multiple collisions are simulated.

\subsection{Agent-Based Flocking Model}
\label{subsec: agent}

\begin{figure}[t!]
\centerline{\includegraphics[width=\textwidth]{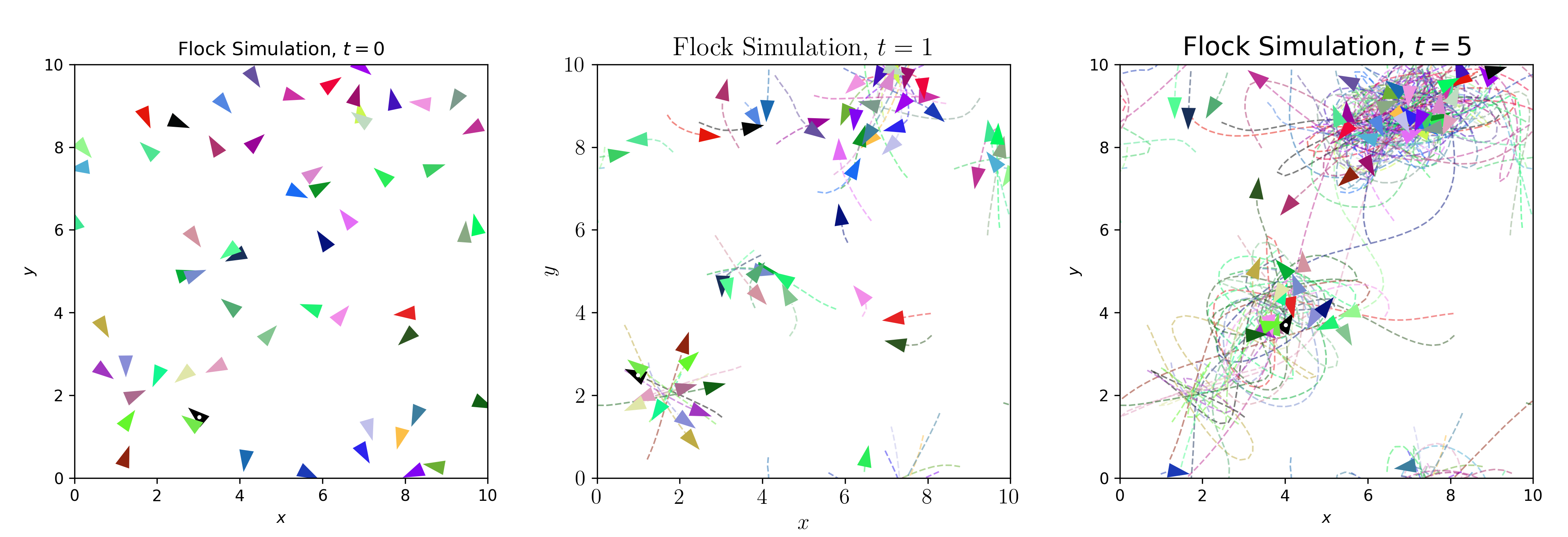}}
\caption{\textit{Agent-Based Flocking Model}.
Three frames are displayed from a multi-agent flock simulation at times $t = 0$ (left), $t = 1$ (middle) and $t = 5$ (right).
The final distance from the origin of the black agent with a white dot is treated as our quantity of interest.
(A video is available on the GitHub repository.)
\alttext{A figure consisting of 3 panels, each of which contains a snapshot of the simulated system at a certain time $t$, where here $t \in \{0,1,5\}$.}
}
\label{fig:flock}
\end{figure}

Our final case study represents a step-change in complexity compared to the first two case studies, where now 60 interacting agents move in a two-dimensional periodic domain to mimic the flocking behaviour of birds in the wild. 
Each agent updates its position and velocity based on pairwise interactions with neighbouring agents, which include short-range repulsion to prevent overlap and longer-range attraction to promote group cohesion. 
A typical simulation is shown in \Cref{fig:flock}.
Such is the complexity of this simulation that it is \emph{a priori} unclear if \emph{any} extrapolation methods can work.
Indeed, from a numerical analysis perspective, the strongly nonlinear, coupled interactions among many agents can produce emergent, sensitive-to-initial-conditions dynamics that render stability, convergence, and error estimation of numerical schemes extremely difficult.

For this experiment, simulation accuracy is controlled by three discretisation parameters; a \emph{time step} $x_1$, a \emph{repulsion softening} parameter $x_2$, and a \emph{cutoff width} $x_3$.
The repulsion softening prevents singular repulsive forces at very small distances, while the cutoff width ensures interaction forces decay smoothly near the interaction radius; both are needed for numerical stability to compensate for using a positive time step $x_1$, while an exact simulation corresponds to a limit in which $\mathbf{x} \rightarrow \mathbf{0}$.
For demonstration purposes, the final distance from the origin of the black agent with a white dot in \Cref{fig:flock} is treated as our scalar quantity of interest.
Full implementation details are provided in \Cref{app:3d-models-detail}.

\paragraph{Predictive Accuracy of \ac{SPRE}}

\begin{figure}[t!]
\centering
\includegraphics[width=0.95\textwidth]{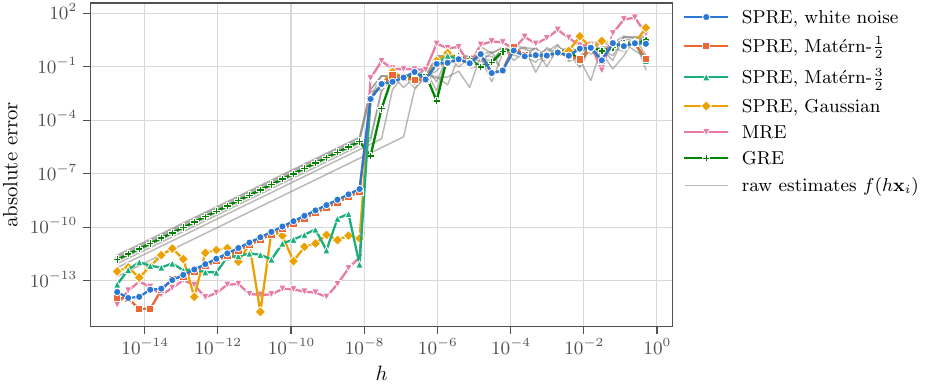}
\caption{\textit{Agent-Based Flocking Model}.
Absolute errors of estimates for $f(\mathbf{0})$ as a function of the scaling factor $h$ controlling the design set $X_n^h$. 
The $n=2^d$ raw estimates $\{f(\mathbf{x}) : \mathbf{x} \in X_n^h \}$ are also displayed.
\alttext{A figure showing the absolute error of different extrapolation methods as the global scaling constant $h$ is varied.}
}
\label{fig:flock-abs}
\end{figure}

Absolute errors are displayed in \Cref{fig:flock-abs}.
Remarkably, \emph{all instances of \ac{SPRE} were able to out-perform the raw estimates} when $h$ is sufficiently small, as was \ac{MRE}; we consider this to be a surprising and highly encouraging result.
Similarly to the five shapes case study, there is a narrow window of $h$ values (here, around $h \approx 10^{-8}$) in which convergence is rapidly accelerated.
However, in contrast to the two earlier case studies (and in contrast to \ac{MRE} in this experiment), the accuracy of \ac{SPRE} continues to improve as $h$ is further decreased, albeit at the same rate as the raw estimates on which the extrapolation estimates are based.
Though the mechanism for this is difficult to isolate, we suspect floating point precision is preventing further convergence acceleration for $h < 10^{-8}$ with \ac{SPRE}.
Nevertheless, the estimates provided by \ac{SPRE} are two orders of magnitude more accurate than simply using the raw estimates to approximate $f(\mathbf{0})$.
Further, the estimates from \ac{SPRE} are seen to be robust to variations in $h$ when either the white noise or Mat\'{e}rn-$\frac{1}{2}$ kernel are used.

\paragraph{Predictive Uncertainty of \ac{SPRE}}

Predictive error bars are shown in \Cref{fig:flock-errorbars}, with results similar to those reported for the two spheres experiment.
That is, the predictive error bars appear to be well-calibrated and decrease in a manner commensurate with the size of the actual prediction error in this experiment.

\paragraph{Evidence of Extrapolation Sparsity}

\begin{figure}[t!]
    \centering
    \includegraphics[width=\textwidth]{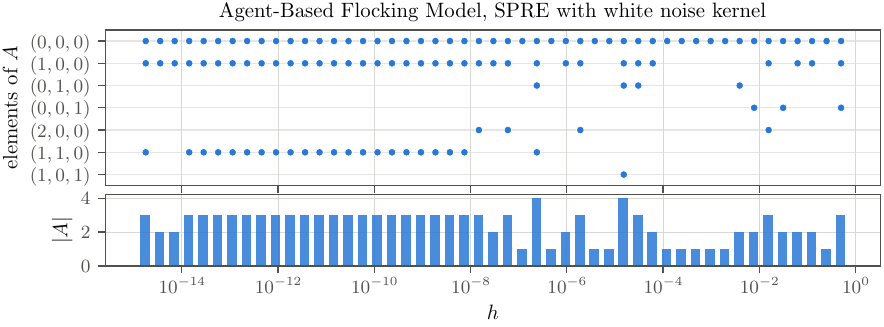}
    \caption{\emph{Agent-Based Flocking Model}.  
    The index set $A$ was estimated based on a design of size $n = 2^d$, for a range of values of the scaling factor $h$ controlling the design set $X_n^h$.
    Here the white noise kernel was used.
    The lower panel displays the cardinality $|A|$ of the estimated index set.
    \alttext{A figure displaying the estimated index set $A$ as the global scaling parameter $h$ is varied.}
    }
    \label{fig: flocking sparsity}
\end{figure}

For the white noise kernel, the estimated index set $A$ typically contains $(0,0,0)$, $(1,0,0)$ and $(1,1,0)$ when $h < 10^{-8}$ (cf. \Cref{fig: flocking sparsity}), suggesting a first order expansion in the time step $x_1$ and repulsion softening $x_2$, together with an interaction effect, so that the extrapolation dimension is $d_{\mathrm{ext}}(f) \approx 3$. 
Results for other kernels  are contained in \Cref{fig: flocking sparsity matern1/2,fig: flocking sparsity matern3/2,fig: flocking sparsity gaussian} of \Cref{app: additional}.
Similarly to the first two case studies, the index set $A$ is robustly estimated only when the white noise or Mat\'{e}rn-$\frac{1}{2}$ kernels are used.

\section{Discussion}
\label{sec: discussion}

Accelerating complex simulation is a fundamental computational task, and while statistical methodologies for surrogate modelling and multi-fidelity modelling are well-developed, the potential for convergence acceleration via statistical approaches to extrapolation has been under-explored.
A first step was recently taken in \citet{oates2024probabilistic}, but the \ac{GRE} method proposed in that work suffered from a prohibitive data requirement.
To move forward, we introduced the notion of extrapolation sparsity, and empirical evidence was presented that suggests extrapolation sparsity is typical for complex simulators, representing an exciting opportunity for methodological development.
To exploit this opportunity we developed \ac{SPRE}, which harnesses extrapolation sparsity to accelerate convergence of complex simulators within a principled statistical framework, and with an exponentially reduced data requirement.

\paragraph{Related Work}

Classical extrapolation methods can be classified according to the basis functions used to construct an interpolant; the original polynomial basis functions of \citet{richardson1911ix,lf1927deferred} can be contrasted with the stability-promoting rational functions of \citep{thiele1909interpolationsrechnung,bulirsch1964fehlerabschatzungen,larkin1967some}, as well as a range of other choices documented in \citet{brezinski2013extrapolation}.
However, guidance on selecting a basis is limited, due to the absence of formal approaches to model selection that would be available within a statistical framework.
Despite this limitation, classical extrapolation methods continue to find new applications, including in optimal transport \citep{chizat2020faster}, training of machine learning models \citep{bach2021effectiveness}, and sampling with Markov chain Monte Carlo \citep{durmus2016stochastic}.

\paragraph{Future Research Directions}

Our contribution can be cast as a \emph{black box probabilistic numerical method} in the sense of \citet{teymur2021black}, with Bayesian principles used to quantify uncertainty regarding the continuum model of interest \citep{cockayne2019bayesian}.
As such, it would be interesting to compare \ac{SPRE} to other probabilistic numerical methods across a range of foundational numerical tasks. 
Probabilistic numerical methods are typically designed from the ground-up using probabilistic primitives \citep{hennig2015probabilistic}, so it would be interesting to understand if such effort can be circumvented by instead extrapolating output from established numerical methods in the \ac{SPRE} framework.

For our purposes the simulator $f$ was assumed to be deterministic, but stochastic simulators are also widely used.
An important numerical task is then to approximate moments of the simulator output, in particular the expected value $\mathbb{E}[f(\mathbf{0})]$.
Unbiased estimation of $\mathbb{E}[f(\mathbf{0})]$ can be possible using the methodology of \citet{rhee2015unbiased} applied to a sequence $\mathbb{E}[f(\mathbf{x}_i)]$ with $\mathbf{x}_i \rightarrow \mathbf{0}$, while multilevel methods have been combined with Richardson extrapolation in \citet{lemaire2017multilevel,beschle2023quasi}.
As it stands, the techniques used to accelerate convergence of deterministic and stochastic simulators are different; developing a unified framework for convergence acceleration would be an exciting direction for further work.

Although we set out a framework for experimental design with \ac{SPRE}, its theoretical analysis appears to be difficult.
\citet{adcock2025function} studied polynomial extrapolation with noisy data, where the $\{\mathbf{x}_i\}$ were sampled independently from a measure $\mu$, asking how $\mu$ ought to be chosen to promote extrapolation accuracy subject to a constraint of the form $\mathbb{E}[c(\mathbf{x_i})] \leq C$.
Extending this analysis to the case of noise-free data and \ac{GP} regression would be an interesting avenue for future work.

\singlespacing

\paragraph{Acknowledgments}

The authors thank the anonymous Reviewers and Editor for their constructive feedback, and Mike Giles and Aretha Teckentrup for discussions which motivated and improved this work.
CJO was supported by the Engineering and Physical Sciences Research Council EP/W019590/1 and a Philip Leverhulme Prize PLP-2023-004.
TK was supported by the Research Council of Finland grants 359183 (``Flagship of Advanced Mathematics for Sensing, Imaging and Modelling'') and 368086 (``Inference and approximation under misspecification'').
TK acknowledges the research environment provided by ELLIS Institute Finland.

\bibliographystyle{abbrvnat}
\bibliography{bibliography}

\newpage
\appendix 
\counterwithin{figure}{section}

\doublespacing

\section*{Supplementary Material}

These appendices are electronic supplementary material for the manuscript \emph{Sparse Probabilistic Richardson Extrapolation} by Oates, Howey, and Karvonen.
\Cref{app: proofs} contains the proofs for all theoretical results reported in the main text.
Protocols used to perform the experiments are described in \Cref{app: protocol}.
Additional experimental results are contained in \Cref{app: additional}.

\section{On the Extrapolation Sparsity Assumption}
\label{app: discuss sparsity}

In discussing when extrapolation sparsity can be expected to hold, it is helpful to separate two questions: (a) whether an expansion of the form \eqref{eq: error} exists, and (b) if so, whether the index set $A$ can be taken to be small.

On (b), suppose that \eqref{eq: error} holds with an index set whose elements have total degree at most $s$, i.e. $A \subseteq \{\bm{\alpha} : |\bm{\alpha}| \leq s\}$.
Then
\begin{align*}
d_{\mathrm{ext}}(f) \leq \binom{d+s}{s} = \prod_{i=1}^s \frac{d+i}{i} \leq (d+1)^s ,
\end{align*}
which is polynomial in $d$ for fixed $s$.
Thus extrapolation sparsity, in the sense of a polynomial dependence of $d_{\mathrm{ext}}(f)$ on $d$, can fail only if the order $s$ of the required expansion grows with $d$; for example, if $f_d(\mathbf{x}) = \sum_{|\bm{\alpha}| = s(d)} \beta_{\bm{\alpha}} \mathbf{x}^{\bm{\alpha}}$ with all $\beta_{\bm{\alpha}} \neq 0$ and $s(d) \geq d / \log(d)$, then $d_{\mathrm{ext}}(f_d)$ grows super-polynomially in $d$.
However, the orders of the individual error contributions are properties of the numerical methods used within the simulator, such as the order of a time-stepping scheme, and not of the number $d$ of discretisation parameters, which is why we regard such instances as pathological.
In fact, sparsity is often stronger than this bound suggests: when the errors due to distinct sources of discretisation are approximately additive, the leading-order terms in \eqref{eq: error} involve few interactions and $d_{\mathrm{ext}}(f)$ grows only linearly in $d$.

On (a), for simple numerical methods expansions of the form \eqref{eq: error} can be rigorously established; examples include the cubature method of \Cref{subsec: cubature}, one-step methods for ordinary differential equations with smooth vector fields \citep{gragg1965extrapolation}, \citep[Section II.8]{hairer1993solving}, and finite element methods for elliptic boundary value problems with sufficiently smooth solutions \citep[][Chapter 5, Section 3]{liem1995splitting}.
Equally, the numerical analysis literature identifies features that can cause such expansions to fail, or to hold only in a small neighbourhood of $\mathbf{0}$, such as singularities in the solution of a differential equation, discontinuous dependence on discretisation parameters, and chaotic dynamics.
However, these simple settings are \emph{not} our main motivation.
Our interest is in complex simulators, such as those considered in \Cref{sec: experiments}, which combine multiple interacting numerical methods, nonlinear and non-smooth dynamics, and \textit{ad hoc} regularisation, and for which no numerical analysis can realistically be performed.

\section{Proofs}
\label{app: proofs}

This appendix contains proofs for all theoretical results reported in the main text.
\Cref{app: prelim} is dedicated to stating the definitions and technical lemmas that will be used.
The proof of \Cref{thm: SEM} is contained in \Cref{app: SEM}, the proof of \Cref{thm: convergence SPRE} in \Cref{subsec: SPRE converge proof}, the proof of \Cref{thm: extra smooth} in \Cref{app: extra smooth}, and the proof of \Cref{lem: equiv norms} in \Cref{app: proof of equiv norms}.

\subsection{Definitions and Technical Lemmas}
\label{app: prelim}

The purpose of this appendix is to recall certain mathematical properties of polynomial approximation that we will use in the proof of \Cref{thm: SEM}.
In addition, we will derive sparsity-aware versions of particular results on \emph{polynomial reproduction} from \citet{wendland2004scattered}, as these will be required to prove \Cref{thm: convergence SPRE}.
Lastly we recall the technical concept of a \emph{native space} in \Cref{app: native}.

\paragraph{Notation}
For compact $\mathcal{X}$, let $C(\mathcal{X})$ denote the set of continuous functions on $\mathcal{X}$, and let the supremum norm for each $f \in C(\mathcal{X})$ be defined as
$$
\|f\|_{\infty,\mathcal{X}} = \sup_{\mathbf{x} \in \mathcal{X}} |f(\mathbf{x})| .
$$

\subsubsection{Lebesgue Functions and Polynomial Approximation}

The aim of this appendix is to clarify how well one can approximate a continuous function with an element of $\mathcal{P}_A$.
Recall the definition of the Vandermonde matrix $\mathbf{V}_A(X_n)$ in \Cref{subsec: notation}, an $n \times m$ matrix where $m$ is the cardinality of the index set $A$.
For the remainder of this appendix we restrict attention to the case where $m$ and $n$ are equal and $X_n$ is $\mathcal{P}_A$-unisolvent. 
Then, the \emph{Vandermonde determinant} is denoted $\mathrm{VDM}_A(X_n) = \mathrm{det}(\mathbf{V}_A(X_n))$.
The (multivariate) \emph{Lagrange polynomials} $\ell_i^A(\cdot  ; X_n)$ are defined as the elements of $\mathcal{P}_A$ for which $\ell_i^A(\mathbf{x}_j ; X_n) = \delta_{i,j}$.
It is straightforward to check that
\begin{align*}
\ell_i^A(\mathbf{x} ; X_n) = \frac{\mathrm{VDM}_A(\{\mathbf{x}\} \cup (X_n \setminus \{\mathbf{x}_i\}))}{\mathrm{VDM}_A(X_n)} .
\end{align*}
The (multivariate) \emph{Lebesgue function} measures the sum of absolute values of the Lagrange polynomials defined by $X_n = \{\mathbf{x}_i\}_{i=1}^n \subset \mathcal{X}$:
\begin{align}
\lambda_A(\mathbf{x} ; X_n) = \sum_{i=1}^n | \ell_i^A(\mathbf{x} ; X_n) |  \label{eq: mv leb fn}
\end{align}
for each $\mathbf{x} \in \mathcal{X}$.
For compact $\mathcal{X}$, let
\begin{align*}
\Pi_{A,\mathcal{X}} : (C(\mathcal{X}) , \|\cdot\|_{\infty,\mathcal{X}}) & \rightarrow \mathbb{R} \\
f & \mapsto f_n(\mathbf{0})
\end{align*}
where $f_n$ denotes the interpolating polynomial (in $\mathcal{P}_A$) of $f$ on $X_n$.
Note that for each $f \in C(\mathcal{X})$ we can bound $\Pi_{A,\mathcal{X}}(f)$ using the Lebesgue function:
\begin{align}
    |\Pi_{A,\mathcal{X}}(f)| & = \left| \sum_{i=1}^n \ell_i^A(\cdot ; X_n) f(\mathbf{x}_i) \right| 
    \leq \lambda_A(\mathbf{0}; X_n) \|f\|_{\infty,\mathcal{X}}  \label{eq: bound Pi}
\end{align}
Then we have the following technical result, which will be used in the proof of \Cref{thm: SEM}:

\begin{lemma}[Sparse polynomial approximation]
\label{thm: leb}
Let $X_n$ be $\mathcal{P}_A$-unisolvent with $n = |A|$.
Let $f \in C(\mathcal{X})$ where $\mathbf{0} \in \mathcal{X} \subset [0,\infty)^d$ is compact, and let $f_n$ denote the interpolating polynomial (in $\mathcal{P}_A$) of $f$ on $X_n$.
Then 
\begin{align}
| f(\mathbf{0}) - f_n(\mathbf{0}) | \leq (1 + \lambda_{A}(\mathbf{0} ; X_n)) \inf_{p \in \mathcal{P}_A} \| f - p \|_{\infty , \mathcal{X}} . \label{eq: best poly approx}
\end{align}
\end{lemma}
\begin{proof}
Let $p_\star$ attain the infimum on the right hand side of \eqref{eq: best poly approx}, which exists and is unique since $X_n$ is $\mathcal{P}_A$-unisolvent with $n = |A|$.
Then by the triangle inequality, and the fact that $\Pi_{A,\mathcal{X}}$ is a projection with operator norm at most $\lambda_A(\mathbf{0}; X_n)$ from \eqref{eq: bound Pi},
\begin{align*}
| f(\mathbf{0}) - f_n(\mathbf{0}) | & \leq | f(\mathbf{0}) - p_\star(\mathbf{0}) | + | p_\star(\mathbf{0}) - f_n(\mathbf{0}) |\\
& =  | f(\mathbf{0}) - p_\star(\mathbf{0}) | + | \Pi_{A,\mathcal{X}}( p_\star ) - \Pi_{A,\mathcal{X}} ( f ) | \\
& =  | f(\mathbf{0}) - p_\star(\mathbf{0}) | + | \Pi_{A,\mathcal{X}}( p_\star - f ) | \\
& \leq  \| p_\star - f \|_{\infty , \mathcal{X}} + \lambda_A(\mathbf{0};X_n)  \| p_\star - f \|_{\infty , \mathcal{X}}
\end{align*}
as required.
\end{proof}

\subsubsection{Sparse Polynomial Reproduction}
\label{app: sparse repro}

The aim of this appendix is to generalise results on polynomial reproduction from \citet{wendland2004scattered}; these will be required to prove \Cref{thm: convergence SPRE}.
The need for generalisation comes from two directions; we use a sparse polynomial basis whereas a dense polynomial basis is considered in \citet{wendland2004scattered}, and we remove the requirement in \citet{wendland2004scattered} for $\mathcal{X}$ to be an open set (this was used by Wendland to ensure derivatives are well-defined, but we do not have that requirement).

\begin{lemma}[Polynomial reproduction at $\mathbf{0}$] 
\label{lem: poly repro}
    Suppose that $X_n = \{\mathbf{x}_i\}_{i=1}^n \subset \mathcal{X}$ is $\mathcal{P}_A$-unisolvent where $\mathbf{0} \in \mathcal{X} \subset [0,\infty)^d$ is compact.
    Then for all $h \in (0,1]$ there exists a vector of coefficients $\tilde{\mathbf{u}}^{(h)} \in \mathbb{R}^n$ and a $h$-independent constant $C_{A,X_n} > 0$ such that 
    \begin{enumerate}
        \item $\sum_{i=1}^n \tilde{u}_i^{(h)} p(h \mathbf{x}_i) = p(\mathbf{0})$ for all $p \in \mathcal{P}_A$ (sparse polynomial reproduction)
        \item $\sum_{i=1}^n |\tilde{u}_i^{(h)}| \leq C_{A,X_n}$.
    \end{enumerate}
\end{lemma}
\begin{proof}
    To begin with we fix $h \in (0,1]$.
    Consider the \emph{sampling operator}
    \begin{align*}
        T : (C(\mathcal{X}) , \|\cdot\|_{\infty,\mathcal{X}} ) & \rightarrow \mathbb{R}^n \\
        f & \mapsto (f(h\mathbf{x}_1) , \dots , f(h\mathbf{x}_n)) ,
    \end{align*}
    which is a bounded linear functional, and its restriction $T|_{\mathcal{P}_A} : \mathcal{P}_A \rightarrow T(\mathcal{P}_A)$ is invertible since $X_n^h$ is $\mathcal{P}_A$-unisolvent.
    In particular $T|_{\mathcal{P}_A}^{-1}$ is also a bounded linear functional.
    It follows that $L_{X_n^h} : T(\mathcal{P}_A) \rightarrow \mathbb{R}$, defined via $L_{X_n^h}(\mathbf{y}) = (T|_{\mathcal{P}_A}^{-1}(\mathbf{y}))(\mathbf{0})$, is a bounded linear functional on $T(\mathcal{P}_A)$.
    This functional $L_{X_n^h}$ therefore has a well-defined operator norm
    \begin{align*}
        \| L_{X_n^h} \|_1 = \sup_{ \mathbf{0} \neq \mathbf{y} \in T(\mathcal{P}_A) } \frac{ |L_{X_n^h}(\mathbf{y})| }{ \|\mathbf{y}\|_\infty } = \sup_{0 \neq p \in \mathcal{P}_A} \frac{ |p(\mathbf{0})| }{ \| (p(h \mathbf{x}_1) , \dots , p(h \mathbf{x}_n)) \|_\infty } .
    \end{align*}
    On the other hand, since $L_{X_n^h}$ is a bounded linear functional on $T(\mathcal{P}_A) \subset \mathbb{R}^n$, from the Hahn--Banach theorem there exists a norm-preserving extension $L_{X_n^h} : \mathbb{R}^n \rightarrow \mathbb{R}$ of the form $L_{X_n^h}(y) = \sum_{i=1}^n \tilde{u}_i^{(h)} y_i$ for some coefficient vector $\tilde{\mathbf{u}}^{(h)} \in \mathbb{R}^n$.
    The operator norm of this function is $\| L_{X_n^h} \|_1 = \sum_{i=1}^n |\tilde{u}_i^{(h)}|$, and we call this constant $C_{A,X_n^h}$.    
    
    To conclude the argument we will show that $C_{A,X_n^h}$ is in fact $h$-independent.
    To see this, let $p(\mathbf{x}) = \sum_{\bm{\alpha} \in A} \beta_{\bm{\alpha}} \mathbf{x}^{\bm{\alpha}}$ parametrise a generic element of $\mathcal{P}_A$.
    Then the operator norm can be expressed as
    \begin{align*}
        \| L_{X_n^h} \|_1 = \sup_{\bm{0} \neq \bm{\beta} \in \mathbb{R}^m } \frac{ |\beta_{\mathbf{0}}| }{ \max \{ | \beta_{\mathbf{0}} + \sum_{ \bm{\alpha} \in A \setminus \{\mathbf{0}\} } \beta_{\bm{\alpha}} ( h \mathbf{x}_i )^{\bm{\alpha}} | \, : \, i = 1,\dots,n \} }    %
    \end{align*}
    which is indeed $h$-independent (as can be seen from the change of variables $\tilde{\beta}_{\bm{\alpha}} = h^{|\bm{\alpha}|} \beta_{\bm{\alpha}}$).
\end{proof}

The content of \Cref{lem: poly repro} is closely related to the topic of \emph{norming sets}, where one seeks $X_n$ such that $\|p\|_{\infty,\mathcal{X}}$ is controlled by $\|p\|_{\infty , X_n}$; see \citet{kroo2011optimal}.
For our purposes, we need only control $|p(\mathbf{0})|$ by $\|p\|_{\infty , X_n}$.

\subsubsection{Native Spaces}
\label{app: native}

This appendix introduces the concept of a \emph{native space}, a space of functions that are formed as sums of elements from $\mathcal{P}_A$ and $\mathcal{H}_k(\mathcal{X})$.
(Recall that $\mathcal{H}_k(\mathcal{X})$ denotes the \ac{RKHS} of real-valued functions on $\mathcal{X}$ for which $k$ is the reproducing ; \citealp[see][for background]{berlinet2011reproducing}.)
The standard construction, for example as described in Chapter 10 of \citet{wendland2004scattered}, is unnecessarily technical for our purposes, and we therefore present a more explicit definition of the native space and derive the properties of the native space that will be required in our work.

Let $\mathcal{N}_k(\mathcal{X})$ denote the vector space of functions $f : \mathcal{X} \rightarrow \mathbb{R}$ of the form $f = p + h$ with $p \in \mathcal{P}_A$ and $h \in \mathcal{H}_k(\mathcal{X})$.
This space can be equipped with a semi-norm
\begin{align*}
    |f|_{\mathcal{N}_k(\mathcal{X})} := \inf \{ \|h\|_{\mathcal{H}_k(\mathcal{X})} : f = p + h , p \in \mathcal{P}_A , h \in \mathcal{H}_k(\mathcal{X}) \}
\end{align*}
and we can define the projection $\Pi_{\mathcal{P}_A} : \mathcal{N}_k(\mathcal{X}) \rightarrow \mathcal{P}_A$ onto $\mathcal{P}_A$ such that $|\Pi_{\mathcal{P}_A}(f)|_{\mathcal{N}_k(\mathcal{X})} = 0$ and $|(I - \Pi_{\mathcal{P}_A})(f)|_{\mathcal{N}_k(\mathcal{X})} = |f|_{\mathcal{N}_k(\mathcal{X})}$ for all $f \in \mathcal{N}_k(\mathcal{X})$.
Thus we have a direct sum $\mathcal{N}_k(\mathcal{X}) = \mathcal{P}_A \oplus (I - \Pi_{\mathcal{P}_A}) \mathcal{H}_k(\mathcal{X})$.

For $f \in \mathcal{N}_k(\mathcal{X})$, let $h_f$ denote the element of $h \in \mathcal{H}_k(\mathcal{X})$ which minimises $\|h\|_{\mathcal{H}_k(\mathcal{X})}$ subject to $f = p + h$ for some $p \in \mathcal{P}_A$.
Then in particular $(I - \Pi_{\mathcal{P}_A}) f = (I - \Pi_{\mathcal{P}_A}) h_f$ and $|f|_{\mathcal{N}_k(\mathcal{X})} = \|h_f\|_{\mathcal{H}_k(\mathcal{X})}$.
Using the direct sum representation of $\mathcal{N}_k(\mathcal{X})$, and then using the reproducing property of $\mathcal{H}_k(\mathcal{X})$, we can write
\begin{align}
    f(\mathbf{x}) & = (\Pi_{\mathcal{P}_A} f)(\mathbf{x}) + ((I - \Pi_{\mathcal{P}_A})(f))(\mathbf{x}) \nonumber \\
    & = (\Pi_{\mathcal{P}_A} f)(\mathbf{x}) + ((I - \Pi_{\mathcal{P}_A})(h_f))(\mathbf{x}) \nonumber \\
    & = (\Pi_{\mathcal{P}_A} f)(\mathbf{x}) + (I - \Pi_{\mathcal{P}_A})_{\mathbf{x}} \langle h_f , k(\cdot,\mathbf{x}) \rangle_{\mathcal{H}_k(\mathcal{X})} \nonumber \\
    & = (\Pi_{\mathcal{P}_A} f)(\mathbf{x}) +  \langle h_f ,  (I - \Pi_{\mathcal{P}_A})_{\mathbf{x}} k(\cdot,\mathbf{x}) \rangle_{\mathcal{H}_k(\mathcal{X})}  \label{eq: native rep prop}
\end{align}
where $(I - \Pi_{\mathcal{P}_A})_{\mathbf{x}} k(\cdot , \mathbf{x})$ denotes the action of $I - \Pi_{\mathcal{P}_A}$ on the $\mathbf{x}$-argument.
Equation \eqref{eq: native rep prop} can be viewed as a generalisation of the usual reproducing property for functions $f$ in the native space $\mathcal{N}_k(\mathcal{X})$, which accounts for the presence of a polynomial basis $\mathcal{P}_A$ in addition to the usual reproducing kernel Hilbert space $\mathcal{H}_k(\mathcal{X})$.
These mathematical properties will be used in the proof of \Cref{thm: convergence SPRE}.

\subsection{Proof of \Cref{thm: SEM}}
\label{app: SEM}

The following result establishes convergence acceleration for \ac{MRE}.

\begin{proof}[Proof of \Cref{thm: SEM}]
First note that $f_n^h$ exists and is unique by the assumption that $X_n$ is $\mathcal{P}_A$-unisolvent, since this in turn implies that $X_n^h$ is $\mathcal{P}_A$-unisolvent.
Then, from \Cref{thm: leb},
\begin{align*}
| f(\mathbf{0}) - f_n^h(\mathbf{0}) | \leq (1 + \lambda_{A}(\mathbf{0}; X_n^h)) \inf_{p \in \mathcal{P}_A} \| f - p \|_{\infty , \mathcal{X}_h} ,
\end{align*}
where $\lambda_A(\cdot ; X_n^h)$ denotes the (multivariate) Lebesgue function defined in \eqref{eq: mv leb fn}.
Since $\lambda_{A}(\mathbf{0}; h X_n) = \lambda_{A} (\mathbf{0}; X_n)$ for all $h > 0$, we have that
\begin{align*}
| f(\mathbf{0}) - f_n^h(\mathbf{0}) | \leq (1 + \lambda_{A}(\mathbf{0}; X_n)) \inf_{p \in \mathcal{P}_A} \| f - p \|_{\infty , \mathcal{X}_h} .
\end{align*}
The first term is constant in $h$.
For the second term, recalling the definition of the residual $R(\mathbf{x})$ in \eqref{eq: error},
\begin{align*}
\inf_{p \in \mathcal{P}_A} \| f - p \|_{\infty , \mathcal{X}_h}  & \leq \sup_{\mathbf{x} \in \mathcal{X}_h} \left| f(\mathbf{x}) - \sum_{\bm{\alpha} \in A} \beta_{\bm{\alpha}} \mathbf{x}^{\bm{\alpha}} \right| = \sup_{\mathbf{x} \in \mathcal{X}_h} | R(\mathbf{x}) | ,
\end{align*}
which completes the argument.
\end{proof}

\subsection{Proof of \Cref{thm: convergence SPRE}}
\label{subsec: SPRE converge proof}

Our approach to proving \Cref{thm: convergence SPRE} leverages the characterisation of the posterior mean in \ac{SPRE} as a minimal semi-norm interpolant:

\begin{lemma}[Characterisation as a minimal semi-norm interpolant]
\label{lem: min norm}
    Let $X_n \subset \mathcal{X}$ be $\mathcal{P}_A$-unisolvent and let $k : \mathcal{X} \times \mathcal{X} \rightarrow \mathbb{R}$ be a symmetric and positive definite kernel.
    Then for any simulator $f : \mathcal{X} \rightarrow \mathbb{R}$ there is a unique solution to the variational problem
    \begin{align}
        \argmin \left\{ |s|_{\mathcal{N}_k(\mathcal{X})} : s(\mathbf{x}) = f(\mathbf{x}) \; \forall \; \mathbf{x} \in X_n, s \in \mathcal{N}_k(\mathcal{X}) \right\} \label{eq: seminorm variational}
    \end{align}
    which we denote $s_{X_n}[f]$.
    Further, $s_{X_n}[f]$ coincides with the posterior mean of \ac{SPRE}.
\end{lemma}
\begin{proof}
    From the definition of the native space semi-norm, a solution $s \in \mathcal{N}_k(\mathcal{X})$ to the variational problem has the form $s = p + g$ where $(p,g) \in \mathcal{P}_A \times  \mathcal{H}_k(\mathcal{X})$ is a solution to
    \begin{align}
        \argmin\{ \|g\|_{\mathcal{H}_k(\mathcal{X})} : g(\mathbf{x}) = f(\mathbf{x}) - p(\mathbf{x}) \; \forall \; \mathbf{x} \in X_n, p \in \mathcal{P}_A, g \in \mathcal{H}_k(\mathcal{X}) \} . \label{eq: variational 2}
    \end{align}
    Let $p(\mathbf{x}) = \sum_{i=1}^m b_i \mathbf{x}^{\bm{\alpha}_i}$ for some coefficients $\mathbf{b} = (b_1 , \dots , b_m)^\top \in \mathbb{R}^m$.
    From the representer theorem, for fixed $p$, there exists a minimal $\|g\|_{\mathcal{H}_k(\mathcal{X})}$ norm interpolant of $f - p$ on $X_n$ that takes the form $g(\mathbf{x}) = \sum_{i=1}^n c_i k(\mathbf{x} , \mathbf{x}_i)$ for some coefficients $\mathbf{c} = (c_1 , \dots , c_n)^\top \in \mathbb{R}^n$, and we have $\|g\|_{\mathcal{H}_k(\mathcal{X})}^2 = \mathbf{c}^\top \mathbf{K} \mathbf{c}$.
    Thus \eqref{eq: variational 2} becomes
    \begin{align*}
        \argmin \{ \mathbf{c}^\top \mathbf{K} \mathbf{c} : \mathbf{K} \mathbf{c} = \mathbf{f} - \mathbf{V}_A \mathbf{b} , \mathbf{b} \in \mathbb{R}^m , \mathbf{c} \in \mathbb{R}^n \} ,
    \end{align*}
    whose solutions can be deduced (by the method of Lagrange multipliers) to be solutions of the linear system
    \begin{align*}
        \left[ \begin{array}{cc} \mathbf{K} & \mathbf{V}_A \\ \mathbf{V}_A^\top & \mathbf{0} \end{array} \right] \left[ \begin{array}{c} \mathbf{c} \\ \mathbf{b} \end{array} \right] = \left[ \begin{array}{c} \mathbf{f} \\ \mathbf{0} \end{array} \right] .
    \end{align*}
    This block matrix is invertible since $X_n$ is $\mathcal{P}_A$-unisolvent, so the solution pair $(\mathbf{b},\mathbf{c})$ is unique in $\mathbb{R}^m \times \mathbb{R}^n$ and we can uniquely identify the associated element $s$ as the minimal semi-norm interpolant $s_{X_n}[f]$.
    Calculating the inverse of this block matrix, we obtain 
    \begin{align*}
        \left[ \begin{array}{c} \mathbf{c} \\ \mathbf{b} \end{array} \right] = \left[ \begin{array}{c} \mathbf{K}^{-1} \mathbf{f} - \mathbf{k}(\mathbf{x})^\top \mathbf{K}^{-1} \mathbf{V}_A (\mathbf{V}_A^\top \mathbf{K}^{-1} \mathbf{V}_A)^{-1} \mathbf{V}_A^\top \mathbf{K}^{-1} \mathbf{f} \\ (\mathbf{V}_A^\top \mathbf{K}^{-1} \mathbf{V}_A)^{-1} \mathbf{V}_A^\top \mathbf{K}^{-1} \mathbf{f} \end{array}  \right] = \left[ \begin{array}{c} \mathbf{K}^{-1} \mathbf{f} - \mathbf{K}^{-1} \mathbf{V}_A \hat{\bm{\beta}}_A \\ \hat{\bm{\beta}}_A \end{array} \right]
    \end{align*}
    so that
    \begin{align}
        s_{X_n}[f](\mathbf{x}) = \mathbf{v}_A(\mathbf{x})^\top \mathbf{b} + \mathbf{k}(\mathbf{x})^\top \mathbf{c}
        = \mathbf{k}(\mathbf{x})^\top \mathbf{K}^{-1} \mathbf{f} + \mathbf{r}_A(\mathbf{x})^\top \hat{\bm{\beta}}_A . \label{eq: derive interp}
    \end{align}
    which is the posterior mean of \ac{SPRE} as given in \eqref{eq: SPRE mean}.    
\end{proof}

\Cref{lem: min norm} establishes the posterior mean of \ac{SPRE} in \eqref{eq: SPRE mean} as a minimal semi-norm interpolant.
Unpacking the terms $\mathbf{r}_A(\mathbf{x})$ and $\hat{\bm{\beta}}_A$ in \eqref{eq: derive interp}, we see that the solution to \eqref{eq: seminorm variational} has the form
\begin{align}
    s_{X_n}[f](\mathbf{x}) = \sum_{i=1}^n u_i^*(\mathbf{x}) f(\mathbf{x}_i) \label{eq: sf as linear}
\end{align}
with coefficient vector 
\begin{align}
     \mathbf{u}^*(\mathbf{x}) = \left[ \begin{array}{c} u_1^*(\mathbf{x}) \\ \vdots \\ u_n^*(\mathbf{x}) \end{array} \right] =  \mathbf{K}^{-1} \mathbf{k}(\mathbf{x}) + \mathbf{K}^{-1} \mathbf{V}_A (\mathbf{V}_A^\top \mathbf{K}^{-1} \mathbf{V}_A)^{-1} [ \mathbf{v}_A(\mathbf{x}) - \mathbf{V}_A^\top \mathbf{K}^{-1} \mathbf{k}(\mathbf{x}) ] .   \label{eq: def u star}
\end{align}
In particular $\mathbf{V}_A^\top \mathbf{u}^*(\mathbf{x}) = \mathbf{v}_A(\mathbf{x})$, which is to say that $s_{X_n}[p](\mathbf{x}) = p(\mathbf{x})$ for each $p(\mathbf{x}) = \mathbf{x}^{\bm{\alpha}}$ with $\bm{\alpha} \in A$.
Since $s_{X_n}$ is linear it follows that $s_{X_n}[p] = p$ for all $p \in \mathcal{P}_A$, which could be described as \emph{sparse polynomial exactness} of the minimal semi-norm interpolant.
This exactness is a key ingredient in the following Lemma, which itself is a slight generalisation of \citet[][Theorem 11.4]{wendland2004scattered} to remove the requirement in that work for $\mathcal{X}$ to be an open set:

\begin{lemma}
\label{lem: gen bound}
    Assume that $k$ is positive definite and $X_n$ is $\mathcal{P}_A$-unisolvent.
    Let $\mathbf{u}^*(\mathbf{x})$ be as defined in \eqref{eq: def u star} and let $\tilde{k}(\cdot,\mathbf{x}) = (I - \Pi_{\mathcal{P}_A})_{\mathbf{x}} k(\cdot,\mathbf{x})$.
    Then
    \begin{align*}
        | f(\mathbf{x}) - s_{X_n}[f](\mathbf{x}) | \leq |f|_{\mathcal{N}_k(\mathcal{X})} \left\| \tilde{k}(\cdot,\mathbf{x}) - \sum_{i=1}^n u_i^*(\mathbf{x}) \tilde{k}(\cdot , \mathbf{x}_i)  \right\|_{\mathcal{H}_k(\mathcal{X})} .
    \end{align*}
\end{lemma}
\begin{proof}

    From the representation of $s_{X_n}[f]$ in \eqref{eq: sf as linear} and the native space version of the reproducing property in \eqref{eq: native rep prop}, we have
    \begin{align*}
        f(\mathbf{x}) & = \Pi_{\mathcal{P}_A}[f](\mathbf{x}) + \langle h_f , \tilde{k}(\cdot, \mathbf{x}) \rangle_{\mathcal{H}_k(\mathcal{X})} \\
        s_{X_n}[f](\mathbf{x}) & = \sum_{i=1}^n u_i^*(\mathbf{x}) \{ \Pi_{\mathcal{P}_A}[f](\mathbf{x}_i) + \langle h_f , \tilde{k}(\cdot, \mathbf{x}_i) \rangle_{\mathcal{H}_k(\mathcal{X})} \} .
    \end{align*}
    From sparse polynomial exactness of the minimal semi-norm interpolant $s_{X_n}[f]$, we have that
    \begin{align*}
        \sum_{i=1}^n u_i^*(\mathbf{x}) \Pi_{\mathcal{P}_A}[f](\mathbf{x}) = \Pi_{\mathcal{P}_A}[f](\mathbf{x})
    \end{align*}
    and thus
    \begin{align*}
        f(\mathbf{x}) - s_{X_n}[f](\mathbf{x}) & = 
        \underbrace{ \Pi_{\mathcal{P}_A}[f](\mathbf{x}) - \Pi_{\mathcal{P}_A}[f](\mathbf{x}) }_{ = 0} + \left\langle h_f , \tilde{k}(\cdot,\mathbf{x}) - \sum_{i=1}^n u_i^*(\mathbf{x}) \tilde{k}(\cdot, \mathbf{x}_i) \right\rangle_{\mathcal{H}_k(\mathcal{X})} .
    \end{align*}
    From Cauchy--Schwarz we thus obtain
    \begin{align*}
        | f(\mathbf{x}) - s_{X_n}[f](\mathbf{x}) | & \leq \| h_f \|_{\mathcal{H}_k(\mathcal{X})} \left\| \tilde{k}(\cdot,\mathbf{x}) - \sum_{i=1}^n u_i^*(\mathbf{x}) \tilde{k}(\cdot, \mathbf{x}_i) \right\|_{\mathcal{H}_k(\mathcal{X})} .
    \end{align*}
    Recognising $\|h_f\|_{\mathcal{H}_k(\mathcal{X})}$ as $|f|_{\mathcal{N}_k(\mathcal{X})}$ completes the argument.
\end{proof}

From the last result, we see that it will be useful to more explicitly calculate the so-called \emph{power function}
\begin{align}
    P_{k,X_n}(\mathbf{x}) & := \left\| \tilde{k}(\cdot,\mathbf{x}) - \sum_{i=1}^n u_i^*(\mathbf{x}) \tilde{k}(\cdot , \mathbf{x}_i)  \right\|_{\mathcal{H}_k(\mathcal{X})} , \label{eq: def Q} 
\end{align}
and we can facilitate this using the native space reproducing property in \eqref{eq: native rep prop}, as the following lemma explains:

\begin{lemma}
Again with $\tilde{k}(\cdot,\mathbf{x}) = (I - \Pi_{\mathcal{P}_A})_{\mathbf{x}} k(\cdot,\mathbf{x})$, the power function in \eqref{eq: def Q} satisfies
\begin{align*}
    P_{k,X_n}(\mathbf{x})^2 
    & = k(\mathbf{x},\mathbf{x}) - 2 \sum_{i=1}^n u_i^*(\mathbf{x}) k(\mathbf{x} , \mathbf{x}_i) + \sum_{i=1}^n \sum_{j=1}^n u_i^*(\mathbf{x}) u_j^*(\mathbf{x}) k(\mathbf{x}_i , \mathbf{x}_j) .
\end{align*}
\end{lemma}
\begin{proof}
    Let $\{\bm{\xi}_i\}_{i=1}^m \subset \mathcal{X}$ be $\mathcal{P}_A$-unisolvent and let $\{p_i\}_{i=1}^m \subset \mathcal{P}_A$ be the corresponding Lagrange basis, so that $p_i(\bm{\xi}_j) = \delta_{i,j}$. 
    Then we can write the projection $\Pi_{\mathcal{P}_A} : \mathcal{N}_k(\mathcal{X}) \rightarrow \mathcal{P}_A$ explicitly as
    $$
    (\Pi_{\mathcal{P}_A} f)(\mathbf{x}) = \sum_{i=1}^m f(\bm{\xi}_i) p_i(\mathbf{x})
    $$
    for all $f \in \mathcal{N}_k(\mathcal{X})$ and all $\mathbf{x} \in \mathcal{X}$.
    Let us now use the usual reproducing property of $\mathcal{H}_k(\mathcal{X})$ to compute
    \begin{align}
        \langle \tilde{k}(\cdot , \mathbf{x}) , \tilde{k}(\cdot , \mathbf{y}) \rangle_{\mathcal{H}_k(\mathcal{X})} 
        & = \left\langle k(\cdot , \mathbf{x}) - \sum_{r=1}^m k(\cdot , \bm{\xi}_r) p_r(\mathbf{x}) , k(\cdot , \mathbf{y}) - \sum_{s=1}^m k(\cdot , \bm{\xi}_s) p_s(\mathbf{y})\right\rangle_{\mathcal{H}_k(\mathcal{X})} \nonumber \\
        & = k(\mathbf{x} , \mathbf{y}) - \sum_{r=1}^m k(\mathbf{x} , \bm{\xi}_r) p_r(\mathbf{y}) - \sum_{r=1}^m k(\mathbf{y} , \bm{\xi}_r) p_r(\mathbf{x}) \nonumber \\
        & \qquad + \sum_{r=1}^m \sum_{s=1}^m p_r(\mathbf{x}) k(\bm{\xi}_r , \bm{\xi}_s) p_s(\mathbf{y}) . \label{eq: in prod tilde k}
    \end{align}
    Now we use sparse polynomial exactness of the interpolant and the expression in \eqref{eq: in prod tilde k} to establish that
    \begin{align*}
        \sum_{i=1}^n u_i^*(\mathbf{x}) \langle \tilde{k}(\cdot , \mathbf{x}) , \tilde{k}(\cdot , \mathbf{x}_i) \rangle_{\mathcal{H}_k(\mathcal{X})}
        & = \sum_{i=1}^n u_i^*(\mathbf{x}) \left\{  k(\mathbf{x} , \mathbf{x}_i) - \sum_{r=1}^m k(\mathbf{x} , \bm{\xi}_r) p_r(\mathbf{x}_i) - \sum_{r=1}^m k(\mathbf{x}_i , \bm{\xi}_r) p_r(\mathbf{x}) \right. \\
        & \qquad \qquad \qquad \left. + \sum_{r=1}^m \sum_{s=1}^m p_r(\mathbf{x}) k(\bm{\xi}_r , \bm{\xi}_s) p_s(\mathbf{x}_i)  \right\} \\
        & = \sum_{i=1}^n u_i^*(\mathbf{x}) k(\mathbf{x} , \mathbf{x}_i) - \sum_{r=1}^m k(\mathbf{x},\bm{\xi}_r) p_r(\mathbf{x}) \\
        & \; - \sum_{i=1}^n u_i^*(\mathbf{x}) \sum_{r=1}^m k(\mathbf{x}_i , \bm{\xi}_r) p_r(\mathbf{x}) + \sum_{r=1}^m \sum_{s=1}^m p_r(\mathbf{x}) k(\bm{\xi}_r , \bm{\xi}_s) p_s(\mathbf{x})
    \end{align*}
and similarly
\begin{align*}
    \sum_{i=1}^n \sum_{j=1}^n u_i^*(\mathbf{x}) u_j^*(\mathbf{x}) \langle \tilde{k}(\cdot , \mathbf{x}_i) , \tilde{k}(\cdot , \mathbf{x}_j) \rangle_{\mathcal{H}_k(\mathcal{X})}
    & = \sum_{i=1}^n \sum_{j=1}^n u_i^*(\mathbf{x}) u_j^*(\mathbf{x}) k(\mathbf{x}_i , \mathbf{x}_j) \\
    & \qquad - 2 \sum_{i=1}^n u_i^*(\mathbf{x}) \sum_{r=1}^m k(\mathbf{x}_i , \bm{\xi}_r) p_r(\mathbf{x}) \\
    & \qquad + \sum_{r=1}^m \sum_{s=1}^m p_r(\mathbf{x}) k(\bm{\xi}_r,\bm{\xi}_s) p_s(\mathbf{x})
\end{align*}
Squaring and expanding \eqref{eq: def Q} we obtain
\begin{align*}
    P_{k,X_n}(\mathbf{x})^2 & = \langle \tilde{k}(\cdot , \mathbf{x}) , \tilde{k}(\cdot , \mathbf{x}) \rangle_{\mathcal{H}_k(\mathcal{X})} - 2 \sum_{i=1}^n u_i^*(\mathbf{x}) \langle \tilde{k}(\cdot , \mathbf{x}) , \tilde{k}(\cdot , \mathbf{x}_i) \rangle_{\mathcal{H}_k(\mathcal{X})} \\
    & \qquad + \sum_{i=1}^n \sum_{j=1}^n u_i^*(\mathbf{x}) u_j^*(\mathbf{x}) \langle \tilde{k}(\cdot , \mathbf{x}_i) , \tilde{k}(\cdot , \mathbf{x}_j) \rangle_{\mathcal{H}_k(\mathcal{X})}
\end{align*}
and upon substituting in the above expressions we obtain the stated result.    
\end{proof}

The following Lemma is a slight generalisation of \citet[][Theorem 11.9]{wendland2004scattered} which allows for $\mathcal{P}_A$ to be a sparse polynomial basis:

\begin{lemma} \label{lem: power fn bound}
    Suppose that $X_n = \{\mathbf{x}_i\}_{i=1}^n \subset \mathcal{X}$ is $\mathcal{P}_A$-unisolvent.
    Assume that $k : \mathcal{X} \times \mathcal{X} \rightarrow \mathbb{R}$ is a symmetric, positive definite and bounded kernel.
    Given $p \in \mathcal{P}_A$, let $\bar{p} : \mathcal{X} \times \mathcal{X} \rightarrow \mathbb{R}$ be defined as $\bar{p}(\mathbf{x} , \mathbf{x}') = p(\mathbf{x} - \mathbf{x}')$.
    Then there exists a $h$-independent constant $\tilde{C}_{A,X_n} > 0$ such that
    $$
    P_{k,X_n^h}(\mathbf{0}) \leq \tilde{C}_{A,X_n} \| k - \bar{p} \|_{\infty , \mathcal{X}_h \times \mathcal{X}_h}^{1/2}
    $$
    for all $\mathbf{x} \in \mathcal{X}$, $p \in \mathcal{P}_A$, and $h \in (0,1]$.
\end{lemma}
\begin{proof}
    To simplify the presentation we present the main argument for $h = 1$, so that $\mathcal{X}_h = \mathcal{X}$.
    Then we will retrospectively argue that the constants we introduced were $h$-independent.
    
    First note that the vector $\mathbf{u}^*(\mathbf{x})$ in \eqref{eq: def u star} minimises the quadratic form 
    \begin{align*}
        Q_{k,X_n,\mathbf{x}}(\mathbf{u}) := k(\mathbf{x},\mathbf{x}) - 2 \sum_{i=1}^n u_i k(\mathbf{x} , \mathbf{x}_i) + \sum_{i=1}^n \sum_{j=1}^n u_i u_j k(\mathbf{x}_i , \mathbf{x}_j) 
    \end{align*}
    over $M_{\mathbf{x}} := \{ \mathbf{u} \in \mathbb{R}^n : \sum_{i=1}^n u_i p(\mathbf{x}_i) = p(\mathbf{x}) \text{ for all } p \in \mathcal{P}_A \}$, and the minimum value is $P_{k,X_n}(\mathbf{x})^2$ for the power function defined in \eqref{eq: def Q}.
    Indeed, one can prove this by constructing the Lagrangian 
    \begin{align*}
    \mathcal{L}(\mathbf{u},\bm{\lambda}) = \frac{1}{2} Q_{k,X_n,\mathbf{x}}(\mathbf{u}) + \bm{\lambda}^\top \underbrace{ (\mathbf{V}_A^\top \mathbf{u} - \mathbf{v}_A(\mathbf{x})) }_{ \mathbf{u} \in M }
    \end{align*}
    where $\bm{\lambda}$ is a vector of Lagrange multipliers, and then differentiate to obtain
    \begin{align*}
        \nabla_{\mathbf{u}} \mathcal{L} & = - \mathbf{k}(\mathbf{x}) + \mathbf{K} \mathbf{u} + \mathbf{V}_A \bm{\lambda} \\
        \nabla_{\bm{\lambda}} \mathcal{L} & = \mathbf{V}_A^\top \mathbf{u} - \mathbf{v}_A(\mathbf{x}) .
    \end{align*}
    One uniquely obtains \eqref{eq: def u star} by simultaneously solving $\nabla_{\mathbf{u}} \mathcal{L} = \mathbf{0}$ and $\nabla_{\bm{\lambda}} \mathcal{L} = \mathbf{0}$.
    Further, the Hessian is
    \begin{align*}
        \left[ \begin{array}{cc} \nabla_{\mathbf{u},\mathbf{u}} \mathcal{L} & \nabla_{\mathbf{u},\bm{\lambda}} \mathcal{L} \\ \nabla_{\bm{\lambda},\mathbf{u}} \mathcal{L} & \nabla_{\bm{\lambda},\bm{\lambda}} \mathcal{L} \end{array} \right] = \left[ \begin{array}{cc} \mathbf{K} & \mathbf{V}_A \\ \mathbf{V}_A^\top & \mathbf{0}  \end{array} \right]
    \end{align*}
    which is a positive definite matrix, so that $\mathbf{u}^*$ is indeed a minimum of $Q_{k,X_n,\mathbf{x}}$ on $M_{\mathbf{x}}$.

    Now we take $\mathbf{x} = \mathbf{0}$.
    To upper-bound $P_{k,X_n}(\mathbf{0})^2$ it suffices to evaluate $Q_{k,X_n,\mathbf{0}}(\mathbf{u})$ on any element $\mathbf{u} \in M_{\mathbf{0}}$.
    For this purpose we pick the element $\tilde{\mathbf{u}} \in \mathbb{R}^n$ defined in \Cref{lem: poly repro}, for which $Q_{\bar{p},X_n,\mathbf{0}}(\tilde{\mathbf{u}}) = 0$ by the polynomial reproduction property.
    Then
    \begin{align*}
        Q_{k,X_n,\mathbf{0}}(\tilde{\mathbf{u}}) & = Q_{k,X_n,\mathbf{0}}(\tilde{\mathbf{u}}) - Q_{\bar{p},\mathbf{0}}(\tilde{\mathbf{u}})
        = Q_{k - \bar{p},X_n,\mathbf{0}}(\tilde{\mathbf{u}}) \\
        & = (k - \bar{p})(\mathbf{0} , \mathbf{0}) - 2 \sum_{i=1}^n \tilde{u}_i (k - \bar{p})(\mathbf{0} , \mathbf{x}_i) + \sum_{i=1}^n \sum_{j=1}^n \tilde{u}_i \tilde{u}_j (k - \bar{p})(\mathbf{x}_i , \mathbf{x}_j) \\
        & \leq |(k - \bar{p})(\mathbf{0} , \mathbf{0})| + 2 \sum_{i=1}^n | \tilde{u}_i | \| k - \bar{p}\|_{\infty , \mathcal{X} \times \mathcal{X}} + \sum_{i=1}^n \sum_{j=1}^n | \tilde{u}_i \tilde{u}_j | \| k - \bar{p} \|_{\infty , \mathcal{X} \times \mathcal{X}} \\
        & \leq \left( 1 + \sum_{i=1}^n |\tilde{u}_i| \right)^2 \| k - \bar{p} \|_{\infty , \mathcal{X} \times \mathcal{X}} \\
        & \leq (1 + C_{A,X_n} )^2 \| k - \bar{p} \|_{\infty , \mathcal{X} \times \mathcal{X}}
    \end{align*}
    as required with $\tilde{C}_{A,X_n} = 1 + C_{A,X_n}$.
    Finally, we notice that the above argument in fact holds for all $h \in (0,1]$, as by \Cref{lem: poly repro} there exist appropriate vectors $\tilde{\mathbf{u}}^{(h)} \in \mathbb{R}^n$ and the constant $C_{A,X_n}$ can be taken to be $h$-independent.
\end{proof}

At last we are ready to prove \Cref{thm: convergence SPRE}:

\begin{proof}[Proof of \Cref{thm: convergence SPRE}]
From \Cref{lem: min norm} we have that $\mathbb{E}[\mathsf{f}(\mathbf{0}) | \mathsf{f}(\mathbf{x}) = f(\mathbf{x}), \, \mathbf{x} \in X_n^h]$ coincides with $s_{X_n^h}[f](\mathbf{0})$, where $s_{X_n^h}[f]$ is a minimal semi-norm interpolant.
From \Cref{lem: gen bound} and \eqref{eq: def Q} we have a bound 
\begin{align}
| f(\mathbf{0}) - s_{X_n^h}[f](\mathbf{0}) | \leq P_{k,X_n^h}(\mathbf{0})  | f |_{\mathcal{N}_k(\mathcal{X}_h)} . \label{eq: goal}
\end{align}
From \eqref{eq: error} we have $f(\mathbf{x}) = p_f(\mathbf{x}) + R(\mathbf{x})$ where $p_f(\mathbf{x}) := \sum_{\bm{\alpha} \in A} \beta_{\bm{\alpha}} \mathbf{x}^{\bm{\alpha}}$, so we can bound the semi-norm in \eqref{eq: goal} via $|f|_{\mathcal{N}_k(\mathcal{X}_h)} = |R + p_f|_{\mathcal{N}_k(\mathcal{X}_h)} \leq \| R \|_{\mathcal{H}_k(\mathcal{X}_h)}$.
Note that the final term is finite since $\| R \|_{\mathcal{H}_k(\mathcal{X}_h)} \leq \| R \|_{\mathcal{H}_k(\mathcal{X})}$ follows from $\mathcal{X}_h \subset \mathcal{X}$ \citep[][Theorem 6]{berlinet2011reproducing}.
For the power function, from \Cref{lem: power fn bound} we obtain 
\begin{align}
    P_{k,X_n^h}(\mathbf{0}) 
    & \leq \tilde{C}_{A,X_n} \| k - \bar{p} \|_{\infty , \mathcal{X}_h \times \mathcal{X}_h }^{1/2}  \label{eq: pow fn bd}
\end{align}
where the bound is valid for any $\bar{p} : \mathcal{X} \times \mathcal{X} \rightarrow \mathbb{R}$ of the form $\bar{p}(\mathbf{x} , \mathbf{x}') = p(\mathbf{x} - \mathbf{x}')$ where $p \in \mathcal{P}_A$.
To obtain the stated result we can take $p = 0$ to obtain 
$$
C_{h,k} := \tilde{C}_{A,X_n} \| k \|_{\infty , \mathcal{X}_h \times \mathcal{X}_h}^{1/2}, 
$$
for which $\sup_{h \in (0,1]} C_{h,k} < \infty$ since $k$ is bounded.
\end{proof}

\subsection{Proof of \Cref{thm: extra smooth}}
\label{app: extra smooth}

\begin{proof}[Proof of \Cref{thm: extra smooth}]
    The proof proceeds as in the proof of \Cref{thm: convergence SPRE}, to the point of obtaining \eqref{eq: pow fn bd}.
    Instead of plugging $p = 0$ into \eqref{eq: pow fn bd}, we will now make a more intelligent choice of $p \in \mathcal{P}_A$.
    Since $\phi$ is $2s$ times continuously differentiable at $\mathbf{0}$, from Taylor's theorem
    \begin{align*}
        k(\mathbf{x},\mathbf{y}) = \sum_{|\bm{\alpha}| < 2s} \frac{ \partial_{\bm{\alpha}} \phi(\mathbf{0}) }{\bm{\alpha}!} (\mathbf{x} - \mathbf{y})^{\bm{\alpha}} + O(\|\mathbf{x} - \mathbf{y}\|^{2s}) .
    \end{align*}
    Since $\{|\bm{\alpha}| < 2 s\} \subseteq A$, we can pick $p \in \mathcal{P}_A$ such that 
    \begin{align*}
        \bar{p}(\mathbf{x},\mathbf{y}) = \sum_{|\bm{\alpha}| < 2 s} \frac{ \partial_{\bm{\alpha}} \phi(\mathbf{0}) }{\bm{\alpha}!} (\mathbf{x} - \mathbf{y})^{\bm{\alpha}}
    \end{align*}
    to obtain that
    \begin{align*}
        \| k - \bar{p} \|_{\infty , \mathcal{X}_h \times \mathcal{X}_h} = O\left( \sup_{\mathbf{x},\mathbf{y} \in \mathcal{X}_h} \|\mathbf{x} - \mathbf{y}\|^{2s} \right) = O(h^{2s}) 
    \end{align*}
    as $h \rightarrow 0$.
    The claim then follows from \eqref{eq: goal} and \eqref{eq: pow fn bd}.
\end{proof}

\subsection{Proof of \Cref{lem: equiv norms}}
\label{app: proof of equiv norms}

The edge case where $k(\mathbf{x} , \mathbf{x}') = \sigma^2  \delta_{\mathbf{x} , \mathbf{x}'}$ is a \emph{white noise} kernel recovers the straight-forward approach of modelling
\begin{align}
\mathsf{f}(\mathbf{x}) = \sum_{ \bm{\alpha} \in A } \beta_{\bm{\alpha}} \mathbf{x}^{\bm{\alpha}} + \epsilon(\mathbf{x})  , \qquad \epsilon(\mathbf{x}) \stackrel{\text{iid}}{\sim} N(0,\sigma^2)  \label{eq: minimal assum}
\end{align}
where the coefficients $\beta_{\bm{\alpha}}$ are assigned a flat prior and inferred in the Bayesian framework.
This choice of kernel might be appealing in that it renders the matrix $\mathbf{K} = \sigma^2 \mathbf{I}$ diagonal and the computational cost $O(n)$ rather than $O(n^3)$.

If the set $\mathcal{X}$ is countably infinite, then we may enumerate its elements as $\{\mathbf{x}_n\}_{n \in \mathbb{N}}$.
In the case, $\mathcal{H}_k(\mathcal{X})$ contains any function $R$ for which 
\begin{align}
\|R\|_{\mathcal{H}_k(\mathcal{X})} = \frac{1}{\sigma} \left( \sum_{n \in \mathbb{N}} R(\mathbf{x}_n)^2 \right)^{1/2} \label{eq: norm discrete}
\end{align}
is finite.
This explicit form for the norm $\|\cdot\|_{\mathcal{H}_k(\mathcal{X})}$ is the key tool in establishing \Cref{lem: equiv norms}.
Let $\|\mathbf{x}\|_\infty := \max\{|x_1| , \dots , |x_d|\}$ be the supremum norm on $\mathbb{R}^d$.

\begin{proof}[Proof of \Cref{lem: equiv norms}]
Note that the size of the design is $n = \sum_{0 \leq |\bm{\alpha}| \leq s} 1$.
Fixing $h = 1 / 2^M$, in terms of the notation in \Cref{thm: convergence SPRE}, we have
\begin{align*}
\mathcal{X}_h = \{\mathbf{0}\} \cup \left\{ \frac{1}{ 2^m}  \frac{1}{ 2^{\bm{\alpha}} }  : m \in \mathbb{N}_0, \; m \geq M, \; 0 \leq |\bm{\alpha}| \leq s \right\} .
\end{align*}
Our assumption implies that there exists $C_1, C_2, M_0$ such that, for all $\mathbf{x} \in \mathcal{X}_{h_0}$ with $h_0 = 1 / 2^{M_0}$, $C_1 \|\mathbf{x}\|_\infty^a \leq | R(\mathbf{x}) | \leq C_2 \|\mathbf{x}\|_\infty^a$.
Then, from the explicit form of the norm in \eqref{eq: norm discrete},
$$
\|R\|_{\mathcal{H}_k(\mathcal{X}_{h_0})}^2 \leq \frac{C_2}{\sigma^2} \sum_{m \geq M_0 } \sum_{0 \leq |\bm{\alpha}| \leq s} \left( \frac{1}{ 2^m} \right)^{2a}
= \frac{C_2 n}{\sigma^2} \sum_{m \geq M_0} \left( \frac{1}{2^{a}} \right)^{2m} =: (\star)
$$
which is finite for all $a > 0$.
Similarly, for $h = 1 / 2^{M + M_0}$, $M \in \mathbb{N}_0$, we get
$$
\|R\|_{\mathcal{H}_k(\mathcal{X}_{h'})}^2 \leq \left( \frac{1}{2^{a}} \right)^{2M} \cdot (\star) .
$$
Rearranging this last equation, we get
$$
(\star)^{-1} C_1 \|R\|_{\mathcal{H}_k(\mathcal{X}_h)} \leq C_1  \left(  \frac{1}{2^a} \right)^{M}  = C_1 \left(  \frac{1}{2^M} \right)^{a}  = C_1 \sup_{\mathbf{x} \in \mathcal{X}_h } \|\mathbf{x}\|_{\infty}^a \leq \sup_{\mathbf{x} \in \mathcal{X}_h} |R(\mathbf{x})| ,
$$
for all permissible scalings $h \leq h_0$, as required.
\end{proof}

\section{Experimental Protocol}
\label{app: protocol}

This appendix contains full details required to reproduce the empirical results that we presented in the main text.
\Cref{app: kernels} formally defines the kernels $k$ which were considered.
Full details for the cubature experiment reported in \Cref{sec: methods} are contained in \Cref{app: cubature illustration}, and full details for the experimental design experiment reported in \Cref{sec: UQ} are contained in \Cref{app: ED detail}.
Implementational details pertaining to the three case studies are contained in \Cref{app:3d-models-detail}.

\subsection{Kernels}
\label{app: kernels}

This appendix contains precise definitions of the parametric kernels $k_\theta$ that we used in our experiments.
In what follows we let $\mathrm{softplus}(z) = \log(1 + \exp(z))$, which transforms unconstrained $z \in \mathbb{R}$ into a positive output.

\begin{itemize}
    \item The \emph{white noise} kernel is
    $$
    k(\mathbf{x},\mathbf{x}') = \sigma^2 \delta_{\mathbf{x},\mathbf{x}'}
    $$
    where $\sigma^2 > 0$ is a learnable parameter and $\delta_{\mathbf{x},\mathbf{x}'} = 1$ if $\mathbf{x} = \mathbf{x}'$ and otherwise $\delta_{\mathbf{x},\mathbf{x}'} = 0$.
    For optimisation we let $\sigma^2 = \mathrm{softplus}(\theta)$ and initialise at $\theta = 1$.
    \item The \emph{Mat\'{e}rn}-$\frac{1}{2}$ kernel is
    $$
    k(\mathbf{x},\mathbf{x}') = \sigma^2 \exp\left( - \frac{\|\mathbf{x}-\mathbf{x}'\|}{\ell} \right)
    $$
    where $\sigma^2 > 0$ and $\ell > 0$ are learnable parameters.
    For optimisation we let $\sigma^2 = \mathrm{softplus}(\theta_1)$ and $\ell = \mathrm{softplus}(\theta_2)$, and initialise at $\bm{\theta} = [1,1]$.
    \item The \emph{Mat\'{e}rn}-$\frac{3}{2}$ kernel is
    $$
    k(\mathbf{x},\mathbf{x}') = \sigma^2 \left( 1 + \frac{\sqrt{3}\|\mathbf{x}-\mathbf{x}'\|}{\ell} \right) \exp\left( - \frac{\sqrt{3}\|\mathbf{x}-\mathbf{x}'\|}{\ell} \right)
    $$
    where $\sigma^2 > 0$ and $\ell > 0$ are learnable parameters.
    For optimisation we let $\sigma^2 = \mathrm{softplus}(\theta_1)$ and $\ell = \mathrm{softplus}(\theta_2)$ initialise at $\bm{\theta} = [1,1]$.
    \item The \emph{Gaussian} kernel is
    $$
    k(\mathbf{x},\mathbf{x}') = \sigma^2 \exp\left( - \frac{\|\mathbf{x}-\mathbf{x}'\|^2}{\ell^2} \right)
    $$
    where $\sigma^2 > 0$ and $\ell > 0$ are learnable parameters.
    For optimisation we let $\sigma^2 = \mathrm{softplus}(\theta_1)$ and $\ell = \mathrm{softplus}(\theta_2)$ initialise at $\bm{\theta} = [1,0.1]$.
\end{itemize}

\noindent For optimisation of kernel parameters $\bm{\theta}$, in all experiments we used automatic differentiation of the \ac{LOOCV} criterion \eqref{eq: LOOCV error} in combination with a trust region method (the \texttt{trust-krylov} method in SciPy), which was run until convergence, except in the cubature experiments of \Cref{sec: methods}, where 10 iterations were performed; the only other exception is \ac{GRE} in \Cref{fig: illus}, for which the derivative-free Nelder--Mead method was used to reduce computational time.

\subsection{Cubature Experiment}
\label{app: cubature illustration}

The cubature experiments reported in \Cref{sec: methods} of the main text were based on the setting described in \Cref{subsec: cubature}.
To construct an integrand $g \in C^{2s+2}([\mathbf{0},\mathbf{1}])$ we first took the scalar function $\phi_0(z) = |z - \frac{1}{2}|$ and computed indefinite integrals $2s+2$ times; i.e. $\phi_i(z) = \int^z \phi_{i-1}(z') \mathrm{d}z'$ for $i = 1,\dots,2s+2$.
Then we set $g(\mathbf{t}) = 1 + \phi_{2s+2}((t_1 + \cdots + t_d)/d) / \|\phi_{2s+2}\|_\infty$.
Symbolic computation was used to calculate the true value of the integral $\int_{[\mathbf{0},\mathbf{1}]} g(\mathbf{t}) \mathrm{d}\mathbf{t}$ exactly.

For the illustration in \Cref{fig: illus} we took $d = 2$ and $s = 1$.
The $n$ inputs $X_n$ were sampled independently and uniformly at random from $[\frac{1}{6},\frac{1}{3}]^2$; for \ac{MRE} $n = 3$, for \ac{GRE} $n = 64$, and for \ac{SPRE} $n = 4$.
For \ac{GRE} we took $k$ to be the Mat\'{e}rn-$\frac{3}{2}$ kernel, with $\mathrm{Lead}(A) = \{(1,1),(2,0),(0,2)\}$ the leading order terms in \eqref{eq: Riemann Taylor}, learning the kernel length scale and amplitude parameters using \ac{LOOCV} as described in \Cref{sec: UQ}.
Note that this is a departure from the maximum likelihood approach proposed for \ac{GRE} in \citet{oates2024probabilistic}; we found empirically that \ac{LOOCV} produced better-calibrated output.
For \ac{SPRE} we took $k$ to be the white noise kernel $k(\mathbf{x},\mathbf{x}') = \sigma^2 \delta_{\mathbf{x},\mathbf{x}'}$ where the amplitude parameter $\sigma>0$ was learned using \ac{LOOCV} as described in \Cref{sec: UQ}.

For the assessment of convergence acceleration in \Cref{fig: convergence} we used the same protocol, now varying $d \in \{1,2,3\}$ and $s \in \{0,1\}$.
However, in this case we compared the performance of \ac{MRE}, \ac{GRE} and \ac{SPRE} based on the same inputs $X_n$, so their accuracy can be directly compared.
The design sets were defined as $X_n = \frac{1}{2} \bar{X}_n$ where, respectively for $d \in \{1,2,3\}$, 
\begin{align*} 
\bar{X}_n & = \textstyle \left\{1 , \frac{1}{2} , \frac{1}{3}, \frac{1}{4} \right\} , \\
\bar{X}_n & = \textstyle \left\{ (1,1) , \left( 1 , \frac{1}{2} \right) ,  \left( \frac{1}{2} , 1 \right) , \left( \frac{1}{2} , \frac{1}{2} \right) , \left( 1 , \frac{1}{3} \right) , \left( \frac{1}{3} , 1 \right) \right\} , \\
\bar{X}_n & = \textstyle \left\{ (1,1,1) , \left( 1, 1 , \frac{1}{2} \right) , \left( 1, \frac{1}{2} ,1 \right) , \left( \frac{1}{2} ,1,1\right) , \left( 1, \frac{1}{2} , \frac{1}{2} \right) , \left( \frac{1}{2}, 1 , \frac{1}{2} \right) , \left( \frac{1}{2}, \frac{1}{2} , 1 \right) , \left( \frac{1}{2} , \frac{1}{2} , \frac{1}{2} \right)   \right\} .
\end{align*}
For \ac{MRE} and \ac{SPRE} the index set $A$ was again set according to the $s$-dependent expansion in \eqref{eq: Riemann Taylor}.
For \ac{GRE} we again used $\mathrm{Lead}(A)$ as explained in \Cref{sec: GRE}.
Since the number of data that can be used in \ac{MRE} is $\mathrm{dim}(A)$, we selected the elements of $X_n$ that were closest to $\mathbf{0}$ for extrapolation using \ac{MRE}.
The scaled inputs $X_n^h$ were multiples $h X_n$ with $h = (\frac{1}{2})^m$ for $m \in \mathbb{N}_0$.

\subsection{Experimental Design}
\label{app: ED detail}

This appendix contains full details to reproduce the experimental design experiment reported in \Cref{sec: UQ}.
As this is a proof-of-concept, we used a simple data-generating function 
\begin{align}
    f(\mathbf{x}) = \underbrace{ 1 + x_1 - 2 x_2 + 3 x_1^2 }_{p(\mathbf{x})} \; + \; 0.0001 x_1 x_2 \epsilon(\mathbf{x}) \label{eq: ED data gen}
\end{align}
where $\epsilon(\mathbf{x})$ is a standard white noise process on $[\mathbf{0},\mathbf{1}]$.
The simple form of $f$ allows us to deduce that $f(\mathbf{0}) = 1$, and that the best polynomial approximation is $p$ in \eqref{eq: ED data gen}, an element of $\mathcal{P}_A$ for the index set $A = \{(0,0), (1,0), (0,1), (2,0)\}$. 
The Mat\'{e}rn-$\frac{1}{2}$ kernel was used, both to learn $(k,A)$ and in the design criterion.
For experimental design, stochastic search was used to numerically solve \eqref{eq: constrained opt}.
Specifically, we constructed a total of 1,000 candidate designs $X_m^*$ and the best design was selected.
To construct each candidate design $X_m^*$ we sequentially sampled elements $\mathbf{x}_i^*$ from the uniform distribution over $[\mathbf{0},\mathbf{1}]$, stopping one iteration before the computational cost $\sum_{i=1}^m c(\mathbf{x}_i^*)$ exceeded the computational budget $C$.

\subsection{Implementation of Case Studies}
\label{app:3d-models-detail}

This appendix summarises the most important implementational aspects of the three case studies that we report.

\paragraph{Implementation of \ac{MRE}}

The traditional approach to \ac{MRE} relies on numerical analysis to select an appropriate index set $A$.
However, for the three case studies that we report conducting a meaningful analysis is extremely difficult.
For the simulations we report, we take $A = \{\bm{\alpha} : |\bm{\alpha}| \leq 1 , \, \bm{\alpha} \in \mathbb{N}_0^d \}$ for \ac{MRE}.
This represents a conservative choice; including more terms in $A$ could in principle further accelerate convergence if a higher-order expansion \eqref{eq: error} holds, but for such complex numerical simulations we do not expect such a higher-order expansion to hold.
Since the number of data required for \ac{MRE} in this case is $n = 1 + d$, which is less than $n$ in each of the three case studies, only the $n$ data points for which $\|\mathbf{x}_i\|$ is smallest were used to train \ac{MRE}.

\paragraph{Implementation of \ac{GRE}}

For the experiments that we report, \ac{GRE} was based on the white noise kernel $k$, so that the covariance between $\mathsf{f}(\mathbf{x})$ and $\mathsf{f}(\mathbf{x}')$ was modelled as $\sigma^2 \epsilon(\mathbf{x}) \epsilon(\mathbf{x}') \delta_{\mathbf{x},\mathbf{x}'}$, where $\epsilon(\mathbf{x})$ is derived from the index set $A$ as explained in \eqref{eq: b}.
This kernel represents the weakest regularity assumptions that could be made, which we considered to be a reasonable choice due to the difficulty of arguing in favour of stronger regularity in the three case studies that we considered.
Being rooted in a statistical model, data-driven selection of $A$ is also possible for \ac{GRE}.
To ensure a fair comparison between \ac{GRE} and \ac{SPRE}, the stepwise selection method described in \Cref{sec: UQ} was also used to select $A$ for \ac{GRE}.
That is, we initialise $A = \{\mathbf{0}\}$ and consider adding terms to this set one at a time, keeping those which lead to a reduction in error as measured with \ac{LOOCV}.

\paragraph{Ground Truth}

Since in all case studies the continuum limit is intractable, for assessment purposes we actually consider $f(\mathbf{x} + \mathbf{x}_0)$ for a small offset $\mathbf{x}_0$, so that with sufficient computational effort the ground truth $f(\mathbf{x}_0)$ can be computed.
The offset used for the two-spheres and five-shapes case studies was $\mathbf{x}_0=(10^{-5},0.02,0.001)$, while for the agent-based flocking case study the offset was $\mathbf{x}_0=(0.01,10^{-15},0)$.
These offsets were selected such that the ground truth could be computed within a computationally feasible time frame, making the study practical.

\paragraph{Design Points}

For each case study, we considered a design set of size $n = 2^d$.
Since each case study involves $d = 3$ discretisation parameters, we employed the same reference design $X_n$ of size $n = 2^3 = 8$ throughout:
\[
X_n =
\left\{ \begin{array}{rcl}
(0.062, & 0.812, & 0.437), \\
(0.187, & 0.312, & 0.937), \\
(0.312, & 0.937, & 0.187), \\
(0.437, & 0.062, & 0.687), \\
(0.562, & 0.687, & 0.062), \\
(0.687, & 0.187, & 0.562), \\
(0.812, & 0.562, & 0.312), \\
(0.937, & 0.437, & 0.812)
\end{array} \right\}
\]
This design was numerically confirmed to be $\mathcal{P}_A$-unisolvent for all the index sets $A$ of cardinality at most $n$ which we considered.
This reference design was then scaled to obtain $X_n^h = \{h\mathbf{x} : \mathbf{x} \in X_n\}$ as described in the main text.

\paragraph{Computational Requirement}

All computations involving \ac{SPRE} for a given parametric kernel (i.e. computing each line for \ac{SPRE} in \Cref{fig:two-spheres-abs}) typically required just over one hour of computation on the high-performance computing service, called Comet, at Newcastle University (ran in serial).
This is negligible compared to the time required to simulate the raw data on which extrapolation was based.

\subsubsection{Two Spheres and Five Shapes 3D Models}

Here we report implementational details pertaining to the two-spheres and five-shapes case studies.

\paragraph{Physics Engine}

All 3D simulations were performed using \emph{Multi-Joint dynamics with Contact} (\acs{MuJoCo}) in Python \citep{todorov2012mujoco}, which is available at \url{https://github.com/google-deepmind/mujoco}.
\ac{MuJoCo} is a high-performance rigid-body dynamics simulator designed for modelling articulated systems with contact-rich interactions. The models were executed through the official Python API, which provides programmatic access to the \ac{MuJoCo} simulation pipeline and solver configuration.
Each simulation scenario is specified by an XML model file defining the physical environment, including bodies, geometries, contact parameters, solver settings, and visualisation options. 
All model files are available on the GitHub repository associated with this manuscript.

\paragraph{Temporal Interpolation}

During each run, the discretisation parameters are inserted into the XML configuration before the model is instantiated. 
The resulting model is then simulated by repeatedly advancing the system dynamics using the \ac{MuJoCo} stepping routine.
Since the desired quantities of interest do not in general coincide exactly with a discrete time step, linear interpolation is performed to estimate these quantities at the precise time required. 

\paragraph{Solver Settings}

For the three case studies that we report, the numerical integrator is set to explicit Euler and the constraint solver to the Newton method. 
The constraint solver is configured with $100$ iterations, $10$ no--slip iterations, and a tolerance of $10^{-8}$. 
Each geometry is assigned contact regularisation parameters that control how contact forces respond to constraint violations, determining both the reference dynamics and the effective stiffness of the contact interactions.

\subsubsection{Agent-Based Flocking Model}
\label{app:mult-agent-detail}

Here we report implementational details pertaining to experiments involving the agent-based flocking model.

\paragraph{Physics Engine}

The multi-agent flocking model is implemented using the \emph{Mesa} Python library \citep{Mesa}. 
Agents interact in a continuous two-dimensional domain $\Omega = [0,10] \times [0,10]$ with periodic (toroidal) boundaries. Each agent repels others at very short distances to prevent overlap, using a repulsion radius to define a short-range exclusion zone. At intermediate distances, agents attract one another to maintain group cohesion up to the interaction radius, with interaction strength smoothly tapering off near the limit. This smooth tailing avoids abrupt force cutoffs, ensuring stable dynamics.

\paragraph{Flocking Model}
The system contains 60 agents with random initial positions (generated using a fixed random seed, so the model is in fact deterministic). 
Each agent $i$ is characterised by its position $\mathbf{u}_i(t) \in \Omega$ and velocity $\mathbf{v}_i(t) \in \mathbb{R}^2$ at time $t$.

\paragraph{Numerical Methods}
Numerical accuracy in \emph{Mesa} is controlled by three discretisation parameters, a time step $x_1$, a repulsion softening parameter $x_2$, and a cutoff width $x_3$.
For two agents $i$ and $j$, the displacement vector is $\mathbf{d}_{ij} = \mathbf{u}_j - \mathbf{u}_i$, with periodic corrections applied to each coordinate to respect toroidal boundaries. 
The inter-agent distance is then $r_{ij} = \lVert \mathbf{d}_{ij} \rVert$, and the unit vector pointing from $i$ to $j$ is $\hat{\mathbf{d}}_{ij} = \mathbf{d}_{ij} / r_{ij}$.
The pairwise force from agent $j$ on agent $i$ has repulsive and attractive components:
\begin{equation*}
F^{\mathrm{rep}}_{ij} =
\max\left(0, -\frac{R_r - r_{ij}}{r_{ij} + x_2 }\right), \qquad F^{\mathrm{att}}_{ij} =
\max\left(0, r_{ij} - R_r\right),
\end{equation*}
where $R_r$ is the \emph{repulsion radius} (in our simulations, $R_r=0.5$) and $x_2$ is the repulsion softening parameter. The total pairwise force vector is $\mathbf{f}_{ij} = \left(F^{\mathrm{rep}}_{ij} + F^{\mathrm{att}}_{ij}\right)\hat{\mathbf{d}}_{ij}$.
Interactions are modulated by a smooth cutoff function
\begin{equation*}
w(r) =
\begin{cases}
1, & x_3 = 0 \text{ and } r < R, \\
0, & x_3 = 0 \text{ and } r \geq R, \\
\dfrac{1}{2}\left(1 - \tanh\left(\dfrac{r - R}{x_3}\right)\right), & x_3 > 0,
\end{cases}
\end{equation*}
where $R$ is the \emph{interaction radius} (in our simulations, $R=2$) and $x_3$ is the cutoff width parameter. This ensures a continuous decay of interaction strength near the cutoff distance.
The net force on agent $i$ is then
\begin{equation*}
\mathbf{F}_i = \sum_{j \neq i} w(r_{ij}) \, \mathbf{f}_{ij}.
\end{equation*}
Agents are updated at each time step by incrementing their velocities and positions according to the computed forces with time increment $x_1$:
\begin{align*}
\mathbf{v}_i(t+x_1) &= \mathbf{v}_i(t) + x_1 \, \mathbf{F}_i(t), \\
\mathbf{u}_i(t+x_1) &= \mathbf{u}_i(t) + x_1 \, \mathbf{v}_i(t+x_1),
\end{align*}
with positions wrapped to remain within the toroidal domain $\Omega$.

\paragraph{Temporal Interpolation}

Since the desired quantities of interest do not in general coincide exactly with a discrete time step, linear interpolation is performed to estimate these quantities at the precise time required.

\section{Additional Experiments}

This appendix reports additional experiments and their results, focusing on sequential experimental design (\Cref{app: sequential balls}) and convergence acceleration (\Cref{app: additional}).

\subsection{Sequential Experimental Design for the Two Spheres Case Study}
\label{app: sequential balls}

For the experiment reported in \Cref{fig: sequential design}, the cost function was $c(\mathbf{x}) = 1/x_1$, and we considered scaling factors $h = 2^{-k}$ for $k \in \{34,\dots,39\}$, together with $k \in \{32,33\}$ for the fixed design only.
For each $h$, the fixed design was $X_n^h$, where $X_n$ is the reference design of size $n = 8$ given in \Cref{app:3d-models-detail}, with total cost $C_h = \sum_{i=1}^n c(h \mathbf{x}_i)$.
The sequential design was initialised with the four elements of $X_n^h$ with largest $x_1$, i.e. the cheapest, and the remaining budget was divided equally among three iterations.
At each iteration the index set $A$ and the kernel parameters were estimated using the stepwise \ac{LOOCV} procedure of \Cref{sec: UQ}, and a new batch of simulations was selected by stochastic search to approximately solve \eqref{eq: constrained opt}.
Specifically, 1,000 candidate batches were constructed by sequentially sampling points $h \mathbf{u}$ with $\mathbf{u}$ drawn uniformly from $[0.01,1]^3$, stopping one iteration before the cost of the batch exceeded the budget for that iteration, and the candidate batch minimising the posterior variance of $f(\mathbf{0})$ was selected.
The Mat\'{e}rn-$\frac{1}{2}$ kernel was used for both the fixed and the sequential designs, and the sequential design was repeated with three random seeds.
The code to reproduce this experiment is contained in the script \texttt{run\_sequential\_design.py} in the GitHub repository associated with this paper.

\subsection{Additional Empirical Results}
\label{app: additional}

This section contains additional empirical results relating to the three case studies in the main text.
\Cref{fig: two spheres sparsity matern1/2,fig: two spheres sparsity matern3/2,fig: two spheres sparsity gaussian} concern the estimated index set $A$ for the two-spheres case study when different kernels are used.
\Cref{fig:five-shapes-errorbars} reports the predictive uncertainty of \ac{SPRE} for the five-shapes case study, and additionally the estimated index sets $A$ are contained in \Cref{fig: five shapes sparsity,fig: five shapes sparsity matern1/2,fig: five shapes sparsity matern3/2,fig: five shapes sparsity gaussian}.
\Cref{fig:flock-errorbars} reports the predictive uncertainty of \ac{SPRE} for the agent-based flocking model, and additionally the estimated index sets $A$ are contained in \Cref{fig: flocking sparsity matern1/2,fig: flocking sparsity matern3/2,fig: flocking sparsity gaussian}.

\begin{figure}[h!]
    \centering
    \includegraphics[width=\textwidth]{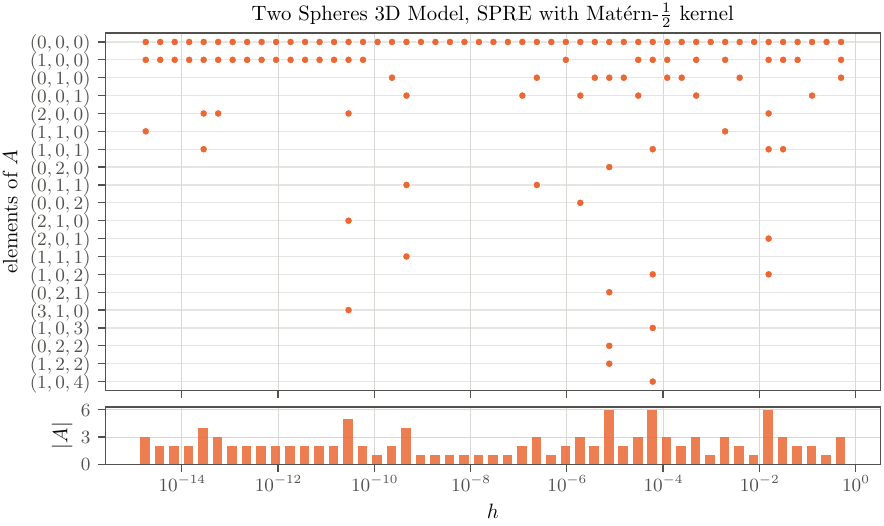}
    \caption{\emph{Two Spheres 3D Model}.  
    The index set $A$ was estimated based on a design of size $n = 2^d$, for a range of values of the scaling factor $h$ controlling the design set $X_n^h$.
    Here the Mat\'{e}rn-$\frac{1}{2}$ kernel was used.
    The lower panel displays the cardinality $|A|$ of the estimated index set.
    \alttext{A figure displaying the estimated index set $A$ as the global scaling parameter $h$ is varied.}
    }
    \label{fig: two spheres sparsity matern1/2}
\end{figure}

\begin{figure}[h!]
    \centering
    \includegraphics[width=\textwidth]{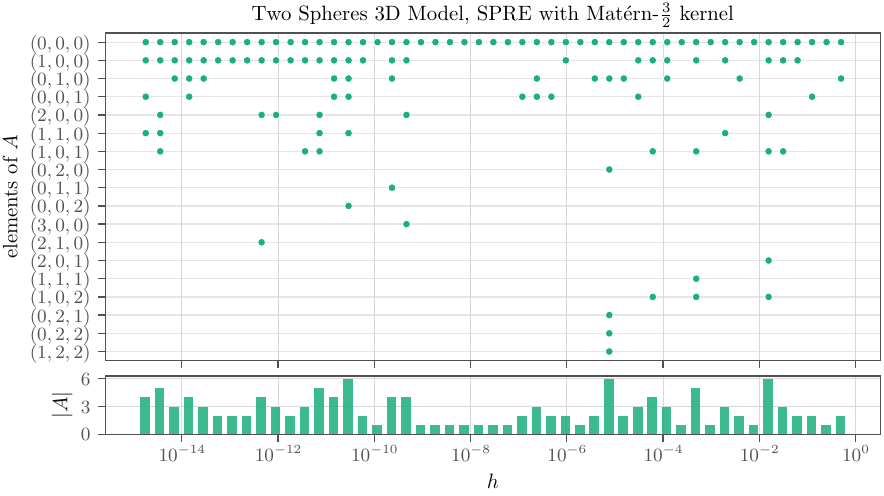}
    \caption{\emph{Two Spheres 3D Model}.  
    The index set $A$ was estimated based on a design of size $n = 2^d$, for a range of values of the scaling factor $h$ controlling the design set $X_n^h$.
    Here the Mat\'{e}rn-$\frac{3}{2}$ kernel was used.
    The lower panel displays the cardinality $|A|$ of the estimated index set.
    \alttext{A figure displaying the estimated index set $A$ as the global scaling parameter $h$ is varied.}
    }
    \label{fig: two spheres sparsity matern3/2}
\end{figure}

\begin{figure}[h!]
    \centering
    \includegraphics[width=\textwidth]{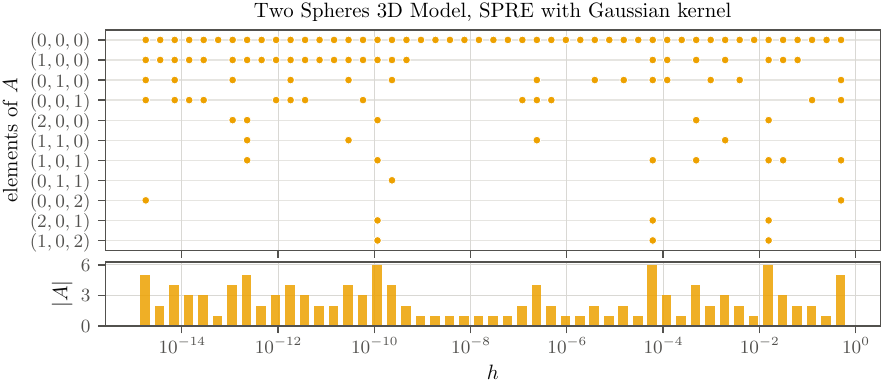}
    \caption{\emph{Two Spheres 3D Model}.  
    The index set $A$ was estimated based on a design of size $n = 2^d$, for a range of values of the scaling factor $h$ controlling the design set $X_n^h$.
    Here the Gaussian kernel was used.
    The lower panel displays the cardinality $|A|$ of the estimated index set.
    \alttext{A figure containing 4 panels, each of which shows the predictive error bars produced by SPRE as a function of the global scaling parameter $h$.
Each panel contains results for a different type of kernel.
\alttext{A figure displaying the estimated index set $A$ as the global scaling parameter $h$ is varied.}
}
    }
    \label{fig: two spheres sparsity gaussian}
\end{figure}

\begin{figure}[t!]
\centerline{\includegraphics[width=\textwidth]{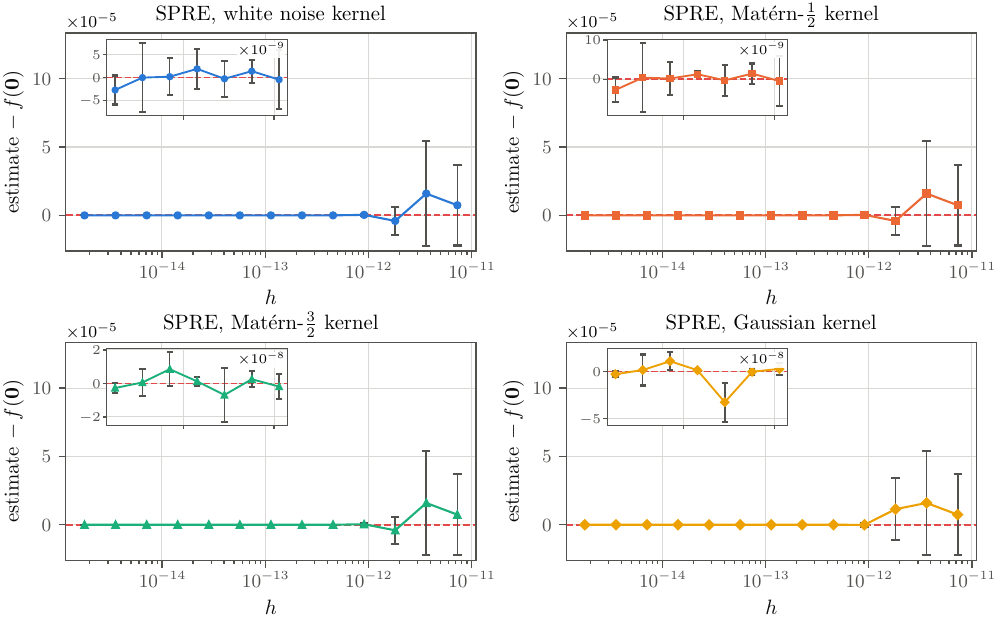}}
\caption{\textit{Five Shapes 3D Model}.
The two standard deviation predictive intervals from \ac{SPRE} are displayed as a function of the scaling factor $h$ controlling the design set $X_n^h$.
The intervals are displayed relative to the true value $f(\mathbf{0})$, which therefore corresponds to the dashed red line at zero; insets display the seven smallest values of $h$ on an expanded scale.
Top left to bottom right: white noise, Mat\'{e}rn-$\frac{1}{2}$, Mat\'{e}rn-$\frac{3}{2}$ and Gaussian kernels.
\alttext{A figure containing 4 panels, each of which shows the predictive error bars produced by SPRE as a function of the global scaling parameter $h$.
Each panel contains results for a different type of kernel.}
}
\label{fig:five-shapes-errorbars}
\end{figure}

\begin{figure}[h!]
    \centering
    \includegraphics[width=\textwidth]{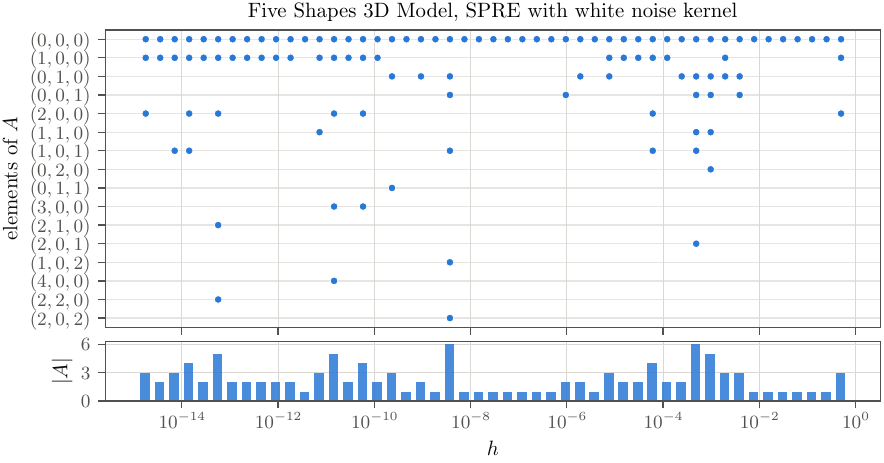}
    \caption{\emph{Five Shapes 3D Model}.  
    The index set $A$ was estimated based on a design of size $n = 2^d$, for a range of values of the scaling factor $h$ controlling the design set $X_n^h$.
    Here the white noise kernel was used.
    The lower panel displays the cardinality $|A|$ of the estimated index set.
    \alttext{A figure displaying the estimated index set $A$ as the global scaling parameter $h$ is varied.}
    }
    \label{fig: five shapes sparsity}
\end{figure}

\begin{figure}[h!]
    \centering
    \includegraphics[width=\textwidth]{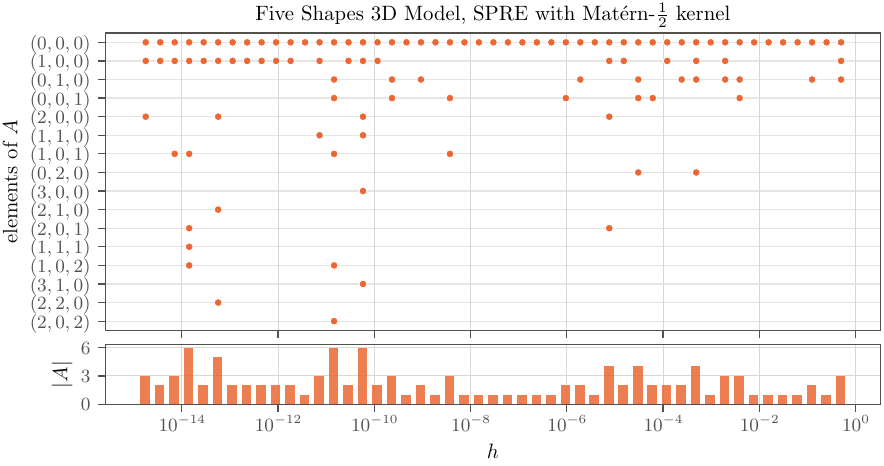}
    \caption{\emph{Five Shapes 3D Model}.  
    The index set $A$ was estimated based on a design of size $n = 2^d$, for a range of values of the scaling factor $h$ controlling the design set $X_n^h$.
    Here the Mat\'{e}rn-$\frac{1}{2}$ kernel was used.
    The lower panel displays the cardinality $|A|$ of the estimated index set.
    \alttext{A figure displaying the estimated index set $A$ as the global scaling parameter $h$ is varied.}
    }
    \label{fig: five shapes sparsity matern1/2}
\end{figure}

\begin{figure}[h!]
    \centering
    \includegraphics[width=\textwidth]{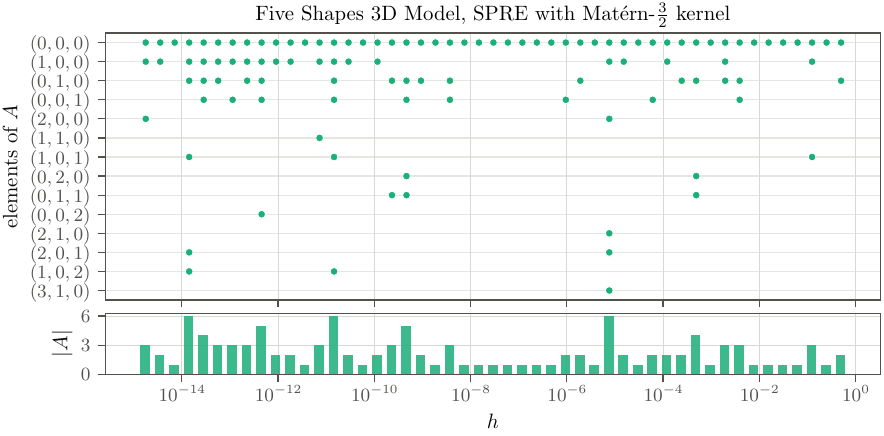}
    \caption{\emph{Five Shapes 3D Model}.  
    The index set $A$ was estimated based on a design of size $n = 2^d$, for a range of values of the scaling factor $h$ controlling the design set $X_n^h$.
    Here the Mat\'{e}rn-$\frac{3}{2}$ kernel was used.
    The lower panel displays the cardinality $|A|$ of the estimated index set.
    \alttext{A figure displaying the estimated index set $A$ as the global scaling parameter $h$ is varied.}
    }
    \label{fig: five shapes sparsity matern3/2}
\end{figure}

\begin{figure}[h!]
    \centering
    \includegraphics[width=\textwidth]{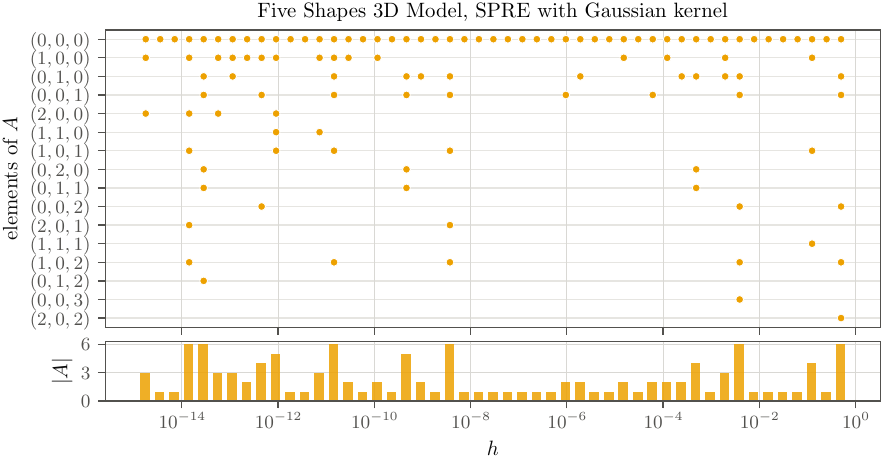}
    \caption{\emph{Five Shapes 3D Model}.  
    The index set $A$ was estimated based on a design of size $n = 2^d$, for a range of values of the scaling factor $h$ controlling the design set $X_n^h$.
    Here the Gaussian kernel was used.
    The lower panel displays the cardinality $|A|$ of the estimated index set.
    \alttext{A figure containing 4 panels, each of which shows the predictive error bars produced by SPRE as a function of the global scaling parameter $h$.
Each panel contains results for a different type of kernel.
\alttext{A figure displaying the estimated index set $A$ as the global scaling parameter $h$ is varied.}
}
    }
    \label{fig: five shapes sparsity gaussian}
\end{figure}

\begin{figure}[t!]
\centerline{\includegraphics[width=\textwidth]{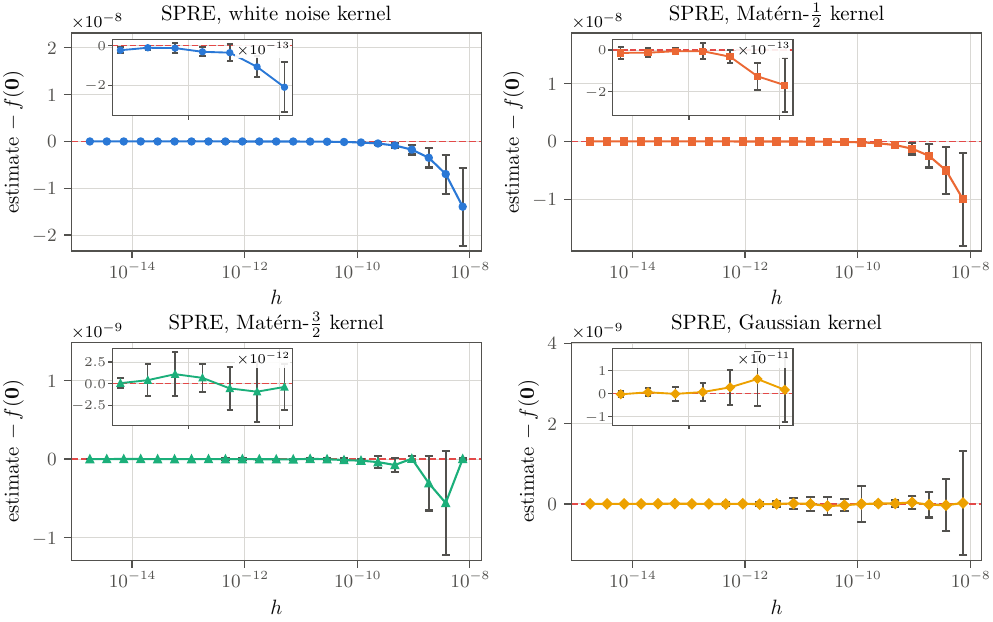}}
\caption{\textit{Agent-Based Flocking Model}.
The two standard deviation predictive intervals from \ac{SPRE} are displayed as a function of the scaling factor $h$ controlling the design set $X_n^h$.  
The intervals are displayed relative to the true value $f(\mathbf{0})$, which therefore corresponds to the dashed red line at zero; insets display the seven smallest values of $h$ on an expanded scale.
Top left to bottom right: white noise, Mat\'{e}rn-$\frac{1}{2}$, Mat\'{e}rn-$\frac{3}{2}$ and Gaussian kernels.
\alttext{A figure containing 4 panels, each of which shows the predictive error bars produced by SPRE as a function of the global scaling parameter $h$.
Each panel contains results for a different type of kernel.}
}

\label{fig:flock-errorbars}
\end{figure}

\begin{figure}[h!]
    \centering
    \includegraphics[width=\textwidth]{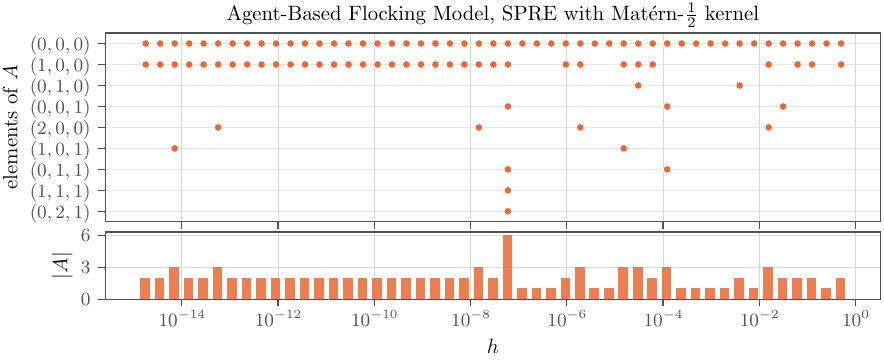}
    \caption{\emph{Agent-Based Flocking Model}.  
    The index set $A$ was estimated based on a design of size $n = 2^d$, for a range of values of the scaling factor $h$ controlling the design set $X_n^h$.
    Here the Mat\'{e}rn-$\frac{1}{2}$ kernel was used.
    The lower panel displays the cardinality $|A|$ of the estimated index set.
    \alttext{A figure displaying the estimated index set $A$ as the global scaling parameter $h$ is varied.}
    }
    \label{fig: flocking sparsity matern1/2}
\end{figure}

\begin{figure}[h!]
    \centering
    \includegraphics[width=\textwidth]{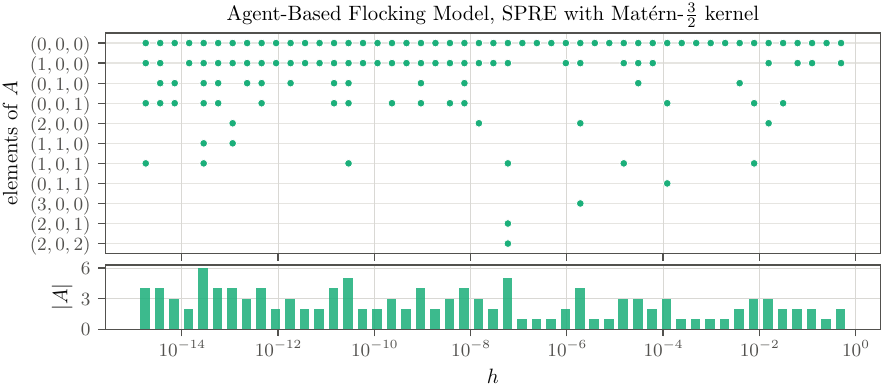}
    \caption{\emph{Agent-Based Flocking Model}.  
    The index set $A$ was estimated based on a design of size $n = 2^d$, for a range of values of the scaling factor $h$ controlling the design set $X_n^h$.
    Here the Mat\'{e}rn-$\frac{3}{2}$ kernel was used.
    The lower panel displays the cardinality $|A|$ of the estimated index set.
    \alttext{A figure displaying the estimated index set $A$ as the global scaling parameter $h$ is varied.}
    }
    \label{fig: flocking sparsity matern3/2}
\end{figure}

\begin{figure}[h!]
    \centering
    \includegraphics[width=\textwidth]{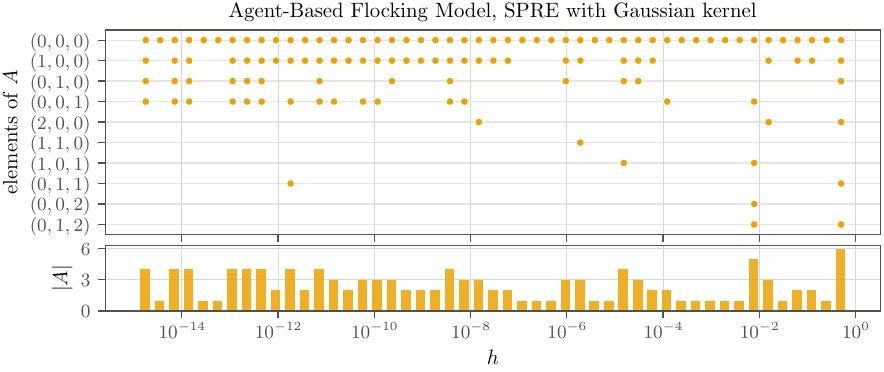}
    \caption{\emph{Agent-Based Flocking Model}.  
    The index set $A$ was estimated based on a design of size $n = 2^d$, for a range of values of the scaling factor $h$ controlling the design set $X_n^h$.
    Here the Gaussian kernel was used.
    The lower panel displays the cardinality $|A|$ of the estimated index set.
    \alttext{A figure displaying the estimated index set $A$ as the global scaling parameter $h$ is varied.}
    }
    \label{fig: flocking sparsity gaussian}
\end{figure}

\end{document}